\documentclass[11pt]{article}
\usepackage[margin=1in]{geometry}
\usepackage{microtype}
\usepackage{graphicx}
\usepackage{subfigure}
\usepackage{booktabs} 
\usepackage[hidelinks]{hyperref}
\usepackage{algorithm}
\usepackage{algpseudocode,float}
	\algtext*{EndIf}
	\makeatletter
	\newenvironment{breakablealgorithm}
	  {
	   \begin{center}
	     \refstepcounter{algorithm}
	     \hrule height.8pt depth0pt \kern2pt
	     \renewcommand{\caption}[2][\relax]{
	       {\raggedright\textbf{\ALG@name~\thealgorithm} ##2\par}%
	       \ifx\relax##1\relax 
	         \addcontentsline{loa}{algorithm}{\protect\numberline{\thealgorithm}##2}%
	       \else 
	         \addcontentsline{loa}{algorithm}{\protect\numberline{\thealgorithm}##1}%
	       \fi
	       \kern2pt\hrule\kern2pt
	     }
	  }{
	     \kern2pt\hrule\relax
	   \end{center}
	  }
	\makeatother
\usepackage[format=plain,labelfont={bf,it},textfont=it]{caption}
\usepackage[numbers]{natbib}
\usepackage{amsmath}
	
	\DeclareMathOperator*{\argmin}{argmin}
	\usepackage{amssymb}
\usepackage{mathtools}
\usepackage{amsthm}
\usepackage[capitalize,noabbrev]{cleveref}
\usepackage{autonum}

\theoremstyle{plain}
\newtheorem{theorem}{Theorem}[section]
\newtheorem{proposition}[theorem]{Proposition}
\newtheorem{lemma}[theorem]{Lemma}
\newtheorem{corollary}[theorem]{Corollary}
\theoremstyle{definition}
\newtheorem{definition}[theorem]{Definition}

\theoremstyle{remark}

\usepackage{stmaryrd}
\usepackage{bbm}

\usepackage{marginnote}

\newcommand{\cp}{X}
\newcommand{\ccp}{\textnormal{SP3}}
\newcommand{\csp}[2]{$\textnormal{SP3}_{#1,#2}$}

\allowdisplaybreaks

\usepackage{pgfplotstable}
	\pgfplotsset{compat=newest}
	\pgfplotsset{every axis/.style={
			error bars/y dir=both,
			error bars/y explicit,
			error bars/error bar style={line width=0.5pt},
			error bars/error mark options={rotate=90, mark size=0.4ex, line width=0.5pt},
		}	
	}
\usetikzlibrary{calc}

\definecolor{primary}{rgb}{0,0,0}
\colorlet{secondary}{blue!80!black}
\colorlet{tertiary}{green!80!black}
\colorlet{deep-red}{red!80!black}

\newenvironment{delayedproof}[1]
    {\begin{proof}[\phantomsection\label{proof:#1}\textbf{Proof of \Cref{#1}}]}
    {\end{proof}}

\newenvironment{cpfc}[1][]
{\begin{trivlist} \item[] {\em Proof #1.}}
{$\hfill\diamond$ \end{trivlist}}

\RequirePackage{ifthen}
\newboolean{proofs}
\setboolean{proofs}{true}

\begin{document}
	
\title{\bf Partial Optimality in Cubic Correlation Clustering}

\author{David Stein
\thanks{
             \url{david.stein1@tu-dresden.de}.
             }
\and
Silvia Di~Gregorio
\thanks{
             \url{silvia.di_gregorio@tu-dresden.de}.
             }
\and
Bjoern Andres
\thanks{
			 \url{bjoern.andres@tu-dresden.de}
			 }
             }
         
\date{TU Dresden}

\maketitle

\begin{abstract}
The higher-order correlation clustering problem is an expressive model, and recently, local search heuristics have been proposed for several applications. Certifying optimality, however, is \textsc{np}-hard and practically hampered already by the complexity of the problem statement. Here, we focus on establishing partial optimality conditions for the special case of complete graphs and cubic objective functions. In addition, we define and implement algorithms for testing these conditions and examine their effect numerically, on two datasets.
\end{abstract}


\section{Introduction}

We study an optimization problem whose feasible solutions are all partitions of a finite set $S$.
Given a cost $c_p \in \mathbb{R}$ for every (unordered) pair 
$p \in \tbinom{S}{2}$ and a cost $c_t \in \mathbb{R}$ for every (unordered) triple
$t \in \tbinom{S}{3}$, the objective is to find a partition $\Pi$ of $S$ so as to minimize the sum of the costs of those pairs and triples whose elements all belong to the same set in $\Pi$:

\begin{definition}\label{def: first def}
The instance of the \emph{cubic set partition} problem with respect to 
a finite set $S$, the set $P_S$ of all partitions of $S$, and a function $c \colon \tbinom{S}{3} \cup \tbinom{S}{2} \cup \{\emptyset\} \to \mathbb{R}$ is:
\begin{align}
\min_{\Pi \in P_S} \quad
\sum_{R \in \Pi} \sum_{t \in \tbinom{R}{3}} c_t
+ \sum_{R \in \Pi} \sum_{p \in \tbinom{R}{2}} c_p
+ c_\emptyset
\label{eq:problem}
\end{align}
\end{definition}

The cubic set partition problem is \textsc{np}-hard, as it  generalizes the \textsc{np}-hard clique partitioning problem for complete graphs \cite{goetschel-1989}, specializing to the latter in the case that $c_t = 0$ for all $t \in \tbinom{S}{3}$.
Applications of cubic set partitioning include the tasks of fitting equilateral triangles to points in a plane (\Cref{sec:experiments-geometric}), and subspace clustering as discussed in \cite{LevKarAndKeu22}.

In this article, we ask whether we can compute a partial solution to the problem efficiently, i.e.~to decide efficiently for some pairs or triples whether their elements are in the same set or distinct sets of an optimal partition.
In order to find such \emph{partial optimality}, we characterize \emph{improving maps} and state efficiently verifiable sufficient conditions of their improvingness, 
a technique introduced by \citet{shekhovtsov-2013}; see also \cite{shekhovtsov-2014,shekhovtsov-2015}. 
In order to examine the effectiveness of these partial optimality conditions numerically, we implement algorithms for testing these, and conduct experiments, cf.~\Cref{figure:experiments}.
\ifthenelse{\boolean{proofs}}{}{For conciseness, all proofs are deferred to the appendix.}

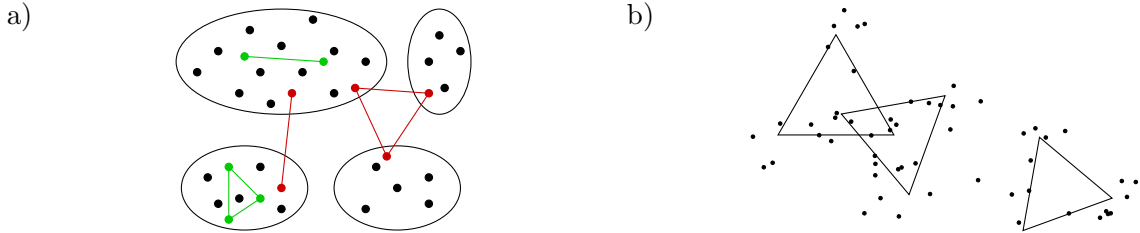
\begin{figure}[t]
	\centering
	\begin{minipage}[t]{0.5\linewidth}
		\centering
		\hspace{-\linewidth}
		\vspace{-2ex}
		a)
		\linebreak
		\begin{tikzpicture}[scale=1.4,rotate=90]
			\label{fig:illustration-partition}
			\draw (1.6, 0.3) ellipse (0.5 and 1);
			\draw (1.6, -1.2) ellipse (0.5 and 0.3);
			\draw (0.4, -0.8) ellipse (0.4 and 0.6);
			\draw (0.4, 0.65) ellipse (0.4 and 0.6);
			
			\filldraw[tertiary] (1.65, 0.65) circle (1pt);
			\filldraw[tertiary] (1.6, -0.1) circle (1pt);
			
			\draw[tertiary] (1.65, 0.65) -- (1.6, -0.1);
			
			\filldraw[tertiary] (0.3, 0.5) circle (1pt);
			\filldraw[tertiary] (0.6, 0.8) circle (1pt);
			\filldraw[tertiary] (0.1, 0.8) circle (1pt);			
			\draw[tertiary] (0.3, 0.5) -- (0.6, 0.8) -- (0.1, 0.8) -- cycle;
			%
			\filldraw[deep-red] (1.3, 0.2) circle (1pt);
			\filldraw[deep-red] (0.4, 0.3) circle (1pt);
			\draw[deep-red] (1.3, 0.2) -- (0.4, 0.3);
			
			\filldraw[deep-red] (1.35, -0.4) circle (1pt);
			\filldraw[deep-red] (0.7, -0.7) circle (1pt);
			\filldraw[deep-red] (1.3, -1.1) circle (1pt);
			\draw[deep-red] (1.35, -0.4) -- (0.7, -0.7) -- (1.3, -1.1) -- cycle;
			
			\filldraw (0.5, 1) circle (1pt);
			\filldraw (0.6, 0.5) circle (1pt);
			\filldraw (0.2, 0.3) circle (1pt);
			\filldraw (0.3, 0.7) circle (1pt);
			\filldraw (0.25, 0.9) circle (1pt);
			
			\filldraw (1.5, 1.1) circle (1pt);
			\filldraw (1.75, 0.3) circle (1pt);
			\filldraw (1.7, 0.9) circle (1pt);
			\filldraw (2.0, 0) circle (1pt);
			\filldraw (1.5, 0.1) circle (1pt);
			\filldraw (1.6, -0.5) circle (1pt);
			\filldraw (1.7, -0.2) circle (1pt);
			\filldraw (1.9, 0.6) circle (1pt);
			\filldraw (1.2, 0.4) circle (1pt);
			\filldraw (1.3, 0.7) circle (1pt);
			\filldraw (1.5, 0.5) circle (1pt);
			\filldraw (1.3, -0.2) circle (1pt);
			
			\filldraw (0.5, -1.1) circle (1pt);
			\filldraw (0.6, -0.6) circle (1pt);
			\filldraw (0.2, -0.5) circle (1pt);
			\filldraw (0.4, -0.8) circle (1pt);
			\filldraw (0.25, -1.1) circle (1pt);
			
			\filldraw (1.7, -1.4) circle (1pt);
			\filldraw (1.85, -1.2) circle (1pt);
			\filldraw (1.6, -1.1) circle (1pt);
			\filldraw (1.35, -1.25) circle (1pt);
		\end{tikzpicture}%
	\end{minipage}%
	\begin{minipage}[t]{0.5\linewidth}
		\centering
		\hspace{-\linewidth}
		\vspace{-2ex}
		b)
		\linebreak
	  \begin{tikzpicture}[xscale=1.4, yscale=-1.4,rotate=70]
	        	\label{fig:illustration-equilateral}
			\draw[black] (0.110281, 1.62544) -- (-0.596785, 0.782788) -- (0.486504, 0.591775) -- cycle;
			\draw[black] (0.302535, 0.00746821) -- (1.06858, 0.650256) -- (0.128886, 0.992276) -- cycle;
			\draw[black] (1.75981, -0.25) -- (0.980385, -0.7) -- (1.75981, -1.15) -- cycle;
			\draw plot[mark=*, mark size=0.1ex, mark options={draw=black, fill=black}] coordinates {(0.124757, 1.29645)};
			\draw plot[mark=*, mark size=0.1ex, mark options={draw=black, fill=black}] coordinates {(0.0156063, 1.56903)};
			\draw plot[mark=*, mark size=0.1ex, mark options={draw=black, fill=black}] coordinates {(0.0442783, 1.85471)};
			\draw plot[mark=*, mark size=0.1ex, mark options={draw=black, fill=black}] coordinates {(0.352589, 1.85776)};
			\draw plot[mark=*, mark size=0.1ex, mark options={draw=black, fill=black}] coordinates {(0.341288, 1.76396)};
			\draw plot[mark=*, mark size=0.1ex, mark options={draw=black, fill=black}] coordinates {(-0.788512, 0.629457)};
			\draw plot[mark=*, mark size=0.1ex, mark options={draw=black, fill=black}] coordinates {(-0.512841, 0.893121)};
			\draw plot[mark=*, mark size=0.1ex, mark options={draw=black, fill=black}] coordinates {(-0.681334, 0.771977)};
			\draw plot[mark=*, mark size=0.1ex, mark options={draw=black, fill=black}] coordinates {(-0.599445, 0.491229)};
			\draw plot[mark=*, mark size=0.1ex, mark options={draw=black, fill=black}] coordinates {(-0.214179, 0.741795)};
			\draw plot[mark=*, mark size=0.1ex, mark options={draw=black, fill=black}] coordinates {(-0.600574, 0.563377)};
			\draw plot[mark=*, mark size=0.1ex, mark options={draw=black, fill=black}] coordinates {(0.684411, 0.85045)};
			\draw plot[mark=*, mark size=0.1ex, mark options={draw=black, fill=black}] coordinates {(0.431443, 0.609925)};
			\draw plot[mark=*, mark size=0.1ex, mark options={draw=black, fill=black}] coordinates {(0.337814, 0.569987)};
			\draw plot[mark=*, mark size=0.1ex, mark options={draw=black, fill=black}] coordinates {(0.402965, 0.536616)};
			\draw plot[mark=*, mark size=0.1ex, mark options={draw=black, fill=black}] coordinates {(0.807089, 0.664689)};
			\draw plot[mark=*, mark size=0.1ex, mark options={draw=black, fill=black}] coordinates {(0.611389, 0.82912)};
			\draw plot[mark=*, mark size=0.1ex, mark options={draw=black, fill=black}] coordinates {(0.244796, 0.331031)};
			\draw plot[mark=*, mark size=0.1ex, mark options={draw=black, fill=black}] coordinates {(0.422437, -0.0280672)};
			\draw plot[mark=*, mark size=0.1ex, mark options={draw=black, fill=black}] coordinates {(0.371178, 0.0800537)};
			\draw plot[mark=*, mark size=0.1ex, mark options={draw=black, fill=black}] coordinates {(0.242193, -0.101483)};
			\draw plot[mark=*, mark size=0.1ex, mark options={draw=black, fill=black}] coordinates {(0.621105, 0.0697063)};
			\draw plot[mark=*, mark size=0.1ex, mark options={draw=black, fill=black}] coordinates {(0.468242, -0.292406)};
			\draw plot[mark=*, mark size=0.1ex, mark options={draw=black, fill=black}] coordinates {(0.315929, 0.175373)};
			\draw plot[mark=*, mark size=0.1ex, mark options={draw=black, fill=black}] coordinates {(0.810572, 0.490854)};
			\draw plot[mark=*, mark size=0.1ex, mark options={draw=black, fill=black}] coordinates {(0.801826, 0.628778)};
			\draw plot[mark=*, mark size=0.1ex, mark options={draw=black, fill=black}] coordinates {(0.786745, 0.888599)};
			\draw plot[mark=*, mark size=0.1ex, mark options={draw=black, fill=black}] coordinates {(1.09038, 1.10942)};
			\draw plot[mark=*, mark size=0.1ex, mark options={draw=black, fill=black}] coordinates {(1.13157, 0.460952)};
			\draw plot[mark=*, mark size=0.1ex, mark options={draw=black, fill=black}] coordinates {(1.23018, 0.811745)};
			\draw plot[mark=*, mark size=0.1ex, mark options={draw=black, fill=black}] coordinates {(1.00766, 0.916428)};
			\draw plot[mark=*, mark size=0.1ex, mark options={draw=black, fill=black}] coordinates {(0.228776, 0.929358)};
			\draw plot[mark=*, mark size=0.1ex, mark options={draw=black, fill=black}] coordinates {(0.245888, 1.25887)};
			\draw plot[mark=*, mark size=0.1ex, mark options={draw=black, fill=black}] coordinates {(0.105466, 1.03045)};
			\draw plot[mark=*, mark size=0.1ex, mark options={draw=black, fill=black}] coordinates {(0.142065, 1.05879)};
			\draw plot[mark=*, mark size=0.1ex, mark options={draw=black, fill=black}] coordinates {(0.33803, 1.17063)};
			\draw plot[mark=*, mark size=0.1ex, mark options={draw=black, fill=black}] coordinates {(0.41758, 0.794262)};
			\draw plot[mark=*, mark size=0.1ex, mark options={draw=black, fill=black}] coordinates {(1.75563, -0.715098)};
			\draw plot[mark=*, mark size=0.1ex, mark options={draw=black, fill=black}] coordinates {(1.42387, -0.433054)};
			\draw plot[mark=*, mark size=0.1ex, mark options={draw=black, fill=black}] coordinates {(1.67543, -0.238974)};
			\draw plot[mark=*, mark size=0.1ex, mark options={draw=black, fill=black}] coordinates {(1.16393, -0.0127456)};
			\draw plot[mark=*, mark size=0.1ex, mark options={draw=black, fill=black}] coordinates {(1.40133, -0.267431)};
			\draw plot[mark=*, mark size=0.1ex, mark options={draw=black, fill=black}] coordinates {(1.1587, -0.546172)};
			\draw plot[mark=*, mark size=0.1ex, mark options={draw=black, fill=black}] coordinates {(1.00785, -0.949736)};
			\draw plot[mark=*, mark size=0.1ex, mark options={draw=black, fill=black}] coordinates {(0.89444, -0.571423)};
			\draw plot[mark=*, mark size=0.1ex, mark options={draw=black, fill=black}] coordinates {(0.886873, -0.692756)};
			\draw plot[mark=*, mark size=0.1ex, mark options={draw=black, fill=black}] coordinates {(1.00278, -0.755562)};
			\draw plot[mark=*, mark size=0.1ex, mark options={draw=black, fill=black}] coordinates {(1.79729, -1.28684)};
			\draw plot[mark=*, mark size=0.1ex, mark options={draw=black, fill=black}] coordinates {(1.63846, -1.29231)};
			\draw plot[mark=*, mark size=0.1ex, mark options={draw=black, fill=black}] coordinates {(1.88776, -1.06207)};
			\draw plot[mark=*, mark size=0.1ex, mark options={draw=black, fill=black}] coordinates {(1.69246, -1.41393)};
			\draw plot[mark=*, mark size=0.1ex, mark options={draw=black, fill=black}] coordinates {(1.89045, -1.0887)};
			\draw plot[mark=*, mark size=0.1ex, mark options={draw=black, fill=black}] coordinates {(1.89538, -0.978857)};
			\draw plot[mark=*, mark size=0.1ex, mark options={draw=black, fill=black}] coordinates {(1.85791, -1.05489)};
		\end{tikzpicture}
	\end{minipage}
	\caption{In order to examine the effectiveness of partial optimality conditions numerically, we implement algorithms for testing these conditions and measure the fraction of fixed variables with respect to a parameter controlling the noise of the problem, for (a) synthetic instances with four clusters and noisy costs, and (b) instances for the task of finding equilateral triangles in a noisy point cloud.}
	\label{figure:experiments}
\end{figure}


\section{Related Work}\label{sec: related work}

We choose to state the cubic set partition problem 
(\Cref{def: first def})
in the form of a non-linear binary program 
(\Cref{def: pb}), 
a special case of the higher-order correlation clustering problem introduced by \citet{kim-2014}.
Combinatorial optimization problems like this involving higher-order objective functions 
have interesting application as accurate models of intrinsically non-linear tasks
\cite{AgaLimZelPerKriBel05,KapSpeReiSch16,kim-2014,LevKarAndKeu22,OchBro12,PulChiSadSut17}. In particular, higher-order correlation clustering has been used for subspace clustering in~\cite[Section~5.1]{LevKarAndKeu22} by introducing negative costs for points sufficiently close to a subspace. 

Being able to efficiently fix some variables to an optimal value and thus reducing the size of the problem can be valuable in practice.
Consequently, much effort has been devoted to studying partial optimality for non-convex problems 
\cite{AdaLasShe98,BilSut92,HamHanSim84,KapSpeReiSch13,KohEtAl08,shekhovtsov-2014,shekhovtsov-2015}.
In particular, we mention the impressive application of partial optimality conditions to Potts models for image segmentation in which more than 95\% of the variables can be fixed~\cite[Fig.~1]{shekhovtsov-2015}.
In contrast to the customary approach of considering a convex, usually linear, relaxation and establishing partial optimality conditions regarding the variables in the extended formulation, we study such conditions directly in the original variable space.
Unlike the above-mentioned articles, we concentrate on taking advantage of the specific structure of the cubic clique partitioning problem.

To this end, we build on the works of \citet{alush-2012} and \citet{Lange-2018,Lange-2019} who establish partial optimality conditions for problems equivalent to correlation clustering with a linear objective function.
Regarding their terminology, we remark that the correlation clustering problem, the clique partitioning problem, and the multicut problem are equivalent if the objective functions are linear. 
The correlation clustering problem keeps attracting considerable attention by the community also in the context of approximation algorithms \cite{Vel22}.
The cubic set partition problem we consider here generalizes the specialization to complete graphs of both the correlation clustering problem and the multicut problem.
Note that correlation clustering for arbitrary, weighted graphs does not become more specific by considering only complete graphs. Instead, any such problem with respect to an arbitrary graph can be stated as a problem with respect to a complete graph and excessive edges having cost zero. 

Here, we transfer all partial optimality conditions established by 
\citet{alush-2012} and \citet{Lange-2018,Lange-2019} 
for the correlation clustering problem and the multicut problem to the cubic set partition problem.
In addition, we establish new results. 
Unlike in \citet{Lange-2018}, the algorithm we define does not exploit the sparsity of edges with non-zero cost and, in this sense, is designed for complete graphs. Moreover, we do not contribute persistency conditions for the max cut problem.

\section{Preliminaries}\label{section:preliminaries}

In order to establish partial optimality conditions for the cubic set partition problem (\Cref{def: first def}), we state this problem in the form of the non-linear integer program introduced by \citet{kim-2014}:

\begin{proposition}\label{def: pb}
	The instance of the cubic set partition problem with respect to a finite set $S$ and a function $c \colon \tbinom S3 \cup \tbinom S2 \cup \{\emptyset\} \to \mathbb{R}$ has the form of the cubic integer program
		\begin{align}\label{eq: pb}
			\min_{x: \tbinom S2 \to \{0, 1\}}
			& \sum_{pqr \in \tbinom S3} c_{pqr}x_{pq}x_{pr}x_{qr} + \sum_{pq \in \tbinom S2}c_{pq}x_{pq} + c_\emptyset 
			\\
			\mathrm{subject~to}\ \ 
			& \forall p \in S\ \forall q \in S \setminus \{p\}\ \forall r \in S \setminus \{p,q\} \colon  \quad x_{pq} + x_{qr} - x_{pr} \leq 1
			\label{eq:def-ccp}
		\end{align}
\end{proposition}

\ifthenelse{\boolean{proofs}}{
\begin{proof}
For each partition $\Pi$ of the set $S$ and every distinct $p, q \in S$, let $x_{pq} = 1$ if and only if $p$ and $q$ are in the same set of $\Pi$.
This establishes a one-to-one relation between the set $P_S$ of all partitions of $S$ and the feasible set $X_S$ of all $x: \tbinom S2 \to \{0, 1\}$ that satisfy the above inequalities \citep{goetschel-1989}.
Under this bijection, the objective functions of \Cref{def: first def} and \Cref{def: pb}
are equivalent.
\end{proof}
}{}

Below, we let \csp{S}{c} denote this instance of the problem, $\phi_c$ its objective function, and $X_S$ its feasible set, i.e.~the set of all $x: \tbinom S2 \to \{0, 1\}$ that satisfy the above inequalities.

\medskip

\label{section:improving-maps}
Our main technique is the construction of improving maps \citep{shekhovtsov-2013}, which is based on the following preliminary notions.
\begin{definition}
	Let $X \neq \emptyset$, $\phi\colon X \to \mathbb{R}$ and $\sigma \colon X \to X$. If for every $x \in X$, we have $\phi(\sigma(x)) \leq \phi(x)$, then
	$\sigma$ is called \emph{improving} for the problem $\min_{x\in X}\phi(x)$.
\end{definition}
\begin{proposition}
	\label{lemma:persistency-predicate}
	Let $X \neq \emptyset$, $\phi \colon X \to \mathbb{R}$ and $\sigma \colon X \to X$ an improving map. Moreover, let $Q\subseteq X$. 
	If, for every $x \in X$, $\sigma(x) \in Q$,
	then there is an optimal solution $x^*$ to $\min_{x\in X}\phi(x)$ such that $x^*\in Q$.
\end{proposition}
\ifthenelse{\boolean{proofs}}{
\begin{proof}
	Let $x^*$ be an optimal solution to $\min_{x\in X}\phi(x)$ such that $x^*\not\in Q$. Then $\sigma(x^*)$ is also an optimal solution to $\min_{x\in X}\phi(x)$ and $\sigma(x^*)\in Q$.
\end{proof}
}{}
\begin{corollary}
	\label{lemma:persistency-variable}
	Let $S \neq \emptyset$, $X \subseteq \{0, 1\}^S$, $\phi \colon X \to \mathbb{R}$ and $\sigma \colon X \to X$ an improving map. 
	Moreover, let $s \in S$ and $\beta \in \{0, 1\}$. 
	If for every $x \in X$, $\sigma(x)_s = \beta$,
	then there is an optimal solution $x^*$ to $\min_{x\in X}\phi(x)$ such that $x^*_s = \beta$.
\end{corollary}
Our construction starts from the elementary maps of \citet{Lange-2019}, i.e.~the map $\sigma_{\delta(R)}$ that cuts a set $R \subseteq S$ from its complement, and the map $\sigma_R$ that joins all sets intersecting with a set $R \subseteq S$:
\begin{definition}
	For any finite, non-empty set $S$ and $R \subseteq S$,
	the \emph{elementary cut map} $\sigma_{\delta(R)}\colon \cp_S \to \cp_S$ is such that for all $x\in \cp_S $ and all $pq\in \tbinom S2$:
	\begin{equation}
		\sigma_{\delta(R)}(x)_{pq} := \begin{cases}
			0 & \textnormal{if $| \{p, q\} \cap R| = 1$} \\
			x_{pq} & \textnormal{otherwise}
		\end{cases}
		\enspace .
	\end{equation}
\end{definition}

\begin{definition}
	For any finite, non-empty set $S$ and $R \subseteq S$,
	the \emph{elementary join map} $\sigma_R \colon \cp_S \to \cp_S$ is such that for all $x \in \cp_S$ and all $pq\in \tbinom S2$:
	\begin{equation}
		\sigma_R(x)_{pq} := \begin{cases}
			1 & \textnormal{if $pq\in \tbinom R2$} \\
			1 & \textnormal{if $\forall p'\in \{p, q\}\setminus R \; \exists q'\in R \colon x_{p'q'} = 1$} \\
			x_{pq} & \textnormal{otherwise}
		\end{cases}
		\enspace .
	\end{equation}
\end{definition}

\section{Partial Optimality Conditions}
\label{section:partial-optimality-criteria}

In this section, we establish partial optimality conditions for the cubic set partition problem by constructing improving maps, starting from the elementary maps $\sigma_{\delta(R)}$ and $\sigma_R$ defined in \Cref{section:preliminaries}.

For simplicity, we introduce some notation:
For any $r\in \mathbb{R}$, let $r^\pm := \max\{0, \pm r\}$.
For any function $f\colon X \to Y$ and any $X'\subseteq X$, let $f\vert_{X'}\colon X' \to Y$ denote the restriction of $f$ to $X'$.
From here onwards, $S$ will always denote a finite set.
For any non-empty set $S$ and any $R, R', R'' \subseteq S$, let
\begin{align}
	\delta(R, R') & := \left\{pq \in \tbinom S2 \;\middle|\; p \in R \land q \in R' \right\} 
	\\	
	\delta(R) & := \delta(R, S \setminus R) \\
	T_{RR'R''} & := \left\{pqr\in \tbinom S3 \;\middle|\; p\in R\land q \in R'\land r\in R'' \right\}
	\enspace .
\end{align}
For $\mathcal{I}_S := \tbinom{S}{3} \cup \tbinom{S}{2} \cup \{\emptyset\}$ and any $c \colon \mathcal{I}_S \to \mathbb{R}$, let
\begin{align}
	P^\pm & := \left\{ pq\in \tbinom S2 \;\middle|\; c_{pq} \gtrless 0 \right\} \\
	T^\pm & := \left\{ pqr\in \tbinom S3 \;\middle|\; c_{pqr} \gtrless 0 \right\} \\
	T_{P'} & := \left\{ pqr\in \tbinom S3 \;\middle|\; P' \cap \tbinom{pqr}{2} \neq \emptyset \right\} \quad \forall P' \subseteq \tbinom S2
	\enspace .	
\end{align}

\subsection{Cut Conditions}
\label{section:partial-optimality-criteria-cuts}
Here, we establish partial optimality conditions that imply the existence of an optimal solution $x^*$ to $\ccp_{S, c}$ such that $x^*_{ij} = 0$ for some $ij \in \tbinom S2$ or $x^*\in \{x\in \cp_S \mid x_{ij}x_{ik}x_{jk} = 0\}$ for some $ijk \in \tbinom S3$.

The following Proposition~\ref{lemma:persistency-subset-separation} generalizes to cubic objective functions the specialization for complete graphs of Theorem~1 of \citet{alush-2012}. 
Intuitively, it says that if there exists a subset $R$ for which joining any pair or triple that has some items in $R$ and some outside of $R$ leads to a penalty, then we can safely cut the whole set $R$ from the rest. 
\begin{proposition}
	\label{lemma:persistency-subset-separation}	
	Let $S \neq \emptyset$, and let $c \in \mathbb{R}^{\mathcal{I}_S}$. 
	If there exists $R \subseteq S$ such that
	\begin{align}
		\label{eq:edge-cut-condition-1}
		c_{pq} &\geq 0 \quad \forall pq\in \delta(R) \\
		\label{eq:edge-cut-condition-2}
		c_{pqr} &\geq 0 \quad \forall pqr\in T_{\delta(R)}
	\end{align}
	then there is an optimal solution $x^*$ to \csp{S}{c} such that $x^*_{ij} = 0$ for all $ij \in \delta(R)$.
\end{proposition}

\ifthenelse{\boolean{proofs}}{
\begin{proof}
We define $\sigma \colon \cp_S \to \cp_S$ such that for all $x \in \cp_S$ we have 
\begin{equation}
\sigma(x) := \begin{cases}
x & \text{ if } x_{ij} = 0 \:\:	 \forall ij \in \delta(R) \\
\sigma_{\delta(R)}(x) & \textnormal{otherwise}
\end{cases}
\enspace .
\end{equation}
For any $x\in \cp_S$, let $x' = \sigma(x)$. 
Firstly, the map $\sigma$ is such that $x'_{ij} = 0$ 
for all $ij \in \delta(R)$. 
Secondly, for any $x\in \cp_S$ such that there exists $ij\in \delta(R)$ such that $x_{ij} = 1$, we have 
\begin{align}
\phi_c(x') - \phi_c(x) & = - \sum_{pqr \in T_{\delta(R)}} c_{pqr}x_{pq}x_{pr}x_{qr}- \sum_{pq\in \delta(R)} c_{pq}x_{pq} \\
& \leq - \sum_{pqr\in T_{\delta(R)}\cap T^-} c_{pqr} - \sum_{pq\in \delta(R)\cap P^-}c_{pq} \\
& = 0
\enspace .
\end{align}
The last equality is due to the fact that those sums vanish by Assumptions~\eqref{eq:edge-cut-condition-1} and \eqref{eq:edge-cut-condition-2}.
Applying Corollary~\ref{lemma:persistency-variable} concludes the proof.
\end{proof}
}{
In the proof, we show that the map that keeps vectors $x\in \cp_S$ with $x^*_{ij} = 0$ for all $ij \in \delta(R)$ as they are and applies $\sigma_{\delta(R)}$ to the others is improving if \eqref{eq:edge-cut-condition-1} and \eqref{eq:edge-cut-condition-2} are satisfied. 
}

This condition can be exploited: When satisfied for a set $R$, \csp{S}{c} decomposes into two independent subproblems. Firstly,
\begin{equation}
	\min_{x\in \cp_S} \phi_c(x) 
	= 
	\min_{x\in \cp_{R}} \phi_{c\vert_{\mathcal{I}_R}}(x) 
	+ 
	\min_{x\in \cp_{S\setminus R}} \phi_{c\vert_{\mathcal{I}_{S\setminus R}}}(x)
	\enspace .
\end{equation}
Secondly, given solutions 
\begin{align}
x' & \in \argmin_{x\in \cp_{R}}\phi_{c\vert_{\mathcal{I}_R}}(x) \\
x'' & \in \argmin_{x\in \cp_{S\setminus R}}\phi_{c\vert_{\mathcal{I}_{S \setminus R}}}(x)
\enspace ,
\end{align}
an optimal solution to the problem $\min_{x\in \cp_S} \phi_c(x)$ is given by the $x\in \cp_S$ such that 
\begin{align}
x_{pq} = \begin{cases}
	x'_{pq} & \text{if\ } pq\in \tbinom R2 \\
	x''_{pq} & \text{if\ } pq\in \tbinom{S\setminus R}{2} \\
	0 & \text{if\ } pq\in \delta(R)
\end{cases}
\enspace .
\end{align}

The following \Cref{proposition:edge-cut-persistency}, together with \Cref{lemma:edge-join-persistency} further below, generalize to cubic objective functions the specialization for complete graphs of Theorem~1 of \citet{Lange-2019}. 
The idea behind this statement is the following: if there exists a pair $ij$ and a subset $R$ that cuts $ij$ such that the penalty that we would have to pay if we were to join $i$ and $j$ is so large that it is at least the best possible reward achieved by joining $R$ and its complement, then it is best to keep $i$ and $j$ separated.  
%
%
\begin{proposition}
	\label{proposition:edge-cut-persistency}
	Let $S \neq \emptyset$, 
	and let $c\in \mathbb{R}^{\mathcal{I}_S}$. Moreover, let $ij\in \tbinom S2$. 
	If there exists $R \subseteq S$ such that $ij\in \delta(R)$ and 
	\begin{equation}
		\label{eq:assumption-edge-cut-inequality}
		c_{ij}^+ \geq \sum_{pqr\in T_{\delta(R)}} c_{pqr}^- + \sum_{pq\in \delta(R)} c_{pq}^-
		\enspace ,
	\end{equation}
	then there is an optimal solution $x^*$ to \csp{S}{c} such that $x^*_{ij} = 0$.
\end{proposition}

\ifthenelse{\boolean{proofs}}{
\begin{proof}
Let  
$\sigma \colon \cp_S \to \cp_S$ be constructed as 
\begin{equation}
\sigma(x) := \begin{cases}
x & \textnormal{if $x_{ij} = 0$} \\
\sigma_{\delta(R)}(x) & \textnormal{otherwise}
\end{cases}
\enspace .
\end{equation}
For any $x\in \cp_S$, let $x' = \sigma(x)$. 
First of all, the map $\sigma$ is such that $x'_{ij} = 0$ for all $x\in \cp_S$. 
Next, for any $x\in \cp_S$ such that $x_{ij} = 1$, we have 
\begin{align}
\phi_c(x') - \phi_c(x) 
& = -c_{ij} -\sum_{pqr\in T_{\delta(R)}} c_{pqr}x_{pq} x_{pr}x_{qr} - \sum_{\substack{pq\in \delta(R) \\ pq \neq ij}}c_{pq}x_{pq} \\
& \leq -c_{ij} + \sum_{pqr\in T_{\delta(R)}}c_{pqr}^- + \sum_{\substack{pq \in \delta(R)\\ pq \neq ij}}c_{pq}^- \\
& = -c_{ij}^+ + \sum_{pqr\in T_{\delta(R)}}c_{pqr}^- + \sum_{pq \in \delta(R)}c_{pq}^- \\
& \leq 0
\enspace .
\end{align}
The last inequality follows from Assumption~\eqref{eq:assumption-edge-cut-inequality}.
We conclude the proof by applying Corollary~\ref{lemma:persistency-variable}. 
\end{proof}
}{
In the proof, we apply $\sigma_{\delta(R)}$ to any $x\in \cp_S$ with $x_{ij} = 1$, on the one hand, and apply the identity to the feasible vectors with $x_{ij} = 0$, on the other hand. We prove that this is an improving map under the assumption. 
}

The following Proposition~\ref{lemma:persistency-triplet-cut} establishes a partial optimality result that implies the existence of an optimal solution $x^*$ such that $x^*_{ij}x^*_{ik}x^*_{jk} = 0$ for some $ijk\in \tbinom S3$. 
The intuition is similar as for \Cref{proposition:edge-cut-persistency}, except here it involves a triple instead of a pair. 
Note, however, that one cannot conclude which of the variables $x^*_{ij}$, $x^*_{ik}$ or $x^*_{jk}$ equals zero.
%

%
\begin{proposition}
	\label{lemma:persistency-triplet-cut}
	Let $S \neq \emptyset$, 
	and let $c\in \mathbb{R}^{\mathcal{I}_S}$. 
	Moreover, let $ijk\in \tbinom S3$ and $R \subseteq S$ such that $ij, ik \in \delta(R)$.
	If
	\begin{align}
		c_{ijk}^+ + c_{ij}^+ + c_{ik}^+
		\label{eq:triplet-cut-condition}
		\geq &\sum_{pqr\in T_{\delta(R)}} c_{pqr}^-+ \sum_{pq\in \delta(R)} c_{pq}^-
		\enspace ,
	\end{align}
	there is an optimal solution $x^*$ to \csp{S}{c} such that $x^*_{ij}x^*_{ik}x^*_{jk} = 0$.
\end{proposition}

\ifthenelse{\boolean{proofs}}{
\begin{proof}
We define $\sigma \colon \cp_S\to \cp_S$ as 
\begin{equation}
\sigma(x) := \begin{cases}
x & \textnormal{if $x_{ij}x_{ik}x_{jk} = 0$}\\
\sigma_{\delta(R)}(x) & \textnormal{otherwise}
\end{cases}.
\end{equation}
For any $x\in \cp_S$, we denote $\sigma(x)$ by $x'$. 
Firstly, we observe that 
$x'_{ij}x'_{ik}x'_{jk} = 0$ for all $x\in \cp_S$. 
Secondly, for any $x\in \cp_S$ such that $x_{ij}x_{ik}x_{jk} = 1$, we have
\begin{align}
\phi_c(x') - \phi_c(x) & = - \sum_{pqr \in T_{\delta(R)}} c_{pqr}x_{pq}x_{pr}x_{qr} - \sum_{pq\in \delta(R)} c_{pq}x_{pq} \\
& \leq -c_{ijk} - c_{ij} - c_{ik} + \sum_{\substack{pqr\in T_{\delta(R)}\\ pqr \neq ijk}} c_{pqr}^- + \sum_{\substack{pq\in \delta(R)\\ pq\not\in \{ij, ik\}}} c_{pq}^-\\
& = -c_{ijk}^+ - c_{ij}^+ - c_{ik}^+ + \sum_{pqr\in T_{\delta(R)}} c_{pqr}^- + \sum_{pq\in \delta(R)} c_{pq}^- \\
& \leq 0
\enspace .
\end{align}
The last inequality holds because of 
Assumption~\eqref{eq:triplet-cut-condition}.
Applying Proposition~\ref{lemma:persistency-predicate} with $Q = \{x\in \cp_S \mid x_{ij}x_{ik}x_{jk} = 0\}$ concludes the proof.
\end{proof}
}{
The map used in the proof is similar to the previous ones:
If we start from a feasible vector such that $x_{ij} x_{ik}x_{jk} = 0$ already, we do not change it.
Otherwise, we apply $\sigma_{\delta(R)}$.
}
We remark that \Cref{lemma:persistency-triplet-cut}, together with its counterpart, \Cref{proposition:triplet-join}, is a novel result that does not extend prior work.

\subsection{Join Conditions}
\label{section:partial-optimality-criteria-joins}

Next, we establish partial optimality conditions that imply the existence of an optimal solution $x^*$ to \csp{S}{c} such that $x^*_{ij} = 1$ for some $ij\in \tbinom S2$. 
This property can be used to simplify a given instance by joining the elements $i$ and~$j$.

As mentioned earlier, the following Proposition~\ref{lemma:edge-join-persistency} transfers a result of \citet{Lange-2019} to the cubic set partition problem. 
Condition \eqref{eq:edge-join-inequality} is rather restrictive.
The idea behind it is the following: if there exists a subset of items $R$ that cuts $i$ and $j$, and the total potential reward for joining $i$ and $j$ (meaning not only joining the pair $ij$ but also the triples that could end up together once $i$ and $j$ are in the same cluster) is higher than the sum of rewards and penalties incurred by joining $R$ and its complement, then it is beneficial to put $i$ and $j$ together.
\begin{proposition}
	\label{lemma:edge-join-persistency}
	Let $S \neq \emptyset$, 
	and let $c \in \mathbb{R}^{\mathcal{I}_S}$. 
	Moreover, let $ij\in \tbinom S2$. 
	If there exists an $R \subseteq S$ such that $ij \in \delta(R)$ and 
	\begin{align}
		2c_{ij}^- + \sum_{pqr\in T_{\{ij\}}} c_{pqr}^- 
		\label{eq:edge-join-inequality}
		\geq  \sum_{pqr\in T_{\delta(R)}}\vert c_{pqr}\vert + \sum_{pq\in \delta(R)}\vert c_{pq}\vert
		\enspace ,
	\end{align}
	there is an optimal solution $x^*$ to \csp{S}{c} such that $x^*_{ij} = 1$.
\end{proposition}

\ifthenelse{\boolean{proofs}}{
\begin{proof}
Let $\sigma \colon \cp_S \to \cp_S$ such that for all $x \in \cp_S$, we have 
\begin{equation}
\sigma(x) := \begin{cases}
x & \textnormal{if $x_{ij} = 1$} \\
\left(\sigma_{ij}\circ \sigma_{\delta(R)}\right)(x) & \textnormal{otherwise}
\end{cases}
\enspace .
\end{equation}
For any $x\in \cp_S$, let $x' = \sigma(x)$. 
The map $\sigma$ is such that $x'_{ij} = 1$ for all $x\in \cp_S$. 
We show that $\sigma$ is improving. 
In particular, let $x \in \cp_S$ such that $x_{ij} = 0$. 
We observe that $x'_{pq} = x_{pq}$ for all $pq \not\in \delta(R)$. 
Therefore,
\begin{align}
\phi_c(x') - \phi(x) & = \sum_{pqr\in T_{\delta(R)}}c_{pqr}\left(x'_{pq}x'_{pr}x'_{qr} - x_{pq}x_{pr}x_{qr}\right) + \sum_{pq\in \delta(R)}c_{pq}\left(x'_{pq} - x_{pq}\right) \\
& = \sum_{pqr\in T_{\delta(R)}\setminus T_{\{ij\}}}c_{pqr}\left(x'_{pq}x'_{pr}x'_{qr} - x_{pq}x_{pr}x_{qr}\right) +\sum_{pqr\in T_{\{ij\}}}c_{pqr}x'_{pq}x'_{pr}x'_{qr}  + c_{ij} \\
& \qquad + \sum_{\substack{pq\in \delta(R)\\ pq \neq ij}}c_{pq}\left(x'_{pq} - x_{pq}\right)\\
& \leq \; \sum_{pqr\in T_{\delta(R)}\setminus T_{\{ij\}}}\vert c_{pqr}\vert + \sum_{pqr\in T_{\{ij\}}}c_{pqr}^+ + c_{ij} + \sum_{\substack{pq\in \delta(R)\\ pq \neq ij}}\vert c_{pq}\vert \\
& = \sum_{pqr\in T_{\delta(R)}}\vert c_{pqr}\vert - \sum_{pqr\in T_{\{ij\}}}c_{pqr}^- - 2c_{ij}^-  + \sum_{pq\in \delta(R)}\vert c_{pq}\vert \\
& \leq  0 
\enspace .
\end{align}
The last inequality here is due to Assumption~\eqref{eq:edge-join-inequality}.
\end{proof}
}{
In the proof, we start from a feasible vector $x \in \cp_S$.
If $x$ satisfies our thesis already, we are done. 
Otherwise, we apply $\sigma_{ij} \circ \sigma_{\delta(R)}$ and show that this map is improving granted that \eqref{eq:edge-join-inequality} is satisfied.
}

While \Cref{lemma:edge-join-persistency} only considered the issue of deciding whether a pair $ij$ of items should end up together or not, \Cref{proposition:triplet-join} deals with the same question for a triple $ijk$ of points. 
The statement is a bit more involved but the intuition is in line with that of \Cref{lemma:edge-join-persistency}. 
\begin{proposition}
	\label{proposition:triplet-join}
	Let $S \neq \emptyset$, 
	and let $c\in \mathbb{R}^{\mathcal{I}_S}$. 
	Moreover, let $ijk\in \tbinom S3$ and $R \subseteq S$ such that $ij, ik \in \delta(R)$. If 
	{
		\begin{align}
			&2c_{ijk}^- + 2c_{ij}^- + 2c_{ik}^- + c_{jk}^- - \sum_{pqr\in \tbinom S3}c_{pqr}^+ - \sum_{pq\in \tbinom S2}c_{pq}^+  +  \min_{\substack{x\in \cp_{ijk} \\ x_{ij}x_{ik}x_{jk} = 0}}\sum_{pq\in \tbinom{ijk}{2}}c_{pq} x_{pq} 
			\\
			\label{eq:assumption-triplet-join}
			\geq \quad &\sum_{pqr \in T_{\delta(R)} }c_{pqr}^- +  \sum_{pq \in \delta(R)}c_{pq}^-
			\enspace ,
		\end{align}
	}%
	there is an optimal solution $x^*$ to \csp{S}{c} such that $x^*_{ij}x^*_{ik}x^*_{jk} = 1$.
\end{proposition}
\ifthenelse{\boolean{proofs}}{
\begin{proof}
We construct 
$\sigma \colon \cp_S \to \cp_S$ as follows. 
\begin{equation}
\sigma(x) := \begin{cases}
x & \textnormal{if $x_{ij}x_{ik}x_{jk} = 1$} \\
\left(\sigma_{ijk}\circ \sigma_{\delta(R)}\right)(x) & \textnormal{otherwise}
\end{cases}
\end{equation}
For any $x \in \cp_S$, let $x' = \sigma(x)$. 
The map $\sigma$ is such that $x'_{ij}x'_{ik}x'_{jk} = 1$ for all $x\in \cp_S$. 
Note that, for any $x\in \cp_S$ such that $x_{ij}x_{ik}x_{jk} = 0$, we have 
$x'_{pq} \geq x_{pq}$ for all $pq \not\in \delta(R)$. 
It follows that 
\begin{align}
\phi_c(x') - \phi(x) & = \sum_{\substack{pqr\in T_{\delta(R)} \\ pqr \neq ijk}}c_{pqr}(x'_{pq}x'_{pr}x'_{qr} - x_{pq}x_{pr}x_{qr}) + c_{ijk} + \sum_{pqr\not\in T_{\delta(R)}}c_{pqr}(x'_{pq}x'_{pr}x'_{qr} - x_{pq}x_{pr}x_{qr}) \\					
& \qquad + \sum_{pq\in \{ij, ik, jk\}}c_{pq}(1-x_{pq}) + \sum_{\substack{pq\in \delta(R)\\ pq \not\in \{ij, ik\}}}c_{pq}(x'_{pq} - x_{pq}) + \sum_{pq\not\in \delta(R)\cup \{jk\}}c_{pq}(x'_{pq} - x_{pq})\\	
& \leq \; c_{ijk} +  \max_{\substack{x\in \cp{ijk} \\ x_{ij}x_{ik}x_{jk} = 0}}\sum_{pq\in \tbinom{ijk}{2}}c_{pq} (1 - x_{pq}) + \sum_{\substack{pqr\in T_{\delta(R)} \\ pqr \neq ijk}}\vert c_{pqr}\vert  + \sum_{\substack{pqr\in T^+ \\ pqr \notin T_{\delta(R)}}} c_{pqr} \\
& \qquad + \sum_{\substack{pq\in \delta(R) \\ pq \notin \{ij, ik\}}} \vert c_{pq} \vert + \sum_{\substack{pq \in P^+ \\ pq \notin \left(\delta(R)\cup \{jk\}\right)}}c_{pq} \\	
& = \; c_{ijk} +  \max_{\substack{x\in \cp_{ijk} \\ x_{ij}x_{ik}x_{jk} = 0}}\sum_{pq\in \tbinom{ijk}{2}}c_{pq} (1 - x_{pq}) + \sum_{\substack{pqr\in T^+ \cup T_{\delta(R)} \\ pqr \neq ijk}}\vert c_{pqr}\vert + \sum_{\substack{pq\in P^+ \cup \delta(R)\\ pq \not\in \{ij, ik, jk\}}}\vert c_{pq} \vert \\
& = \; -2c_{ijk}^- -  2c_{ij}^- - 2c_{ik}^- - c_{jk}^- - \min_{\substack{x\in \cp_{ijk} \\ x_{ij}x_{ik}x_{jk} = 0}}\sum_{pq\in \tbinom{ijk}{2}}c_{pq} x_{pq}\\
& + \sum_{\substack{pqr\in T_{\delta(R)}}}c_{pqr}^- +\sum_{pqr\in \tbinom S3}c_{pqr}^+ + \sum_{pq\in \tbinom S2}c_{pq}^+ + \sum_{pq\in \delta(R)}c_{pq}^- \\
& \leq  0
\enspace .
\end{align}
Assumption \eqref{eq:assumption-triplet-join} provides the last inequality.
We arrive at the thesis by applying Proposition~\ref{lemma:persistency-predicate} with $Q = \{x\in \cp_S \mid x_{ij}x_{ik}x_{jk} = 1\}$. 
\end{proof}
}{
The map that we show to be improving in the proof follows the actions described previously, namely it is constructed as $\sigma_{ijk} \circ \sigma_{\delta(R)}$.
}

The following \Cref{lemma:persistency-triangle-edge-join} expands to the cubic set partition problem Corollary~1 of \citet{Lange-2019}.
%
It considers triplets $ijk\in \tbinom S3$ and states three requirements under which joining the pair $ik$ does not compromise optimality. 
Firstly, \eqref{eq:triangle-edge-join-1} and \eqref{eq:triangle-edge-join-2} say there is a subset $R\subseteq S$ such that the total potential reward from joining $i, j$ and $k$ is greater than or equal to the sum of rewards and penalties incurred by joining $R$ and its complement. 
Secondly, \eqref{eq:triangle-edge-join-3} states that the cost of joining the triple $ijk$ must be at most the negative of the sum of rewards incurred when joining $ijk$ and its complement. Under these assumptions, we can put $i$ and $k$ together.
%
\begin{proposition}
	\label{lemma:persistency-triangle-edge-join}
	Let $S \neq \emptyset$, 
	and let $c\in \mathbb{R}^{\mathcal{I}_S}$. 
	Moreover, let $ijk \in \tbinom S3$ and $R, R'\subseteq S$ such that $ij, ik \in \delta(R)$ and $jk, ik\in \delta(R')$. 
	If all of the following conditions hold, there exists an optimal solution $x^*$ to \csp{S}{c} such that $x^*_{ik} = 1$.
	\begin{align}
		c_{ijk}^- + 2c_{ij}^- + 2c_{ik}^- + \sum_{pqr\in T_{\{ij, ik\}}}c_{pqr}^-
		& \geq \sum_{pqr\in T_{\delta(R)}}\vert c_{pqr}\vert + \sum_{pq\in \delta(R)}\vert c_{pq} \vert  \label{eq:triangle-edge-join-1}		
		\\
		c_{ijk}^- + 2c_{jk}^- + 2c_{ik}^- + \sum_{pqr\in T_{\{jk, ik\}}}c_{pqr}^-  		
		& \geq \sum_{pqr\in T_{\delta(R')}}\vert c_{pqr}\vert + \sum_{pq\in \delta(R')}\vert c_{pq} \vert  \label{eq:triangle-edge-join-2}		
		\\
		c_{ijk} + c_{ij} + c_{ik} + c_{jk} 	
		& \leq - \sum_{\substack{pqr\in T_{\delta(ijk)}\cap T^-\\ pqr\not\in T_{\{ij, ik, jk\}}}}\vert c_{pqr}\vert
		- \sum_{pq\in \delta(ijk)\cap P^-}\vert c_{pq} \vert \label{eq:triangle-edge-join-3}
	\end{align}
\end{proposition}
\ifthenelse{\boolean{proofs}}{
\begin{proof}
Let $\sigma \colon \cp_S \to \cp_S$ be defined as 
\begin{equation}
\sigma(x) := \begin{cases}
x & \textnormal{if $x_{ik} = 1$}\\
(\sigma_{ik} \circ \sigma_{\delta(R)})(x) & \textnormal{if $x_{ik} = x_{ij} = 0 ,x_{jk} = 1$} \\
(\sigma_{ik}\circ \sigma_{\delta(R')})(x) & \textnormal{if $x_{ik} = x_{jk} = 0 , x_{ij} = 1$} \\
(\sigma_{ijk}\circ \sigma_{\delta(ijk)})(x) & \textnormal{if $x_{ik} = x_{ij} = x_{jk} = 0$}
\end{cases}
\enspace .
\end{equation}
We use the notation $x' = \sigma(x)$ for all $x \in \cp$.
Firstly, the map $\sigma$ is such that $x'_{ik} = 1$ for all $x\in \cp_S$. 
Secondly, for any $x\in \cp_S$ such that $x_{ik} = x_{ij} = 0$ and $x_{jk} = 1$, the map $\sigma_{ik}\circ \sigma_{\delta(R)}$ is such that $x'_{ij}=x'_{ik} = x'_{jk} = 1$ and $x'_{pq} = x_{pq}$ for any $pq \not\in \delta(R)$. 
We conclude:
\begin{align}
\phi_c(x') - \phi_c(x) & = c_{ijk} + c_{ij} + c_{ik} + \sum_{\substack{pqr \in T_{\{ij, ik\}}\\ pqr\neq ijk}} c_{pqr}x'_{pq}x'_{pr}x'_{qr} \\
& \qquad + \sum_{\substack{pqr\in T_{\delta(R)} \\ pqr \notin T_{\{ij, ik\}}}} c_{pqr}\left(x'_{pq}x'_{pr}x'_{qr} - x_{pq}x_{pr}x_{qr}\right) + \sum_{\substack{pq\in \delta(R) \\ pq \not\in \{ij, ik\}}} c_{pq}(x'_{pq} - x_{pq}) 
\\
& \leq c_{ijk} + c_{ij} + c_{ik} +\sum_{\substack{pqr\in T_{\{ij, ik\}}\cap T^+ \\ pqr\neq ijk}}c_{pqr} + \sum_{\substack{pqr\in T_{\delta(R)} \\ pqr \notin T_{\{ij, ik\}}}}\vert c_{pqr}\vert + \sum_{\substack{pq\in \delta(R) \\ pq \not\in \{ij, ik\}}}  \vert c_{pq}\vert \\
& = -c_{ijk}^- - 2c_{ij}^- - 2c_{ik}^- +\sum_{\substack{pqr\in T_{\{ij, ik\}}}}c_{pqr}^+ + \sum_{pqr\in T_{\delta(R)}}\vert c_{pqr}\vert - \sum_{pqr\in T_{\{ij, ik\}}}\vert c_{pqr}\vert+ \sum_{pq\in \delta(R)}  \vert c_{pq}\vert 
\\
& = -c_{ijk}^- - 2c_{ij}^- - 2c_{ik}^- -\sum_{\substack{pqr\in T_{\{ij, ik\}}}}c_{pqr}^- + \sum_{pqr\in T_{\delta(R)}}\vert c_{pqr}\vert + \sum_{pq\in \delta(R)}  \vert c_{pq}\vert \\
& \leq 0
\enspace .
\end{align}	
The last inequality follows from 
Assumption~\eqref{eq:triangle-edge-join-1}.
Thirdly, for any $x\in \cp_S$ such that $x_{ik} = x_{jk} = 0$ and $x_{ij} = 1$, the map $\sigma_{ik}\circ \sigma_{\delta(R')}$ is improving by analogous arguments and Assumption~\eqref{eq:triangle-edge-join-2}.
Finally, for any $x\in \cp_S$ such that $x_{ik} = x_{jk} = x_{ij} = 0$, the map $\sigma_{ijk}\circ \sigma_{\delta(ijk)}$ is such that 
\begin{equation}
(\sigma_{ijk}\circ \sigma_{\delta(ijk)})_{pq} = \begin{cases}
0 & \textnormal{if $pq \in \delta(ijk)$} \\
1 & \textnormal{if $pq \in \{ij, ik, jk\}$} \\
x_{pq} & \textnormal{otherwise}
\end{cases}
\enspace .
\end{equation}
Therefore, 
\begin{align}
\phi_c(x') - \phi_c(x) & = c_{ijk} + c_{ij} + c_{ik} + c_{jk} - \sum_{pqr\in T_{\delta(ijk)}\setminus T_{\{ij, ik, jk\}}} c_{pqr}x_{pq}x_{pr}x_{qr} - \sum_{pq\in \delta(ijk)} c_{pq} x_{pq} 
\\
& \leq c_{ijk} + c_{ij} + c_{ik} + c_{jk}  +\sum_{\substack{pqr\in T_{\delta(ijk)} \cap T^- \\ pqr\not\in T_{\{ij, ik, jk\}}}} \vert c_{pqr} \vert + \sum_{pq\in \delta(ijk)\cap P^-} \vert c_{pq}\vert  \leq 0
\enspace .
\end{align}
The last inequality is true thanks to Assumption~\eqref{eq:triangle-edge-join-3}.
Applying Corollary~\ref{lemma:persistency-variable} concludes the proof.
\end{proof}
}{
In the proof, we distinguish four cases. 
As always, we start by fixing $x \in \cp_S$.
Firstly, if $x_{ik} = 1$, we are good.
Secondly, if $x_{ik} = x_{ij} = 0$, $x_{jk} = 1$ and \eqref{eq:triangle-edge-join-1} is satisfied, we show $\sigma_{ik} \circ \sigma_{\delta(R)}$ to be improving.  
Thirdly, if $x_{ik} = x_{jk} = 0$, $x_{ij} = 1$ and \eqref{eq:triangle-edge-join-2} is satisfied, we prove that $\sigma_{ik} \circ \sigma_{\delta(R')}$ is improving.
Lastly, if $x_{ij} = x_{jk} = x_{ik} = 0$ and \eqref{eq:triangle-edge-join-3} is satisfied, then we apply $\sigma_{ijk}\circ \sigma_{\delta(ijk)}$ and show that it is improving.
}

Next, we discuss a generalization of Theorem~2 of \citet{Lange-2019} 
in the context of instances on complete graphs with a cubic objective function.
To this end, let $S_H \subseteq S$. 
We define $c'\in \mathbb{R}^{\mathcal{I}_{S_H}}$ by the equations written below.
\begin{align}
	c'_\emptyset 
	& = \frac 12 \sum_{pqr\in \tbinom{S_H}{3}}c_{pqr} + \sum_{pq\in \tbinom{S_H}{2}}c_{pq}	
	\label{eq: defcprime1} 
	\\
	\forall pq\in \tbinom{S_H}{2} \colon \quad 
	c'_{pq} 
	& = - c_{pq} + \frac 12 \sum_{\substack{r\in S_H \\ r\neq p, q}} c_{pqr} 
	\label{eq: defcprime2} 
	\\
	\forall pqr\in \tbinom{S_H}{3} \colon \quad 
	c'_{pqr} 
	& = -2c_{pqr}  
	\label{eq: defcprime3}
\end{align}
\Cref{lemma:general-subgraph-edge-join} studies subsets $S_H \subseteq S$ such that $\mathbbm{1}_{\tbinom{S_H}{2}}$ is a trivial solution to the problem $\max_{x\in \cp_{S_H}}\phi_{c'}(x)$, i.e.~$\max_{x\in \cp_{S_H}}\phi_{c'}(x) = 0$. 
Here, we compare the negative parts of the costs involved in cutting $S_H$ from its complement with the total cost of any inner cut of $S_H$.
See also~\Cref{fig:illustration-subset-join}a.
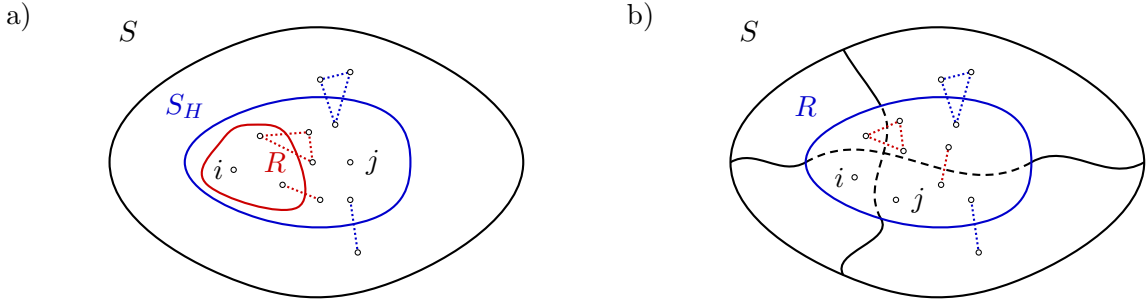
\begin{figure}
	\centering
	\begin{minipage}{0.5\linewidth}
		a)
		\hspace{-2ex}
		\vspace{-6ex}
		\linebreak
		\begin{center}
				\begin{tikzpicture}
				\coordinate (c1) at (0, 0);
				\coordinate (c2) at (1.5, -1.5);
				\coordinate (c3) at (3.75, -1.6);
				\coordinate (c4) at (3.75, 1.6); 
				\coordinate (c5) at (1.5, 1.5);
				\coordinate (c6) at (5.5, 0);
				
				\draw[thick] plot[smooth cycle, tension=0.75] coordinates{(c1) (c2) (c3) (c6) (c4) (c5)};
				
				\coordinate (sh1) at (1, 0);
				\coordinate (sh2) at (2, -0.75);
				\coordinate (sh3) at (3.5, -0.75);
				\coordinate (sh4) at (3.5, 0.75);
				\coordinate (sh5) at (2, 0.75);
				\coordinate (sh6) at (4, 0);
				
				\coordinate (r1) at (1.5, 0.3);
				\coordinate (r2) at (1.9, 0.5);
				\coordinate (r3) at (2.5, -0.6);
				\coordinate (r4) at (1.3, -0.3);
				\coordinate (r5) at (2.4, 0.3);
				
				\draw[deep-red, thick] plot[smooth cycle, tension=0.75] coordinates{(r1) (r2) (r5) (r3) (r4)};
				\node[deep-red] at (2.2, 0.0) {$R$};
				\draw[deep-red, thick, densely dotted] (2.7, 0.0) -- (2.65, 0.4) -- (2.0, 0.35) -- cycle;
				\draw[deep-red, thick, densely dotted] (2.8, -0.5) -- (2.3, -0.3);
				\draw[fill=white] (2.7, 0.0) circle (1pt);
				\draw[fill=white] (2.65, 0.4) circle (1pt);
				\draw[fill=white] (2.0, 0.35) circle (1pt);
				\draw[fill=white] (2.8, -0.5) circle (1pt);
				\draw[fill=white] (2.3, -0.3) circle (1pt);
				
				\draw[secondary, thick] plot[smooth cycle, tension=0.75] coordinates{(sh1) (sh2) (sh3) (sh6) (sh4) (sh5)};
				\node[secondary] at (1, 0.75) {$S_H$};
				\draw[secondary, thick, densely dotted] (3, 0.5) -- (3.2, 1.2) -- (2.8, 1.1) -- cycle;
				\draw[secondary, thick, densely dotted] (3.2, -0.5) -- (3.3, -1.2);
				\draw[fill=white] (3.0, 0.5) circle (1pt);
				\draw[fill=white] (3.2, 1.2) circle (1pt);
				\draw[fill=white] (2.8, 1.1) circle (1pt);
				\draw[fill=white] (3.2, -0.5) circle (1pt);
				\draw[fill=white] (3.3, -1.2) circle (1pt);

				%

				\coordinate (i) at (1.45, -0.1);
				\draw ($(i) + (0.2, 0.0)$) circle (1pt);
				\node at (i) {$i$};
				
				\coordinate (j) at (3.5, 0);
				\draw ($(j) + (-0.3, 0)$) circle (1pt);
				\node at (j) {$j$};
				
				\node at (0.25, 1.75) {$S$};		
			\end{tikzpicture}
		\end{center}
	\end{minipage}%
	\begin{minipage}{0.5\linewidth}
		b)
		\hspace{-2ex}
		\vspace{-6ex}
		\linebreak
		\begin{center}
			\begin{tikzpicture}
				\coordinate (c1) at (0, 0);
				\coordinate (c2) at (1.5, -1.5);
				\coordinate (c3) at (3.75, -1.6);
				\coordinate (c4) at (3.75, 1.6); 
				\coordinate (c5) at (1.5, 1.5);
				\coordinate (c6) at (5.5, 0);
				
				\draw[thick] plot[smooth cycle, tension=0.75] coordinates{(c1) (c2) (c3) (c6) (c4) (c5)};
				
				\coordinate (sh1) at (1, 0);
				\coordinate (sh2) at (2, -0.75);
				\coordinate (sh3) at (3.5, -0.75);
				\coordinate (sh4) at (3.5, 0.75);
				\coordinate (sh5) at (2, 0.75);
				\coordinate (sh6) at (4, 0);
				
				\coordinate (r1) at (1.5, 0.3);
				\coordinate (r2) at (1.9, 0.5);
				\coordinate (r3) at (2.5, -0.6);
				\coordinate (r4) at (1.3, -0.3);
				\coordinate (r5) at (2.4, 0.3);
				
				\draw[solid, black, thick] (c2) to[out=120, in=-60] (sh2);
				\draw[densely dashed, black, thick] (sh2) to[out=120, in=-60] (sh5);
				\draw[solid, black, thick] (sh5) to[out=120, in=-60] (c5);
				
				\draw[solid, black, thick] (c1) to [out=30, in=210] (sh1);
				\draw[densely dashed, black, thick] (sh1) to [out=30, in=210] (sh6);
				\draw[solid, black, thick] (sh6) to [out=30, in=210] (c6);

				\draw[deep-red, thick, densely dotted] (2.3, 0.15) -- (2.25, 0.55) -- (1.8, 0.35) -- cycle;
				\draw[deep-red, thick, densely dotted] (2.8, -0.3) -- (2.9, 0.2);
				\draw[fill=white] (2.3, 0.15) circle (1pt);
				\draw[fill=white] (2.25, 0.55) circle (1pt);
				\draw[fill=white] (1.8, 0.35) circle (1pt);
				\draw[fill=white] (2.9, 0.2) circle (1pt);
				\draw[fill=white] (2.8, -0.3) circle (1pt);
				
				\draw[secondary, thick] plot[smooth cycle, tension=0.75] coordinates{(sh1) (sh2) (sh3) (sh6) (sh4) (sh5)};
				\node[secondary] at (1, 0.75) {$R$};
				\draw[secondary, thick, densely dotted] (3, 0.5) -- (3.2, 1.2) -- (2.8, 1.1) -- cycle;
				\draw[secondary, thick, densely dotted] (3.2, -0.5) -- (3.3, -1.2);
				\draw[fill=white] (3.0, 0.5) circle (1pt);
				\draw[fill=white] (3.2, 1.2) circle (1pt);
				\draw[fill=white] (2.8, 1.1) circle (1pt);
				\draw[fill=white] (3.2, -0.5) circle (1pt);
				\draw[fill=white] (3.3, -1.2) circle (1pt);

				\coordinate (i) at (1.45, -0.2);
				\draw ($(i) + (0.2, 0.0)$) circle (1pt);
				\node at (i) {$i$};
				
				\coordinate (j) at (2.5, -0.5);
				\draw ($(j) + (-0.3, 0)$) circle (1pt);
				\node at (j) {$j$};
				
				\node at (0.25, 1.75) {$S$};		
			\end{tikzpicture}
		\end{center}
	\end{minipage}
	\caption{a) \Cref{lemma:general-subgraph-edge-join} compares the total penalty (blue, dotted tuples) of cutting $S_H$ from its complement $S \setminus S_H$to the total cost (red, dotted tuples) of edges and triples lying within $S_H$ and  cut by $R$. b) \Cref{proposition:subset-join-proposition} compares the total cost of separated edges and triples lying within $R$ (red, dotted) to the total cost of edges and triples cut by $R$ (blue, dotted), for any partition (black) of the node set separating any two nodes $i$ and $j$.}
	\label{fig:illustration-subset-join}
\end{figure}
\begin{proposition}
	\label{lemma:general-subgraph-edge-join}
	Let $S \neq \emptyset$ and $c\in \mathbb{R}^{\mathcal{I}_S}$. 
	Moreover, let $S_H \subseteq S$ and $ij \in \tbinom{S_H}{2}$. 
	We define $c'\in \mathbb{R}^{\mathcal{I}_{S_H}}$ as in \eqref{eq: defcprime1}, \eqref{eq: defcprime2}, \eqref{eq: defcprime3}.
	If $\max_{x\in \cp_{S_H}}\phi_{c'}(x) = 0$
	and for all $R \subseteq S_H$ with $i \in R$ and $j \in S_H\setminus R$ we have
	\begin{equation}\label{eq:assumption-subgraph-criterion-uv-cuts}
		\sum_{pq\in \delta(S_H)\cap P^-}c_{pq} + \sum_{pqr\in T_{\delta(S_H)}\cap T^-} c_{pqr} \geq \sum_{pq\in \delta(R, S_H\setminus R)}c_{pq} + \sum_{pqr\in T_{\delta(R, S_H \setminus R)}\cap \tbinom{S_H}{3}} c_{pqr} \enspace ,		
	\end{equation}
	then there exists an optimal solution $x^*$ to \csp{S}{c} such that $x^*_{ij} = 1$.
\end{proposition}
In order to prove this proposition, we first prove an auxiliary lemma.
\begin{lemma}
	\label{lemma:subgraph-helping-lemma-1}
	Let $S \neq \emptyset$ and $c \in \mathbb{R}^{\mathcal{I}_S}$. 
	We define $c'\in \mathbb{R}^{\mathcal{I}_S}$ as in \eqref{eq: defcprime1}, \eqref{eq: defcprime2}, \eqref{eq: defcprime3} for $S_H = S$.
	Then for any partition $\mathcal{R}$ of $S$:
	{
		\begin{align}
			\phi_{c'}(x^\mathcal{R}) 
			&= \;\frac 12 \sum_{pqr\in \tbinom S3} c_{pqr}\prod_{uv\in \tbinom{pqr}{2}}(1-x^\mathcal{R}_{uv}) + \sum_{pqr\in \tbinom S3}c_{pqr}\sum_{uv \in \tbinom{pqr}{2}}x^\mathcal{U}_{uv}\prod_{\substack{u'v'\in \tbinom{pqr}{2} \\ u'v' \neq uv}}(1- x^\mathcal{R}_{u'v'})	
			\\
			& \qquad + \sum_{pq\in \tbinom S2}c_{pq}(1-x^\mathcal{R}_{pq})	\label{eq:partition-subgraph-helper-1} 
			\\ 
			& = \;\frac 12 \sum_{RR'R''\in \tbinom{\mathcal{R}}{3}}\sum_{pqr\in T_{RR'R''}}c_{pqr} +\sum_{RR'\in \tbinom{\mathcal{R}}{2}}\sum_{pq\in \delta(R, R')}c_{pq} 
			\\
			& \qquad + \sum_{RR'\in \tbinom{\mathcal{R}}{2}}\left(\sum_{pqr\in T_{RRR'}}c_{pqr} +  \sum_{pqr\in T_{RR'R'}}c_{pqr}\right) 
			\enspace ,
			\label{eq:partition-subgraph-helper-2}
		\end{align}
	}%
	where $x^\mathcal{R}$ denotes the feasible vector corresponding to the partition $\mathcal{R}$ of $S$.
\end{lemma}
\begin{proof}[Proof of \Cref{lemma:subgraph-helping-lemma-1}]
We use the fact that for any partition $\mathcal{R}$ of $S$ and any $pqr\in \tbinom S3$ we have that
$x_{pq}^\mathcal{R}x_{qr}^\mathcal{R}=x_{pq}^\mathcal{R}x_{pr}^\mathcal{R} = x_{pr}^\mathcal{R}x_{qr}^\mathcal{R} = x_{pq}^\mathcal{R}x_{pr}^\mathcal{R}x_{qr}^\mathcal{R}$. 
Expanding the inner products and the inner sums leads to
\begin{align}
		\prod_{uv\in \tbinom{pqr}{3}}\left(1-x^\mathcal{R}_{uv}\right) & = 1 - x_{pq}^\mathcal{R} - x_{pr}^\mathcal{R}-x_{qr}^\mathcal{R} + 2x_{pq}^\mathcal{R}x_{pr}^\mathcal{R}x_{qr}^\mathcal{R} \enspace, \\
		\sum_{uv\in \tbinom{pqr}{2}}x_{uv}^\mathcal{R}\prod_{\substack{u'v'\in \tbinom{pqr}{2} \\ u'v' \notin \{uv\}}}(1-x_{u'v'}^\mathcal{R}) & = x^\mathcal{R}_{pq} + x^\mathcal{R}_{pr} + x^\mathcal{R}_{qr} -3 x^\mathcal{R}_{pq}x^\mathcal{R}_{qr}x^\mathcal{R}_{pr}\enspace.
\end{align}
By substituting and collecting terms, we conclude the proof for Equation~\eqref{eq:partition-subgraph-helper-1}. 
Equation~\eqref{eq:partition-subgraph-helper-2} follows instead from the following observations: 
\begin{align}
\prod_{ab\in \tbinom{pqr}{3}}\left(1-x^\mathcal{R}_{ab}\right) = 1 
& \quad \Leftrightarrow \quad \exists \;RR'R''\in \tbinom{\mathcal{R}}{3}\colon pqr\in T_{RR'R''}
\\
\sum_{ab\in \tbinom{pqr}{2}}x_{ab}^\mathcal{R}\prod_{a'b'\in \tbinom{pqr}{2}\setminus \{ab\}}(1-x_{a'b'}^\mathcal{R}) = 1 
& \quad \Leftrightarrow \quad \exists \;RR'\in \tbinom{\mathcal{R}}{2}\colon \left(pqr\in T_{RRR'} \lor pqr\in T_{RR'R'}\right)
\\
1 - x^\mathcal{R}_{pq} = 1 
& \quad \Leftrightarrow \quad \exists \; RR'\in \tbinom{\mathcal{R}}{2}\colon pq\in \delta(R, R') 
\enspace .
\end{align}
This concludes the proof.
\end{proof}

\begin{proof}[Proof of~\Cref{lemma:general-subgraph-edge-join}]
We define $\sigma \colon \cp_S \to \cp_S$ as 
\begin{equation}
\sigma(x) := \begin{cases}
x & \textnormal{if $x_{ij} = 1$}\\
(\sigma_{S_H} \circ \sigma_{\delta(S_H)})(x) & \textnormal{otherwise}
\end{cases}
\enspace .
\end{equation}
Let $x' = \sigma(x)$ for any $x \in \cp_S$. 
It is easy to see that $x'_{ij} = 1$ for all $x\in \cp_S$.
Similarly to before, 
we show that $\sigma$ is an improving map. 
For any $x \in \cp_S$ such that $x_{ij} = 1$ we have $\phi_c(x') - \phi_c(x) = 0$, by definition of $x'$. 
Now, we consider $x\in \cp_S$ such that $x_{ij} = 0$. 
Let $P_H = \tbinom{S_H}{2}$ and $T_H = \tbinom{S_H}{3}$. 
We let $x\vert_{P_H}$ denote the restriction of $x$ containing only components corresponding to elements in $P_H$. 
Let $\mathcal{R}$ be the partition of $S$ such that $x = x^\mathcal{R}$, and let $\mathcal{R}_H$ be the induced partition of $S_H$ such that $x\vert_{P_H} = x^{\mathcal{R}_H}$. 
Since $x_{ij} = 0$, there exist $R_1, R_2\in \mathcal{R}_H$ such that $i \in R_1$, $j \in R_2$. 
We have:
\begin{equation}
x'_{pq} = \begin{cases}
1 & \textnormal{if $pq \in P_H$}\\
0 & \textnormal{if $pq \in \delta(S_H)$} \\
x_{pq} & \textnormal{otherwise}
\end{cases}
\enspace .
\end{equation}
Therefore, it follows 
\begin{align}
& \phi_c(x') - \phi_c(x) \\
= & \sum_{pq\in P_H} c_{pq}(1- x_{pq}) + \sum_{pqr\in T_H}c_{pqr}(1-x_{pq}x_{pr}x_{qr}) - \sum_{pq\in \delta(S_H)}c_{pq}x_{pq} - \sum_{pqr\in T_{\delta(S_H)}}c_{pqr}x_{pq}x_{pr}x_{qr}
\enspace .
\label{eq:subgraph-plugin-map-1} 
\end{align}
In order to find an upper bound for the sums over $P_H$ and $T_H$, we show that there exists a subset $R \subset S_H$ with $i \in R$ and $j \in S_H \setminus R$ such that 
\begin{equation} \label{eq:subgraph-simplification-1}
\sum_{pq\in P_H}c_{pq}\left(1-x_{pq}\right) + \sum_{pqr\in T_H}c_{pqr}(1-x_{pq}x_{pr}x_{qr}) \leq \sum_{pq\in \delta(R, S_H \setminus R)}c_{pq} + \sum_{pqr\in T_{\delta(R, V_H \setminus R)}\cap T_H} c_{pqr}
\enspace .
\end{equation}
	
For the sake of contradiction, we assume that there is no such $R \subset S_H$.
For any $\mathcal{R}'\subset \mathcal{R}_H$, let $R_{\mathcal{R}'} = \bigcup_{P'\in \mathcal{R}'}P'$. 
Furthermore, we define $t \colon\tbinom{\mathcal{R}_H}{2}\cup \tbinom{\mathcal{R}_H}{3} \to \mathbb{R}$ and $p \colon \tbinom{\mathcal{R}_H}{2} \to \mathbb{R}$ such that
\begin{align}
t_{RR'R''} & = \sum_{pqr\in T_{RR'R''}}c_{pqr}
& \forall RR'R''\in \tbinom{\mathcal{R}_H}{3} 
\\ 
t_{RR'} & = \sum_{pqr\in T_{RRR'}\cup T_{RR'R'}}c_{pqr}
& \forall RR'\in \tbinom{\mathcal{R}_H}{2} 
\\
p_{RR'} & = \sum_{pq\in \delta(R, R')}c_{pq}
& \forall RR'\in \tbinom{\mathcal{R}_H}{2} 
\enspace .
\end{align}
Therefore, let $\mathcal{R}'\subset \mathcal{R}_H$ with $R_1\in \mathcal{R}'$ and $R_2\not\in \mathcal{R}'$. 
We observe that this implies that $i \in R_{\mathcal{R}'}$ and $j \notin R_{\mathcal{R}'}$, since $\mathcal{R}_H$ is a partition of $H$. 
We have
\begin{equation}\label{eq:subgraph-simplification-2}
\sum_{pq \in P_H}c_{pq}\left(1-x_{pq}\right) + \sum_{pqr\in T_H}c_{pqr}(1-x_{pq}x_{pr}x_{qr}) > \sum_{pq \in \delta(R_{\mathcal{R}'}, S_H\setminus R_{\mathcal{R}'})}c_{pq} + \sum_{pqr\in T_{\delta(R_{\mathcal{R}'}, S_H\setminus R_{\mathcal{R}'})}\cap T_H} c_{pqr}
\enspace .
\end{equation}
We evaluate the terms in \eqref{eq:subgraph-simplification-2} one-by-one, and express them as sums over elements in $\mathcal{R}'$ and $\mathcal{R}_H \setminus \mathcal{R}'$.
Firstly, we observe that for any $pq \in P_H$ we have $x_{pq} = 0$ if and only if there exist $RR'\in \tbinom{\mathcal{R}_H}{2}$ such that $pq\in \delta(R, R')$.
Therefore,
\begin{equation}
\sum_{pq \in P_H}c_{pq}\left(1-x_{pq}\right) = \sum_{RR'\in \tbinom{\mathcal{R}_H}{2}}p_{RR'}
\end{equation}
whereas
\begin{equation}
\sum_{pq \in \delta(R_{\mathcal{R}'}, S_H \setminus R_{\mathcal{R}'})}c_{pq} = \sum_{R \in \mathcal{R}'} \sum_{R'\in \mathcal{R}_H \setminus \mathcal{R}'} p_{RR'}
\enspace .
\end{equation}
For the first sum, we use the decomposition 
\begin{equation} \label{eq: decomposition1}
\tbinom{\mathcal{R}_H}{2} = \tbinom{\mathcal{R}'}{2}\cup \left\{RR' \mid R \in \mathcal{R}'\land R'\in \mathcal{R}_H\setminus \mathcal{R}'\right\} \cup \tbinom{\mathcal{R}_H\setminus \mathcal{R}'}{2}
\end{equation}
where the subsets are mutually disjoint.
Consequently:
\begin{equation}\label{eq:subgraph-edge-difference}
\sum_{pq\in P_H}c_{pq}\left(1-x_{pq}\right) - \sum_{pq\in \delta(W_{\mathcal{R}'}, V_H\setminus W_{\mathcal{R}'})}c_{pq} = \sum_{RR'\in \tbinom{\mathcal{R}'}{2}}p_{RR'} + \sum_{RR'\in \tbinom{\mathcal{R}_H\setminus \mathcal{R}'}{2}}p_{RR'}
\enspace .
\end{equation}
Secondly, for any $pqr\in T_H$, we have $x_{pq}x_{pr}x_{qr} = 0$ if and only if there exist $RR'\in \tbinom{\mathcal{R}_H}{2}$ such that $pqr\in T_{RRR'}\cup T_{RR'R'}$ or there exist $RR'R''\in \tbinom{\mathcal{R}_H}{3}$ such that $pqr\in T_{RR'R''}$.
Therefore,
\begin{equation}
\sum_{pqr\in T_H}c_{pqr} \left(1 - x_{pq}x_{pr}x_{qr}\right) = \sum_{RR'R''\in \tbinom{\mathcal{R}_H}{3}}t_{RR'R''} + \sum_{RR'\in \tbinom{\mathcal{R}_H}{2}}t_{RR'} 
\end{equation}
whereas
\begin{align}
& \sum_{pqr\in T_{\delta(R_{\mathcal{R}'}, S_H \setminus R_{\mathcal{R}'})}\cap T_H} c_{pqr} 
\\
= & \sum_{RR'\in \tbinom{\mathcal{R}'}{2}}\sum_{R''\in \mathcal{R}_H \setminus \mathcal{R}'}t_{RR'R''} + \sum_{R\in \mathcal{R}'}\sum_{R'R''\in \tbinom{\mathcal{R}_H \setminus \mathcal{R}'}{2}}t_{RR'R''} + \sum_{R\in \mathcal{R}'}\sum_{R'\in \mathcal{R}_H \setminus \mathcal{R}'}t_{RR'}
\enspace .
\end{align}
For the first sum, we use the decomposition
\begin{align} \label{eq: decomposition2}
\tbinom{\mathcal{R}_H}{3} = & \ \tbinom{\mathcal{R}'}{3}\cup \left\{RR'R''\mid RR'\in \tbinom{\mathcal{R'}}{2}\land R''\in \mathcal{R}_H \setminus \mathcal{R}'\right\} \\
& \cup \left\{RR'R''\mid R\in \mathcal{R}'\land R'R''\in \tbinom{\mathcal{R}_H\setminus \mathcal{R}'}{2}\right\} \cup \tbinom{\mathcal{R}_H \setminus \mathcal{R}'}{3}
\end{align}
where again the subsets are mutually disjoint.	
By \eqref{eq: decomposition1} and \eqref{eq: decomposition2}, 
it follows 
\begin{align}
& \sum_{pqr\in T_H}c_{pqr} \left(1 - x_{pq}x_{pr}x_{qr}\right) - \sum_{pqr\in T_{\delta(R_{\mathcal{R}'}, S_H \setminus R_{\mathcal{R}'})}\cap T_H}c_{pqr} 
\\ 
= & \sum_{RR'R''\in \tbinom{\mathcal{R'}}{3}}t_{RR'R''} + \sum_{RR'R''\in \tbinom{\mathcal{R}_H\setminus \mathcal{R'}}{3}}t_{RR'R''}
+ \sum_{RR'\in \tbinom{\mathcal{R'}}{2}}t_{RR'} + \sum_{RR'\in \tbinom{\mathcal{R}_H\setminus \mathcal{R'}}{2}}t_{RR'}
\enspace . 
\label{eq:subgraph-triplet-difference}
\end{align}
\noindent
Combining \eqref{eq:subgraph-simplification-2}, \eqref{eq:subgraph-edge-difference} and \eqref{eq:subgraph-triplet-difference} 
yields
\begin{align}
0 < 
& \sum_{pq\in P_H} c_{pq}(1-x_{pq}) 
+ \sum_{pqr\in T_H} c_{pqr}(1-x_{pq}x_{pr}x_{qr}) 
- \sum_{pq\in \delta(R_{\mathcal{R}'}, S_H \setminus R_{\mathcal{R}'})} \hspace{-3ex} c_{pq} 
- \sum_{pqr\in T_{\delta(R_{\mathcal{R}'}, S_H \setminus R_{\mathcal{R}'})}\cap T_H} \hspace{-4ex} c_{pqr}
\\
& = \sum_{RR'\in \tbinom{\mathcal{R'}}{2}}p_{RR'} + \sum_{RR'\in \tbinom{\mathcal{R}_H\setminus \mathcal{R'}}{2}}p_{RR'} + \sum_{RR'R''\in \tbinom{\mathcal{R'}}{3}}t_{RR'R''} + \sum_{RR'R''\in \tbinom{\mathcal{R}_H\setminus \mathcal{R'}}{3}}t_{RR'R''}
\\
& + \sum_{RR'\in \tbinom{\mathcal{R'}}{2}}t_{RR'} + \sum_{RR'\in \tbinom{\mathcal{R}_H\setminus \mathcal{R'}}{2}}t_{RR'} =: S_{\mathcal{R}'}
\enspace .
\end{align}
Let $k = \vert \mathcal{R}_H\vert$, and $S_{\mathcal{R}'}$ the right-hand side of the last inequality. 
Recall that $R_1, R_2 \in \mathcal{R}_H$, $R_1 \in \mathcal R'$, and $R_2 \notin \mathcal R'$. 
As $S_{\mathcal{R}'} > 0$, it follows that at least one of the sums in its definition must not be vacuous. 
Moreover, since its sums are indexed by pairs or triplets of subsets all belonging either to $\mathcal R'$ or to $\mathcal R_H \setminus \mathcal R'$, we observe that there must exist at least another subset 
of elements in $\mathcal R_H$ different from $R_1$ and $R_2$. 
Hence, $k \geq 3$. 
We calculate 
\begin{equation}
\mathcal S = \sum_{\substack{\mathcal{R'}\subseteq \mathcal{R}_H\colon \\R_1\in \mathcal{R}', R_2\not\in \mathcal{R}'}}S_{\mathcal{R}'}.
\end{equation}
We need this in order to contradict $\max_{x\in \cp_{S_H}}\phi_{c'}(x) = 0$. 
For any $RR'\in \tbinom{\mathcal{R}_H}{2}\setminus \{R_1R_2\}$, there are exactly $2^{k-3}$ 
subsets $\mathcal{R}'\subseteq \mathcal{R}_H$ such that $p_{RR'}$ or $t_{RR'}$ occurs in $S_{\mathcal{R}'}$ and $R_1\in \mathcal{R}', R_2 \not\in \mathcal{R}'$. 
There is no $\mathcal{R}'\subseteq \mathcal{R}_H$ such that $p_{R_1R_2}$ or $t_{R_1R_2}$ occurs in $S_{\mathcal{R}'}$ with $R_1\in \mathcal{R}', R_2\not\in \mathcal{R}'$. 
For any $RR'R''\in \tbinom{\mathcal{R}_H}{3}\setminus\left\{R_1R_2R\mid R\in \mathcal{R}_H\setminus \{R_1, R_2\}\right\}$, there are exactly $\lfloor 2^{k-4} \rfloor$ subsets $\mathcal{R}'\subseteq \mathcal{R}_H$ such that $t_{RR'R''}$ occurs in $S_{\mathcal{R}'}$ and $R_1\in \mathcal{R}', R_2\not\in \mathcal{R}'$. 
There is no $\mathcal{R}'\subseteq \mathcal{R}_H$ such that $t_{R_1R_2R}$ occurs in $S_{\mathcal{R}'}$ for any $R\in \mathcal{R}_H\setminus \{R_1, R_2\}$ for which $R_1\in \mathcal{R}', R_2\not\in \mathcal{R}'$.
Therefore, 
\begin{align}
0 < \mathcal S 
& = 2^{k-3}\sum_{RR'\in \tbinom{\mathcal{R}_H}{2}}p_{RR'} - 2^{k-3}p_{R_1R_2} + \lfloor 2^{k-4} \rfloor \sum_{RR'R''\in \tbinom{\mathcal{R}_H}{3}}t_{RR'R''} - \lfloor 2^{k-4} \rfloor \sum_{\substack{R\in \mathcal{R}_H \\ R \not\in \{R_1, R_2\}}}t_{R_1R_2R}
\\
& \qquad + 2^{k-3}\sum_{RR'\in \tbinom{\mathcal{R}_H}{2}}t_{RR'} - 2^{k-3}t_{R_1R_2} 
\\
& = 2^{k-3}\sum_{RR'\in \tbinom{\mathcal{R''}}{2}}p_{RR'}+ \lfloor 2^{k-4} \rfloor \sum_{RR'R''\in \tbinom{\mathcal{R''}}{3}}t_{RR'R''} + 2^{k-3}\sum_{RR'\in \tbinom{\mathcal{R''}}{2}}t_{RR'} = 2^{k-3}\phi_{c'}(x^{\mathcal{R''}})
\enspace ,
\end{align}
where $\mathcal{R}'' = \left(\mathcal{R}_H\setminus \{R_1, R_2\}\right)\cup \{R_1 \cup  R_2\}$ is the partition obtained by merging $R_1$ and $R_2$. 
The last equality follows from Lemma~\ref{lemma:subgraph-helping-lemma-1}. 
That contradicts $\max_{x\in \cp_{S_H}}\phi_{c'}(x) = 0$. 
Therefore, this implies that there exists a subset $R \subset S_H$ with $i \in R$ and $j \in S_H \setminus R$ such that inequality \eqref{eq:subgraph-simplification-1} is fulfilled.
	
Let $R \subset S_H$ be a subset such that \eqref{eq:subgraph-simplification-1} holds. Therefore, we have that
\begin{align}
\phi_c(x') - \phi_c(x) 
& \overset{\eqref{eq:subgraph-plugin-map-1}}{=} \sum_{pq\in P_H} c_{pq}(1- x_{pq}) + \sum_{pqr\in T_H}c_{pqr}(1-x_{pq}x_{pr}x_{qr}) - \sum_{pq\in \delta(S_H)}c_{pq}x_{pq} 
\\
& \qquad - \sum_{pqr\in T_{\delta(S_H)}}c_{pqr}x_{pq}x_{pr}x_{qr} 
\\
& \overset{\eqref{eq:subgraph-simplification-1}}{\leq} \sum_{pq\in \delta(R, S_H\setminus R)}c_{pq} + \sum_{pqr\in T_{\delta(R, S_H\setminus R)}\cap T_H} c_{pqr} - \sum_{pq\in \delta(S_H)\cap P^-}c_{pq}- \sum_{pqr\in T_{\delta(S_H)}\cap T^-}c_{pqr} 
\\
& \overset{\eqref{eq:assumption-subgraph-criterion-uv-cuts}}{\leq} \;0.
\end{align}
Consequently, the map $p$ is improving. 
By applying Corollary~\ref{lemma:persistency-variable}, we conclude the proof.
\end{proof}


We are unaware of an efficient method for finding subsets $S_H \subseteq S$ and $R\subseteq S$ for which \eqref{eq:assumption-subgraph-criterion-uv-cuts} are satisfied. For subsets $S_H$ with $\vert S_H \vert \in \{2,3\}$, two corollaries of  \Cref{lemma:general-subgraph-edge-join} provide efficiently-verifiable partial optimality conditions:

\begin{corollary}
	\label{corollary:edge-subgraph-edge-join}
	Let $S \neq \emptyset$, $c\in \mathbb{R}^{\mathcal{I}_S}$ and $ij \in \tbinom S2$. 
	If  
	\begin{equation}
		c_{ij} \leq \sum_{pq\in \delta(ij)\cap P^-}c_{pq} + \sum_{pqr\in T_{\delta(ij)}\cap T^-}c_{pqr}
	\end{equation}
	then there exists an optimal solution $x^*$ to \csp{S}{c} such that $x^*_{ij} = 1$.
\end{corollary}

\begin{corollary}
	\label{corollary:triplet-subgraph-edge-join}
	Let $S \neq \emptyset$, $c\in \mathbb{R}^{\mathcal{I}_S}$ and $ijk \in \tbinom S3$. 
	If 
	\begin{align}
		c_{ij} + c_{ik} & \leq 0 
		\\
		c_{ij} + c_{jk} & \leq 0 
		\\
		c_{ik} + c_{jk} & \leq 0
		\\
		c_{ij} + c_{ik} + c_{jk} & \leq 0
		\\
		c_{ij} + c_{ik} + c_{jk} + \frac 12 c_{ijk} & \leq 0 
		\\
		c_{ij} + c_{ik} + c_{ijk} 
		& \leq \sum_{pq \in \delta(ijk) \cap P^-} c_{pq}
		+ \sum_{pqr \in T_{\delta(ijk)}\cap T^-} c_{pqr}
		\\
		c_{jk} + c_{ik} + c_{ijk} 
		& \leq \sum_{pq \in \delta(ijk) \cap P^-} c_{pq}
		+ \sum_{pqr \in T_{\delta(ijk)} \cap T^-} c_{pqr}
	\end{align}
	then there exists an optimal solution $x^*$ to \csp{S}{c} such that $x^*_{ik} = 1$.
\end{corollary}

We now present the last main partial optimality condition of this article.
It observes that separating a whole subset $R$ from the rest and then joining everything in it yields a better objective value if, for every pair $ij\in \tbinom R2$ and any partition of $S$ that separates $i$ from $j$, the sum of costs of separated pairs and triples within $R$ is at most the sum of costs of joined pairs and triples cut by $R$.
See also~\Cref{fig:illustration-subset-join}b.
\begin{proposition}
	\label{proposition:subset-join-proposition}
	Let $S \neq \emptyset$ and $c\in \mathbb{R}^{\mathcal{I}_S}$. 
	Moreover, let $R \subseteq S$. 
	If for every $ij \in \tbinom R2$ we have 
	{
		\begin{align}
			&\max_{\substack{x\in \cp_S \\ x_{ij} = 0}} \Bigl\{ \sum_{pqr \in \tbinom R3} c_{pqr}(1-x_{pq}x_{pr}x_{qr}) + \sum_{pq\in \tbinom R2}c_{pq}(1-x_{pq}) \Bigr\} 
			\\ 
			\label{eq:subset-join-inequality}
			\leq \quad &\min_{\substack{x\in \cp_S \\x_{ij} = 0}} \Bigl\{ \sum_{pqr\in T_{\delta(R)}}c_{pqr}x_{pq}x_{pr}x_{qr} + \sum_{pq\in \delta(R)}c_{pq}x_{pq} \Bigr\}
		\end{align}
	}%
	then there is an optimal solution $x^*$ to \csp{S}{c} such that $\forall ij \in \tbinom R2 \colon x^*_{ij} = 1$.
\end{proposition}

\ifthenelse{\boolean{proofs}}{
\begin{proof}
We define $\sigma \colon \cp_S \to \cp_S$ such that 
\begin{equation}
\sigma(x) := \begin{cases}
x & \textnormal{if $x_{ij} = 1 \; \forall ij \in \tbinom R2$}\\
\left(\sigma_R \circ \sigma_{\delta(R)}\right)(x) & \textnormal{otherwise}
\end{cases}.
\end{equation}
Let $x' = \sigma(x)$ for every $x \in \cp_S$. 
Firstly, we have $x'_{ij} = 1$ for every $ij \in \tbinom R2$. 
Secondly, we show that $\sigma$ is an improving map. 
Let $x \in \cp_S$ such that $x_{ij} = 1$ for all $ij \in \tbinom R2$.
In this case, we have $\phi_c(x') = \phi_c(x)$ by definition of $x'$. 
Now, let us consider the complementary case, i.e.~let $x \in \cp_S$ such that there exists $ij \in \tbinom R2$ for which $x_{ij} = 0$. 
Then,
\begin{equation}
x'_{pq} = \begin{cases}
1 & \textnormal{if $pq \in \tbinom R2$}\\
0 & \textnormal{if $pq \in \delta(R)$}\\
x_{pq} & \textnormal{otherwise}
\end{cases}
\enspace .
\end{equation} 
Therefore, it follows that
\begin{align}
\phi_c(x') - \phi_c(x) =
& \sum_{pq\in \tbinom R2} c_{pq}(1-x_{pq}) 
- \sum_{pq\in \delta(R)} c_{pq}x_{pq} + \sum_{pqr \in \tbinom R3} c_{pqr}(1-x_{pq}x_{pr}x_{qr}) 
\\
& \qquad - \sum_{pqr\in T_{\delta(R)}} c_{pqr}x_{pq}x_{pr}x_{qr}   
\\
\leq & \max_{\substack{x\in \cp_S \\ x_{ij} = 0}} \Bigl\{ \sum_{pqr\in \tbinom R2} c_{pqr}(1-x_{pq}x_{pr}x_{qr}) + \sum_{pq\in \tbinom R2} c_{pq}(1-x_{pq}) \Bigr\} \\
& - \min_{\substack{x\in \cp_S \\ x_{ij} = 0}}\Bigl\{ \sum_{pqr\in T_{\delta(R)}} c_{pqr}x_{pq}x_{pr}x_{qr} + 
\sum_{pq\in \delta(R)}c_{pq}x_{pq} \Bigr\} 
\\
\overset{\eqref{eq:subset-join-inequality}}{\leq} & 0
\enspace .
\end{align}
This concludes the proof.
\end{proof}
}{
In the proof, we apply the identity to feasible vectors that satisfy our claim, and apply $\sigma = \sigma_{R} \circ \sigma_{\delta(R)}$, otherwise.
}

The above condition, together with its subsequent corollary, has been established independently of prior work on the linear correlation clustering problem. 
We are unaware of an efficient method for checking \eqref{eq:subset-join-inequality} for arbitrary subsets $R \subseteq S$ and costs $c \in \mathbb{R}^{\mathcal{I}_S}$. 
Yet, \Cref{corollary:subset-join-all-pairs-at-once} below describes one setting in which a suitable subset can be searched for heuristically, in polynomial time. 
Specifically, the objective function $c\in \mathbb{R}^{\mathcal{I}_S}$ needs to be such that $c_{pq}\leq 0$ for all $pq\in \tbinom R2$ and $c_{pqr} \leq 0$ for all $pqr\in \tbinom R3$.
An intuition for this corollary is as follows. 
For a moment let us consider a fixed subset of items $R$ and consider all the possible ways in which we could divide $R$ in two parts. 
Let us recall that the costs of all the triples and pairs inside of $R$ are non-positive. 
If the worst possible cost of joining these two parts of $R$ back together is still less than or equal to the reward obtained by joining $R$ with the rest, then we can safely start by putting all the objects of $R$ in the same set and decide independently whether or not to join $R$ with other sets.

\begin{corollary}
	\label{corollary:subset-join-all-pairs-at-once}
	Let $S \neq \emptyset$ and $c\in \mathbb{R}^{\mathcal{I}_S}$. 
	Moreover, let $R \subseteq S$. 
	If 
	\begin{align}
	c_{pq} & \leq 0 & \forall pq \in \tbinom R2 
	\label{eq:cond-submod-1}
	\\
	c_{pqr} & \leq 0 & \forall pqr\in \tbinom R3
	\label{eq:cond-submod-2}
	\end{align}
	and
	\begin{equation}
		\max_{\substack{R'\subset R \\ R' \neq \emptyset}} \Bigl\{ \sum_{pqr\in T_{\delta(R')}\cap \tbinom R3}c_{pqr} + \sum_{pq\in \delta(R', R\setminus R')}c_{pq} \Bigr\} 
		\leq \quad \sum_{pqr\in T_{\delta(R)}\cap T^-}c_{pqr} + \sum_{pq\in \delta(R)\cap P^-}c_{pq} \enspace ,
		\label{eq:subset-join-corollary-all-pairs-condition}
	\end{equation}
	then there is an optimal solution $x^*$ to \csp{S}{c} such that $x^*_{ij} = 1$, $\forall ij\in \tbinom R2$.
\end{corollary}

The previous corollary follows from the following two facts. 
Let \eqref{eq:cond-submod-1} and \eqref{eq:cond-submod-2} be satisfied. 
First of all, for any $ij \in \tbinom R2$, we have that the left-hand side of \eqref{eq:subset-join-inequality} is equal to 
\begin{align}
\max_{\substack{R'\subsetneq R \\ i \in R' \\ j \not\in R'}}\sum_{pqr \in T_{\delta(R')}\cap \tbinom R3}c_{pqr} + \sum_{pq \in \delta(R', R \setminus R')} c_{pq}
\enspace .
\end{align}
I.e., the maximizer is given by a feasible $x\in \cp_S$ whose restriction to $\tbinom R2$ corresponds to a partition $\mathcal{R}$ of 
$R$ into two subsets. 
To see this, note that for any $ij\in \tbinom R2$, instead of maximizing the left-hand side of \eqref{eq:subset-join-inequality} over $x\in X_S$ such that $x_{ij} = 0$ we can equivalently maximize over all $x'\in X_R$ such that $x'_{ij} = 0$. 
Now, let us assume that there exists $ij \in \tbinom R2$ such that the maximizer is given by a feasible $x' \in \cp_R$ corresponding to a partition $\mathcal{R}'$ of 
$R$ into more than two subsets.
Without loss of generality, let $R_1, R_2 \in \mathcal{R}'$ such that $i\in R_1$, $j\in R_2$. 
Then, the vector $x'' \in \cp_R$ corresponding to the partition $\mathcal{R}'' = \left\{ R_1, \bigcup_{R\in \mathcal{R}' \setminus \{R_1\}} R\right\}$ has objective value at least the objective value of $x'$.
This follows from the facts that all the costs are non-positive and that $\mathcal{R}'$ is a refinement of $\mathcal{R}''$.
Secondly, by using the trivial lower bound of the right-hand side of \eqref{eq:subset-join-inequality} we have that
it is at least 
\begin{align}
\sum_{pqr\in T_{\delta(R)}\cap T^-}c_{pqr} + \sum_{pqr\in \delta(R)\cap P^-}c_{pq}
\enspace .
\end{align}

\section{Efficient Testing of Partial Optimality}\label{sec: testing}

Next, we describe algorithms for testing all the partial optimality conditions introduced above.
This includes exact algorithms and heuristics, and we discuss their runtimes.
We start by examining \Cref{lemma:persistency-subset-separation}. 
\Cref{alg:region-growing-separation} terminates in $O(\vert S\vert^3)$ time and finds a subset $R \subseteq S$ that satisfies \eqref{eq:edge-cut-condition-1}--\eqref{eq:edge-cut-condition-2}. 
Note that \eqref{eq:edge-cut-condition-1}--\eqref{eq:edge-cut-condition-2} hold in particular for the trivial subset $R = S$.
We formalize the correctness of \Cref{alg:region-growing-separation} in \Cref{prop: algo works}.

\begin{proposition}\label{prop: algo works}
Let $S \neq \emptyset$, $c \in \mathbb{R}^{\mathcal{I}_S}$.
Then, \Cref{alg:region-growing-separation} outputs a partition that contains a subset satisfying \eqref{eq:edge-cut-condition-1}--\eqref{eq:edge-cut-condition-2}.
Moreover, if there exists a non-trivial $R \subseteq S$, namely $R \neq \emptyset, S$, such that \eqref{eq:edge-cut-condition-1}--\eqref{eq:edge-cut-condition-2} are satisfied by $R$, then \Cref{alg:region-growing-separation} returns a non-trivial partition.
\end{proposition}

\ifthenelse{\boolean{proofs}}{
\begin{proof}
We start by observing that \Cref{alg:region-growing-separation} always terminates.
If it returns a non-trivial partition $\mathcal{R}$, then $\mathcal{R}$ contains a subset $R$ that satisfies \eqref{eq:edge-cut-condition-1}--\eqref{eq:edge-cut-condition-2} by construction.
Therefore, let us assume that the output of \Cref{alg:region-growing-separation} is the trivial partition $\mathcal{R} = \{ S \}$.
If indeed there exists no non-trivial subset of $S$ for which \eqref{eq:edge-cut-condition-1}--\eqref{eq:edge-cut-condition-2} hold, then \Cref{alg:region-growing-separation} is returning the correct output.
Next, we consider the case in which there exists a non-trivial subset of $S$ that satisfies the assumptions of \Cref{lemma:persistency-subset-separation}, but \Cref{alg:region-growing-separation} still returns the trivial partition.
We prove that this cannot happen.
Let $R \subseteq S$ be a non-trivial subset of $S$ for which \eqref{eq:edge-cut-condition-1}--\eqref{eq:edge-cut-condition-2} are satisfied.
Note that such a set must exist by the assumptions of this case.
Moreover, we have that both $R$ and $S \setminus R$ are non-empty.
Two cases can arise at this point: \Cref{alg:region-growing-separation} starts either from an element of $R$ or from an item of $S \setminus R$.
Let \Cref{alg:region-growing-separation} start sampling from $R$.
The fact that $\mathcal{R} = \{ S \}$ implies that $\exists pq \in \delta(R)$ such that $c_{pq} < 0$ or $\exists pqr \in T_{\delta(R)}$ such that $c_{pqr} < 0$, by definition of \Cref{alg:region-growing-separation}.
However, this is in contradiction with the assumption that $R$ satisfies \eqref{eq:edge-cut-condition-1}--\eqref{eq:edge-cut-condition-2}.
Since the second scenario is symmetrical, we again reach a contradiction 
by applying an analogous reasoning.
Therefore, we have shown that if there exists a non-trivial subset of $S$ that fulfills \eqref{eq:edge-cut-condition-1}--\eqref{eq:edge-cut-condition-2}, then \Cref{alg:region-growing-separation} finds such a subset.
\end{proof}
}{}

\begin{breakablealgorithm}
	\caption{Region Growing}
	\label{alg:region-growing-separation}
 	\begin{algorithmic}[1]
	\State {\bfseries Input:} $S\neq \emptyset$, $c: \mathcal{I}_S \to \mathbb{R}$
	\State {\bfseries Initialize:} $\mathcal{R} = \{\}$, queue $Q = S$
	\Repeat
	\State $p := Q.pop$
	\State $R = \{p\}$
	\State Initialize $noChange = true$.
	\Repeat
	\State Set $noChange = true$
		\If{$\exists pq\in \delta(R): c_{pq} < 0$}
		\State $R := R \cup \{p, q\}$
		\State remove $p, q$ from $Q$
		\State set $noChange = false$
		\EndIf
		\If{$\exists pqr\in T_{\delta(R)}: c_{pqr} < 0$}
		\State $R := R \cup \{p, q, r\}$
		\State remove $p, q, r$ from $Q$
		\State set $noChange = false$
		\EndIf
	\Until{$noChange$ is $true$}
	\State Add $R$ to $\mathcal{R}$	
	\Until{$Q = \emptyset$}
	\end{algorithmic}
\end{breakablealgorithm}

Partial optimality according to Propositions~\ref{proposition:edge-cut-persistency}--\ref{lemma:persistency-triangle-edge-join} is conditional to the existence of a pair $ij \in \tbinom S2$ or triple $ijk \in \tbinom S3$, together with a subset $R \subseteq S$, and in case of \Cref{lemma:persistency-triangle-edge-join} a second subset $R' \subseteq S$ independent of $R$, such that specific inequalities are satisfied, namely \eqref{eq:assumption-edge-cut-inequality}--\eqref{eq:triangle-edge-join-3}.
For every triple, we test \eqref{eq:triangle-edge-join-3} explicitly, in quadratic time.
For every pair or triple, we reduce the search for subsets $R$ or $R'$ that satisfy \eqref{eq:assumption-edge-cut-inequality}--\eqref{eq:triangle-edge-join-2} \emph{with maximum margin} to the min $st$-cut problem (in Section~\ref{sec:tech-min-cut}).
In order to test for partial optimality efficiently, we solve the dual max $st$-flow problems by the implementation in the C++ library \citet{boost} of the push-relabel algorithm of \citet{goldberg-1988}.

As mentioned already in \Cref{section:partial-optimality-criteria-joins},
we are unaware of an efficient method for finding subsets that satisfy the conditions of \Cref{lemma:general-subgraph-edge-join} or \ref{proposition:subset-join-proposition}.
Regarding \Cref{lemma:general-subgraph-edge-join}, we resort to the special case of \Cref{corollary:edge-subgraph-edge-join} that we test for each pair in quadratic time, and to the special case of \Cref{corollary:triplet-subgraph-edge-join} that we test for each triple in quadratic time.
Regarding \Cref{proposition:subset-join-proposition}, we employ 
the special case of \Cref{corollary:subset-join-all-pairs-at-once} and search heuristically for a witness $R$ of \eqref{eq:subset-join-corollary-all-pairs-condition}, as follows.
In an outer loop, we iterate over all pairs $R = \{i, j\}$ with $c_{ij} \leq 0$.
For each of these initializations of $R$, we add elements to $R$ for which the costs of all pairs and triples inside $R$ is non-positive, greedily considering ones for which the costs of newly considered pairs and triples is minimal.
Upon termination of this inner loop, we take $R$ to be a candidate.
By construction of $R$, all coefficients on the left-hand side of \eqref{eq:subset-join-corollary-all-pairs-condition} are non-positive. 
By applying \Cref{prop: reduction costs} to the left-hand side of \eqref{eq:subset-join-corollary-all-pairs-condition}, this problem takes the form of a min cut problem with non-negative capacities that we solve exactly using the implementation in the C++ library \citet{boost} of the algorithm by \citet{stoer-1997}.

\subsection{Reductions to Minimum Cut Problems}\label{sec:tech-min-cut}
Here, we discuss how, for a given pair or triple, we reduce the search for subsets $R, R' \subseteq S$ that satisfy \eqref{eq:assumption-edge-cut-inequality}--\eqref{eq:triangle-edge-join-2} maximally to the min $st$-cut problem.
In any of these cases, we have $S \neq \emptyset$ and $c\in \mathbb{R}^{\mathcal{I}_S}$ such that $c_{pqr} \geq 0$ for all $pqr\in \tbinom S3$, and $c_{pq}\geq 0$ for all $pq\in \tbinom S2$. 
Moreover, we have $i\in S$, a pair or triple $\{i\} \cup S_0 \subseteq S$ and a problem of the form
\begin{equation}
	\min_{\substack{R \subseteq S \colon \\ i\in R,\\ j\not\in R, \forall j \in S_0}}\quad \sum_{pqr\in T_{\delta(R)}}c_{pqr} + \sum_{pq\in \delta(R)}c_{pq} 
	\enspace .
	\label{eq:verification-problem-unified}
\end{equation}
To begin with, we move costs of triples to costs of pairs:
\begin{proposition}\label{prop: reduction costs}
	\label{eq:cubic-st-cut-reduction-min-st-cut}
	Let $S \neq \emptyset$, $c\in \mathbb{R}^{\mathcal{I}_S}$ and $R\subseteq S$. Then
	\begin{align} 
	\sum_{pqr\in T_{\delta(R)}} c_{pqr} = \frac{1}{2}\sum_{pq\in \delta(R)}\sum_{r\in S\setminus \{p, q\}}c_{pqr}
	\enspace .
	\end{align}
\end{proposition} 

\ifthenelse{\boolean{proofs}}{
\begin{proof}
Let $R \subseteq S$.
Observe that
\begin{align} \label{eq:identity-sum-all-TdeltaR-triplets}
\sum_{pqr\in T_{\delta(R)}}c_{pqr} & = \sum_{pq\in \tbinom R2}\sum_{r\in S \setminus R}c_{pqr} + \sum_{pq\in \tbinom{S\setminus R}{2}}\sum_{r\in R} c_{pqr} = \frac 12 \sum_{p\in R}\sum_{q\in R \setminus \{p\}}\sum_{r\in S\setminus R}c_{pqr} \\
& + \frac 12 \sum_{p\in S \setminus R}\sum_{q\in S \setminus \left(R \cup \{p\}\right)}\sum_{r\in R}c_{pqr} \\
& = \frac 12 \sum_{p\in R}\sum_{q\in S \setminus R}\left(\sum_{r\in R \setminus \{p\}}c_{pqr} + \sum_{r\in S \setminus \left(R \cup \{q\}\right)}c_{pqr}\right) = \frac 12 \sum_{pq\in \delta(R)}\sum_{r\in S \setminus \{p, q\}}c_{pqr}.
\end{align}%
\end{proof}
}{}

\noindent
Consequently, \eqref{eq:verification-problem-unified} is equivalent to 
\begin{equation}
	\label{proposition:cubic-to-normal-st-cut-reduction}
	\min_{\substack{R \subseteq S \colon \\ i\in R,\\ j\not\in R, \forall j \in S_0}}\quad \sum_{pq\in \delta(R)}c'_{pq},
\end{equation}
where $c'_{pq} = c_{pq} + \frac{1}{2}\sum_{r\in S\setminus \{p, q\}}c_{pqr}$ for all $pq\in \tbinom S2$.
Next, we reduce \eqref{proposition:cubic-to-normal-st-cut-reduction} to a \emph{quadratic unconstrained binary optimization problem}, by applying the following proposition:
\begin{proposition}
	\label{lemma:qpbo-translation}
	Let $S \neq \emptyset$ and $c\in \mathbb{R}^{\tbinom S2}$. 
	Moreover, let $i\in S$ and $S_0 \subseteq S\setminus \{i\}$. 
	Furthermore, let $S' = S \setminus \left(S_0 \cup \{i\}\right)$ and $c' \colon \tbinom{S'}{2}\cup S'\cup \{\emptyset\} \to \mathbb{R}$ such that
	\begin{align}
		c'_p & = \sum_{q\in S \setminus \{p\}}c_{pq} - 2c_{pi} & \forall p \in S'
		\\
		c'_{pq} & = -2c_{pq} & \forall pq \in \tbinom{S'}{2}
		\\
		c'_\emptyset & = \sum_{q\in S \setminus \{i\}}c_{qi}
		\enspace .
		\end{align}
	Then:
	\begin{equation}
		\min_{\substack{R \subseteq S \colon \\ i\in R, \\ j \not\in R, \forall j \in S_0}} \sum_{pq\in \delta(R)} c_{pq} 
		= \min_{y\in \{0, 1\}^{S'}} 
		\sum_{pq\in \tbinom{S'}{2}} c'_{pq}y_p y_q + \sum_{p\in S'}c'_p y_p + c'_\emptyset
		\enspace .
		\label{eq:node-encoding-equality-2} 
	\end{equation}
\end{proposition}

\ifthenelse{\boolean{proofs}}{
\begin{proof}
Let $R \subseteq S$ such that $i\in R$ and $\forall j \in S_0 \colon j\not \in R$. 
We define $y\in \{0, 1\}^S$ such that $y = \mathbbm{1}_R$. 
Then, we have $y_i = 1$ and $\forall j \in S_0 \colon y_j = 0$. 
Moreover, it follows 
\begin{align}
\sum_{pq\in \delta(R)} c_{pq} & = \sum_{pq\in \tbinom S2}c_{pq}\left(y_p (1-y_q) + y_q (1-y_p)\right) = \sum_{pq\in \tbinom S2}c_{pq}\left(y_p + y_q - 2y_p y_q\right) \\
& = -2\sum_{pq\in \tbinom S2}c_{pq}y_p y_q + \sum_{\substack{p,q \in S \\ p \neq q}} c_{pq}y_p \\
& = -2 \sum_{pq\in \tbinom{S'}{2}}c_{pq}y_p y_q - 2\sum_{p\in S'}c_{pi} y_p + \sum_{p\in S'}\sum_{q\in S \setminus \{p\}}c_{pq}y_p + \sum_{q\in S \setminus \{i\}}c_{qi}\\
& = -2\sum_{pq\in \tbinom{S'}{2}}c_{pq}y_p y_q + \sum_{p\in S'}\left(-2c_{pi} + \sum_{q\in S \setminus \{p\}}c_{pq}\right)y_p + \sum_{q\in S \setminus \{i\}}c_{qi} \\
& =\sum_{pq\in \tbinom{S'}{2}}c'_{pq}y_p y_q + \sum_{p\in S'} c'_p y_p + c'_\emptyset
\enspace .
\end{align}
This concludes the proof. 
\end{proof}
}{}

For the instances of \eqref{eq:node-encoding-equality-2} that arise from testing  \eqref{eq:assumption-edge-cut-inequality}--\eqref{eq:triangle-edge-join-2}, we have $\forall pq\in \tbinom{S'}{2} \colon c'_{pq}\leq 0$.
Thus, the right-hand side of \eqref{eq:node-encoding-equality-2} is \emph{submodular} and can be minimized in strongly polynomial time \citep{boros-2008, kolmogorov-2004}.
For completeness, we describe the reduction to an instance of min $st$-cut in detail in \Cref{appendix:qpbo-to-min-st-cut-propositions}.

\section{Combining Partial Optimality Conditions}\label{sec: combining}

Next, we discuss how we apply partial optimality conditions iteratively and why this requires special attention. 
Let $S \neq \emptyset$ and $c\in \mathbb{R}^{\mathcal{I}_S}$. 
Furthermore, let $Q_1, Q_2 \subseteq \cp_{S}$. 
If there is an optimal solution $x^*_1\in \cp_S$ to \csp{S}{c} such that $x^*_1\in Q_1$, and an optimal solution $x^*_{2}\in X_S$ to \csp{S}{c} such that $x^*_2\in Q_2$, then there is not necessarily an optimal solution $x^*\in \cp_S$ to \csp{S}{c} such that $x^*\in Q_1 \cap Q_2$. 
For example, consider $S = \{1, 2, 3\}$ and $c \in \mathbb{R}^{\mathcal{I}_S}$ such that $c_{123} = 5$, $c_{12}=c_{13}=c_{23} = -2$ and $c_\emptyset = 0$.
Then, $\min_{x\in \cp_S}\phi_c(x) = -2$.  
Furthermore, let $Q_1 = \{x\in \cp_S \mid x_{12} = 1\}$ and $Q_2 = \{x\in \cp_S \mid x_{13} = 1\}$. 
It follows that $Q_1 \cap Q_2 = \{ x \in \cp_S \mid x_{12} = 1, x_{13} = 1 \} = \{ (1, 1, 1) \}$.
The set of optimal solutions is the set of all $x\in \cp_S$ for which there is exactly one $pq\in \tbinom S2$ with $x_{pq} = 1$. 
Thus, the feasible $x'\in \cp_S$ such that $x'_{12} = x'_{13} = x'_{23} = 1$ is not optimal.

\subsection{Cut Conditions}
\label{section:cut-conditions-parallel}

For any $T_0 \subseteq \tbinom S3$, $P_0 \subseteq \tbinom S2$, we define $\cp_S \vert_{T_0, P_0}\subseteq \cp_S$ such that $x \in \cp_S \vert_{T_0, P_0}$ if and only if $\forall pqr\in T_0 \colon x_{pq}x_{pr}x_{qr} = 0$ and $\forall pq\in P_0 \colon x_{pq} = 0$.
For any $R \subseteq S$, we have that $\sigma_{\delta(R)}$ restricted to $\cp_S \vert_{T_0, P_0}$ has image $\cp_S \vert_{T_0, P_0}$. 
All our cut results use either $\sigma_{\delta(R)}$ for some $R \subseteq S$ or the identity. 
Furthermore, the cut conditions do not change when applied to the restricted set $\cp_S \vert_{T_0, P_0}$. 
Therefore, we can apply all our cut conditions simultaneously.
On the contrary, $\sigma_R$ for some $R \subseteq S$ restricted to $\cp_S \vert_{T_0, P_0}$ does not necessarily have image $\cp_S \vert_{T_0, P_0}$, e.g. for the map $\sigma_S$ we have $\sigma_S(\cp_S \vert_{T_0, P_0}) = \{\mathbbm{1}_S\}$. 
Therefore, we cannot expect that partial optimality statements would hold when applying any join condition together with any another condition.

\subsection{Join Conditions}
Let us assume the existence of an optimal solution $x^*$ to \csp{S}{c} such that $x^*_{ij} = 1$ for some $ij\in \tbinom S2$. 
We define $\cp_S \vert_{x_{ij} = 1} = \{x\in \cp_S \mid x_{ij} = 1\}$.
Then, we have 
\begin{equation}
	\label{eq:one-persistency-minimization-problem}
	\min_{x\in \cp_S} \phi_c(x) = \min_{x \in \cp_S \vert_{x_{ij} = 1}}\phi_c(x)
	\enspace .
\end{equation}
Let $S' = S \setminus \{j\}$. 
Now, we relate feasible vectors of $\cp_S \vert _{x_{ij} = 1}$ to feasible vectors of $\cp_{S'}$. 
We observe that for any $x \in \cp_S \vert_{x_{ij} = 1}$ we have $\forall p\in S \setminus \{i, j\} \colon x_{pi} = x_{pj}$. 
We define $\varphi_{ij}\colon \cp_S \vert_{x_{ij} = 1} \to \cp_{S'}$ as 
\begin{align}
	\varphi_{ij}(x)_{pi} & = x_{pi} & \forall x \in \cp_S\vert_{x_{ij} = 1} \quad \forall p\in S'\setminus \{i\} 
	\phantom{\enspace ..}
	\\
	\varphi_{ij}(x)_{pq} & = x_{pq} & \forall x\in \cp_S\vert_{x_{ij} = 1} \quad \forall pq \in \tbinom{S'\setminus \{i\}}{2} \enspace .
\end{align}
It is easy to see that $\varphi_{ij}$ is bijective.
Proposition~\ref{lemma:contraction-cost-adjustments} below shows that solving the right-hand side of \eqref{eq:one-persistency-minimization-problem} is equivalent to solving a smaller instance of the original problem. 
\begin{proposition}
	\label{lemma:contraction-cost-adjustments}
	Let $S \neq \emptyset$ and $c\in \mathbb{R}^{\mathcal{I}_S}$. 
	Moreover, let $ij \in \tbinom S2$ and $S' = S \setminus \{j\}$, and let $c'\in \mathbb{R}^{\mathcal{I}_{S'}}$ such that
	\begin{align}
	c'_{pqr} &= c_{pqr} & \forall pqr \in \tbinom{S'\setminus \{i\}}{3} \\
	c'_{pqi} &= c_{pqi} + c_{pqj} & \forall pq \in \tbinom{S'\setminus \{i\}}{2} \\
	c'_{pq} &= c_{pq} & \forall pq \in \tbinom{S'\setminus \{i\}}{2} \\
	c'_{pi} &= c_{pi} + c_{pj} + c_{pij} & \forall p \in S'\setminus \{i\} \\
	c'_\emptyset & = c_0 + c_{ij}
	\enspace .
	\end{align}
	Furthermore, let $\varphi_{ij}\colon \cp_S\vert_{x_{ij} = 1}\to \cp_{S'}$ be the map that relates feasible vectors of $\cp_S\vert_{x_{ij} = 1}$ to feasible vectors of $\cp_{S'}$.
	Then, we have
	\begin{align}
	\min_{x\in \cp_S\vert_{x_{ij} = 1}} \phi_c(x) = \min_{x\in \cp_{S'}}\phi_{c'}(x)
	\enspace .
	\end{align}
	Moreover, if $x^* \in \argmin_{x\in \cp_S\vert_{x_{ij} = 1}}\phi_c(x)$, then $\varphi_{ij}(x^*)\in \argmin_{x\in \cp_{S'}}\phi_{c'}(x)$.
\end{proposition}

\ifthenelse{\boolean{proofs}}{
\begin{proof}
Let $x\in \cp_S \vert_{x_{ij} = 1}$. 
We show that $\phi_c(x) = \phi_{c'}(\varphi_{ij}(x))$. 
Let $x' = \varphi_{ij}(x)$. 
We use the fact that $x_{pi} = x_{pj}$ for all $p\in S \setminus \{i, j\}$, and $x_{ij} = 1$. 
It follows 
\begin{align}
\phi_{c'}(x') & = \sum_{pqr\in \tbinom{S'}{3}}c'_{pqr}x'_{pq}x'_{pr}x'_{qr} + \sum_{pq\in \tbinom S2} c'_{pq}x'_{pq} + c'_{\emptyset}\\
& = \sum_{pq\in \tbinom{S \setminus\{i, j\}}{2}}c'_{pqi}x'_{pi}x'_{qi}x'_{pq} + \sum_{pqr\in \tbinom{S \setminus\{i, j\}}{3}}c'_{pqr}x'_{pq}x'_{pr}x'_{qr} + \sum_{p\in S \setminus\{i, j\}}c'_{pi}x'_{pi} \\
& \qquad + \sum_{pq\in \tbinom{V\setminus\{i, j\}}{2}}c'_{pq}x'_{pq} + c'_{\emptyset}\\
& = \sum_{pq\in \tbinom{S \setminus\{i, j\}}{2}}(c_{pqi} + c_{pqj})x_{pi}x_{qi}x_{pq} + \sum_{pqr\in \tbinom{S \setminus\{i, j\}}{3}}c_{pqr}x_{pq}x_{pr}x_{qr} + \sum_{pq\in \tbinom{S \setminus\{i, j\}}{2}}c_{pq}x_{pq} \\
& \qquad + \sum_{p\in S \setminus\{i, j\}}(c_{pi} + c_{pj} + c_{pij})x_{pi}+ c'_{\emptyset}\\
&= \sum_{pq\in \tbinom{S \setminus\{i, j\}}{2}}c_{pqi}x_{pi}x_{qi}x_{pq} + \sum_{pq\in \tbinom{S \setminus\{i, j\}}{2}}c_{pqj}x_{pj}x_{qj}x_{pq} + \sum_{pqr\in \tbinom{S \setminus\{i, j\}}{3}}c_{pqr}x_{pq}x_{pr}x_{qr} \\
& \qquad + \sum_{p\in V}c_{pij}x_{ij}x_{pi}x_{pj} + \sum_{pq\in \tbinom{S \setminus\{i, j\}}{2}}c_{pq}x_{pq} + \sum_{p\in V\setminus\{i, j\}}c_{pi}x_{pi} + \sum_{p\in S \setminus\{i, j\}}c_{pj}x_{pj} + c_{ij}x_{ij} + c_{\emptyset} \\ 
&= \sum_{pqr\in \tbinom S3}c_{pqr}x_{pq}x_{pr}x_{qr} + \sum_{pq\in S}c_{pq}x_{pq} + c_{\emptyset} = \phi_c(x).
\end{align}
Therefore, we have 
\begin{equation}
\min_{x\in \cp_S \vert_{x_{ij} = 1}} \phi_c(x) = \min_{x\in \cp_S \vert_{x_{ij} = 1}} \phi_{c'}(\varphi_{ij}(x)) = \min_{x\in \cp_{S'}}\phi_{c'}(x)
\end{equation}
This concludes the proof.
\end{proof}
}{}

\subsection{Mixing Cut and Join Conditions}
\label{section:mixing-conditions}
Here, we describe how we apply the partial optimality properties recursively. As soon as a condition leads to a smaller instance, we start the procedure again on the smaller set (in case of a join) or sets (in case of a cut). 
Firstly, we apply Proposition~\ref{lemma:persistency-subset-separation}, which leaves us with independent sub-problems. 
Secondly, we apply our join conditions until we find a pair $ij\in \tbinom S2$ or triplet $ijk \in \tbinom S3$ to join, starting from \Cref{corollary:subset-join-all-pairs-at-once} and then moving on to \Cref{lemma:edge-join-persistency,lemma:persistency-triangle-edge-join}, \Cref{corollary:edge-subgraph-edge-join,corollary:triplet-subgraph-edge-join} and \Cref{proposition:triplet-join}, in this order. 
Thirdly, we apply the remaining cut conditions, which can be applied jointly, as we have seen in \Cref{section:cut-conditions-parallel}.
We remark, that the order in which we apply our join conditions is arbitrary, and we do not claim it to be optimal.

\section{Numerical Experiments}\label{sec: experiments}

\begin{figure}[t]
	\centering
	\begin{minipage}{0.33\linewidth}
		a)
		\hspace{-2ex}
		\vspace{-2ex}
		\linebreak
		\begin{tikzpicture}\small
			\pgfplotsset{
				width=\linewidth,
				height=\linewidth	
			}
			\begin{axis}[
				xlabel=$\alpha$,
				xmin=0.25,
				xmax=0.7,
				ymin=0.0,
				ymax=100,
				ylabel=Variables \%,
				legend pos=north east,
				legend style={
					nodes={
						scale=0.55
					}
				}
				]
				\addplot+[
				color=primary,
				only marks,
				mark=*,
				mark options={scale=0.5,fill}
				]
				table[
				x=alpha,
				y=median,
				y error plus expr=\thisrow{q75} - \thisrow{median},
				y error minus expr=\thisrow{median} - \thisrow{q25},
				col sep=comma
				]
				{./data-partition-experiment_vars_beta=0.0.csv};
				\addlegendentry{$\beta = 0.0$};
				\addplot+[
				color=tertiary,
				only marks,
				mark=triangle*,
				mark options={scale=0.9,fill}
				]
				table[
				x=alpha,
				y=median,
				y error plus expr=\thisrow{q75} - \thisrow{median},
				y error minus expr=\thisrow{median} - \thisrow{q25},
				col sep=comma
				]
				{./data-partition-experiment_vars_beta=0.5.csv};
				\addlegendentry{$\beta = 0.5$};
				\addplot+[
				color=secondary,
				only marks,
				mark=square*,
				mark options={scale=0.5,fill}
				]
				table[
				x=alpha,
				y=median,
				y error plus expr=\thisrow{q75} - \thisrow{median},
				y error minus expr=\thisrow{median} - \thisrow{q25},
				col sep=comma
				]
				{./data-partition-experiment_vars_beta=1.0.csv};
				\addlegendentry{$\beta = 1.0$};
				\addplot+[
				color=deep-red,
				only marks,
				mark=diamond*,
				mark options={scale=0.7,fill}
				]
				table[
				x=alpha,
				y=median,
				y error plus expr=\thisrow{q75} - \thisrow{median},
				y error minus expr=\thisrow{median} - \thisrow{q25},
				col sep=comma
				]
				{./data-partition-experiment_vars_beta=0.01.csv};
				\addlegendentry{$\beta = 0.01$};
			\end{axis}
		\end{tikzpicture}%
	\end{minipage}%
	\begin{minipage}{0.33\linewidth}
		b)
		\hspace{-2ex}
		\vspace{-2ex}
		\linebreak	
		\begin{tikzpicture}\small
			\pgfplotsset{
				width=\linewidth,
				height=\linewidth	
			}
			\begin{axis}[
				xlabel=$\alpha$,
				xmin=0.25,
				xmax=0.7,
				ymin=0.0,
				ylabel=Runtime / s,
				legend pos=north west,
				legend style={
					nodes={
						scale=0.55
					}
				}
				]
				\addplot+[
				color=primary,
				only marks,
				mark=*,
				mark options={scale=0.5,fill}
				]
				table[
				x=alpha,
				y expr=\thisrow{median} / 1000,
				y error plus expr=(\thisrow{q75} - \thisrow{median}) / 1000,
				y error minus expr=(\thisrow{median} - \thisrow{q25}) / 1000,
				col sep=comma
				]
				{./data-partition-experiment_runtimes_beta=0.0.csv};
				\addlegendentry{$\beta = 0.0$};
				\addplot+[
				color=tertiary,
				only marks,
				mark=triangle*,
				mark options={scale=0.9,fill}
				]
				table[
				x=alpha,
				y expr=\thisrow{median} / 1000,
				y error plus expr=(\thisrow{q75} - \thisrow{median}) / 1000,
				y error minus expr=(\thisrow{median} - \thisrow{q25}) / 1000,
				col sep=comma
				]
				{./data-partition-experiment_runtimes_beta=0.5.csv};
				\addlegendentry{$\beta = 0.5$};
				\addplot+[
				color=secondary,
				only marks,
				mark=square*,
				mark options={scale=0.5,fill}
				]
				table[
				x=alpha,
				y expr=\thisrow{median} / 1000,
				y error plus expr=(\thisrow{q75} - \thisrow{median}) / 1000,
				y error minus expr=(\thisrow{median} - \thisrow{q25}) / 1000,
				col sep=comma
				]
				{./data-partition-experiment_runtimes_beta=1.0.csv};
				\addlegendentry{$\beta = 1.0$};
				\addplot+[
				color=deep-red,
				only marks,
				mark=diamond*,
				mark options={scale=0.7,fill}
				]
				table[
				x=alpha,
				y expr=\thisrow{median} / 1000,
				y error plus expr=(\thisrow{q75} - \thisrow{median}) / 1000,
				y error minus expr=(\thisrow{median} - \thisrow{q25}) / 1000,
				col sep=comma
				]
				{./data-partition-experiment_runtimes_beta=0.01.csv};
				\addlegendentry{$\beta = 0.01$};			
			\end{axis}
		\end{tikzpicture}%
	\end{minipage}\linebreak\hfill
	\begin{minipage}{0.33\linewidth}
		c)
		\hspace{-2ex}
		\vspace{-2ex}
		\linebreak	
		\begin{tikzpicture}\small
			\pgfplotsset{
				width=\linewidth,
				height=\linewidth	
			}
			\begin{axis}[
				xlabel=$\sigma$,
				xmin=0.0,
				xmax=0.1,
				xtick distance=0.05,
				x tick label style={
					/pgf/number format/.cd,
					fixed
				},
				ymin=0.0,
				ymax=100,
				ylabel=Variables \%,
				legend pos=north east
				]
				\addplot+[
				only marks,
				mark=*,
				mark options={scale=0.5,fill},
				color=primary
				]
				table[
				x=sigma,
				y=median,
				y error plus expr=\thisrow{q75} - \thisrow{median},
				y error minus expr=\thisrow{median} - \thisrow{q25},
				col sep=comma
				]
				{./data-triangles-experiment_equilateral_triangles_all_vars.csv};
			\end{axis}
		\end{tikzpicture}%
	\end{minipage}%
	\begin{minipage}{0.33\linewidth}
		d)
		\hspace{-2ex}
		\vspace{-2ex}
		\linebreak	
		\begin{tikzpicture}\small
			\pgfplotsset{
				width=\linewidth,
				height=\linewidth	
			}
			\begin{axis}[
				xlabel=$\sigma$,
				xmin=0.0,
				xmax=0.1,
				xtick distance=0.05,
				x tick label style={
					/pgf/number format/.cd,
					fixed
				},
				ymin=0.0,
				ylabel=Runtime / s,
				legend pos=north east
				]
				\addplot+[
				only marks,
				mark=*,
				mark options={scale=0.5,fill},
				color=primary
				]
				table[
				x=sigma,
				y expr=\thisrow{median} / 1000,
				y error plus expr=(\thisrow{q75} - \thisrow{median}) / 1000,
				y error minus expr=(\thisrow{median} - \thisrow{q25}) / 1000,
				col sep=comma
				]
				{./data-triangles-experiment_runtimes.csv};
			\end{axis}
		\end{tikzpicture}%
	\end{minipage}%
	\caption{We report above the percentage of fixed variables and triples after applying all conditions jointly, as described in \ref{section:mixing-conditions}, as well as the corresponding runtimes. 
	(a) and (b) show these for the partition dataset with respect to the parameters $\alpha$ and $\beta$ and for 48 elements. 
	(c) and (d) show these for the geometric dataset with respect to the parameter $\sigma$ and for 45 points.}
	\label{fig:partition-all-criteria-results}
\end{figure}
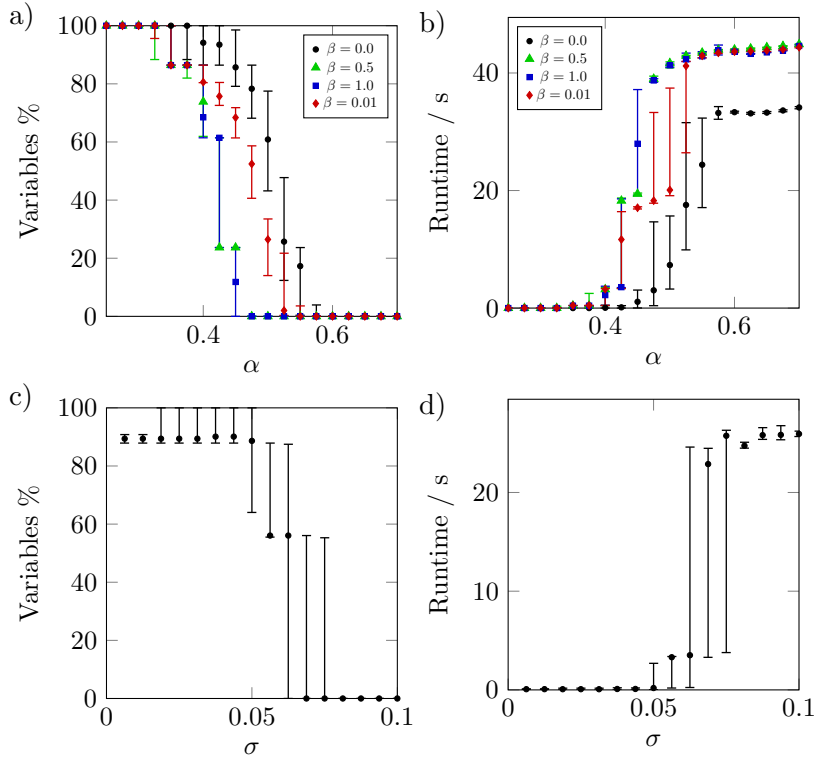

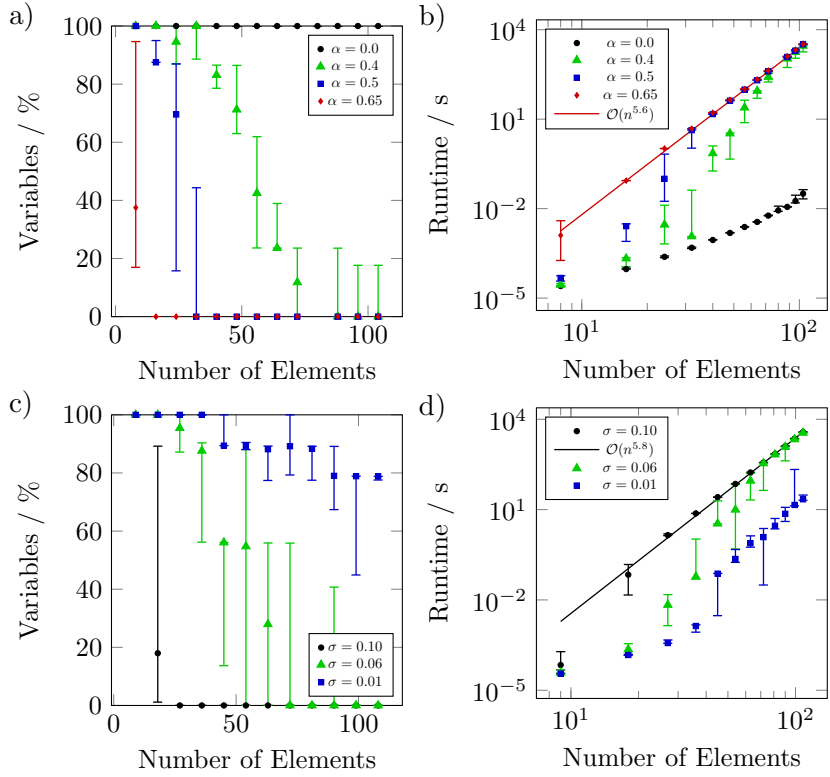
\begin{figure}[t]
	\centering
	\begin{minipage}{0.33\linewidth}
		a)
		\hspace{-2ex}
		\vspace{-2ex}
		\linebreak
		\begin{tikzpicture}\small
			\pgfplotsset{
				width=\linewidth,
				height=\linewidth	
			}
			\begin{axis}[
				xlabel=Number of Elements,
				ylabel=Variables / \%,
				legend pos=north east,
				legend style={
					nodes={
						scale=0.55
					}
				},
				ymin=0,
				ymax=100
				]
				\addplot+[
				only marks, 
				color=primary, 
				mark=*,
				mark options={scale=0.5,fill}
				]
				table[
				x=n, 
				y expr=\thisrow{median},
				y error plus expr = (\thisrow{q75} - \thisrow{median}),
				y error minus expr = (\thisrow{median} - \thisrow{q25}),
				col sep=comma
				]{data-partition-variables_alpha_0.0_beta_0.5.csv};
				\addlegendentry{$\alpha = 0.0$}
				\addplot+[
				only marks, 
				color=tertiary, 
				mark=triangle*, 
				mark options={scale=0.9,fill}
				]
				table[
				x=n, 
				y expr=\thisrow{median},
				y error plus expr = (\thisrow{q75} - \thisrow{median}) ,
				y error minus expr = (\thisrow{median} - \thisrow{q25}),
				col sep=comma
				]{data-partition-variables_alpha_0.4_beta_0.5.csv};
				\addlegendentry{$\alpha=0.4$}
				\addplot+[
				only marks, 
				mark=square*, 
				color=secondary, 
				mark options={scale=0.5,fill}
				]
				table[
				x=n, 
				y expr=\thisrow{median},
				y error plus expr = (\thisrow{q75} - \thisrow{median}),
				y error minus expr = (\thisrow{median} - \thisrow{q25}),
				col sep=comma
				]{data-partition-variables_alpha_0.5_beta_0.5.csv};
				\addlegendentry{$\alpha=0.5$}
				\addplot+[
				only marks, 
				mark=diamond*, 
				color=deep-red, 
				mark options={scale=0.5,fill}
				]
				table[
				x=n, 
				y expr=\thisrow{median},
				y error plus expr = (\thisrow{q75} - \thisrow{median}),
				y error minus expr = (\thisrow{median} - \thisrow{q25}),
				col sep=comma
				]{data-partition-variables_alpha_0.65_beta_0.5.csv};
				\addlegendentry{$\alpha=0.65$}
			\end{axis}
		\end{tikzpicture}%
	\end{minipage}%
	\begin{minipage}{0.33\linewidth}
		b)
		\hspace{-2ex}
		\vspace{-2ex}
		\linebreak	
		\begin{tikzpicture}\small
			\pgfplotsset{
				width=\linewidth,
				height=\linewidth	
			}
			\begin{loglogaxis}[
				xlabel=Number of Elements,
				ylabel=Runtime / s,
				legend pos=north west,
				legend style={
					nodes={
						scale=0.55
					}
				}
				]
				\addplot+[
				only marks, 
				color=primary, 
				mark=*,
				mark options={scale=0.5,fill}
				]
				table[
				x=n, 
				y expr=\thisrow{median} / 1000,
				y error plus expr = (\thisrow{q75} - \thisrow{median}) / 1000,
				y error minus expr = (\thisrow{median} - \thisrow{q25}) / 1000,
				col sep=comma
				]{data-partition-runtimes_alpha_0.0_beta_0.5.csv};
				\addlegendentry{$\alpha=0.0$}
				\addplot+[
				only marks, 
				color=tertiary, 
				mark=triangle*, 
				mark options={scale=0.9,fill}
				]
				table[
				x=n, 
				y expr=\thisrow{median} / 1000,
				y error plus expr = (\thisrow{q75} - \thisrow{median}) / 1000,
				y error minus expr = (\thisrow{median} - \thisrow{q25}) / 1000,
				col sep=comma
				]{data-partition-runtimes_alpha_0.4_beta_0.5.csv};
				\addlegendentry{$\alpha=0.4$}
				\addplot+[
				only marks, 
				mark=square*, 
				color=secondary, 
				mark options={scale=0.5,fill}
				]
				table[
				x=n, 
				y expr=\thisrow{median} / 1000,
				y error plus expr = (\thisrow{q75} - \thisrow{median}) / 1000,
				y error minus expr = (\thisrow{median} - \thisrow{q25}) / 1000,
				col sep=comma
				]{data-partition-runtimes_alpha_0.5_beta_0.5.csv};
				\addlegendentry{$\alpha=0.5$}
				\addplot+[
				only marks, 
				mark=diamond*, 
				color=deep-red, 
				mark options={scale=0.5,fill}
				]
				table[
				x=n, 
				y expr=\thisrow{median} / 1000,
				y error plus expr = (\thisrow{q75} - \thisrow{median}) / 1000,
				y error minus expr = (\thisrow{median} - \thisrow{q25}) / 1000,
				col sep=comma
				]{data-partition-runtimes_alpha_0.65_beta_0.5.csv};
				\addlegendentry{$\alpha=0.65$};
				\addplot+[
					color=deep-red,
					no marks,
					line width=0.5pt,
				]
				table[
					x=n, 
					y expr=1.57*10^(-8)*(\thisrow{n})^(5.6),
					col sep=comma
				]{data-partition-runtimes_alpha_0.65_beta_0.5.csv};
				\addlegendentry{$\mathcal{O}(n^{5.6})$}
			\end{loglogaxis}	
		\end{tikzpicture}%
	\end{minipage}\linebreak\hfill
	\begin{minipage}{0.33\linewidth}
		c)
		\hspace{-2ex}
		\vspace{-2ex}
		\linebreak	
		\begin{tikzpicture}\small
			\pgfplotsset{
				width=\linewidth,
				height=\linewidth	
			}
			\begin{axis}[
				xlabel=Number of Elements,
				ylabel=Variables / \%,
				legend pos=south east,
				legend style={
					nodes={
						scale=0.55
					}
				},
				ymin=0,
				ymax=100
				]
				\addplot+[
				only marks, 
				color=primary, 
				mark=*,
				mark options={scale=0.5,fill}
				]
				table[
				x=n, 
				y expr=\thisrow{median},
				y error plus expr = (\thisrow{q75} - \thisrow{median}),
				y error minus expr = (\thisrow{median} - \thisrow{q25}),
				col sep=comma
				]{data-triangles-variables_sigma_0.1.csv};
				\addlegendentry{$\sigma=0.10$}
				\addplot+[
				only marks, 
				color=tertiary, 
				mark=triangle*, 
				mark options={scale=0.9,fill}
				]
				table[
				x=n, 
				y expr=\thisrow{median},
				y error plus expr = (\thisrow{q75} - \thisrow{median}) ,
				y error minus expr = (\thisrow{median} - \thisrow{q25}),
				col sep=comma
				]{data-triangles-variables_sigma_0.06.csv};
				\addlegendentry{$\sigma=0.06$}
				\addplot+[
				only marks, 
				mark=square*, 
				color=secondary, 
				mark options={scale=0.5,fill}
				]
				table[
				x=n, 
				y expr=\thisrow{median},
				y error plus expr = (\thisrow{q75} - \thisrow{median}),
				y error minus expr = (\thisrow{median} - \thisrow{q25}),
				col sep=comma
				]{data-triangles-variables_sigma_0.01.csv};
				\addlegendentry{$\sigma=0.01$}
			\end{axis}
		\end{tikzpicture}%
	\end{minipage}%
	\begin{minipage}{0.33\linewidth}
		d)
		\hspace{-2ex}
		\vspace{-2ex}
		\linebreak	
		\begin{tikzpicture}\small
			\pgfplotsset{
				width=\linewidth,
				height=\linewidth	
			}
		\begin{loglogaxis}[
			xlabel=Number of Elements,
			ylabel=Runtime / s,
			legend pos=north west,
			legend style={
				nodes={
					scale=0.55
				}
			}
			]
			\addplot+[
			only marks, 
			color=primary, 
			mark=*,
			mark options={scale=0.5,fill}
			]
			table[
			x=n, 
			y expr=\thisrow{median} / 1000,
			y error plus expr = (\thisrow{q75} - \thisrow{median}) / 1000,
			y error minus expr = (\thisrow{median} - \thisrow{q25}) / 1000,
			col sep=comma
			]{data-triangles-runtimes_sigma_0.1.csv};
			\addlegendentry{$\sigma=0.10$}
			\addplot+[
			color=primary,
			no marks,
			line width=0.5pt,
			]
			table[
			x=n, 
			y expr=5.297*10^(-9)*(\thisrow{n})^(5.8279),
			col sep=comma
			]{data-triangles-variables_sigma_0.1.csv};
			\addlegendentry{$\mathcal{O}(n^{5.8})$}
			\addplot+[
			only marks, 
			color=tertiary, 
			mark=triangle*, 
			mark options={scale=0.9,fill}
			]
			table[
			x=n, 
			y expr=\thisrow{median} / 1000,
			y error plus expr = (\thisrow{q75} - \thisrow{median}) / 1000,
			y error minus expr = (\thisrow{median} - \thisrow{q25}) / 1000,
			col sep=comma
			]{data-triangles-runtimes_sigma_0.06.csv};
			\addlegendentry{$\sigma=0.06$}
			\addplot+[
			only marks, 
			mark=square*, 
			color=secondary, 
			mark options={scale=0.5,fill}
			]
			table[
			x=n, 
			y expr=\thisrow{median} / 1000,
			y error plus expr = (\thisrow{q75} - \thisrow{median}) / 1000,
			y error minus expr = (\thisrow{median} - \thisrow{q25}) / 1000,
			col sep=comma
			]{data-triangles-runtimes_sigma_0.01.csv};
			\addlegendentry{$\sigma=0.01$}
		\end{loglogaxis}
		\end{tikzpicture}%
	\end{minipage}%
	\caption{We report above the percentage of fixed variables and triples after applying all conditions jointly, as described in \Cref{section:mixing-conditions}, as well as the corresponding runtimes. 
		(a) and (b) show these for the partition dataset with respect to the number of elements and $\beta = 0.5$. 
		(c) and (d) show these for the geometric dataset with respect to the number of elements.}
	\label{fig:persistency-runtimes-and-variables-wrt-instance-size}
\end{figure}

We examine the effect of the algorithms empirically on two datasets.
For both, we report the percentage of fixed variables and triples, as well as the runtime. 
More specifically, we report the median as well as lower and upper quartile over 30 instances.
We apply all partial optimality conditions jointly, as described in \Cref{section:mixing-conditions},
and we also evaluate the effect of each condition separately. 
All algorithms are implemented in C++ and run on one core of an Intel Core i5-6600 equipped with 16 GB of RAM.

\subsection{Partition Dataset}

\begin{figure}[t]
	\centering
	\begin{minipage}{0.33\linewidth}
		a)
		\hspace{-2ex}
		\vspace{-2ex}
		\linebreak	
		\begin{tikzpicture}\small
			\pgfplotsset{%
				width=1.0\linewidth,
				height=1.0\linewidth
			}
			\begin{axis}[
				xlabel=$\alpha$,
				xmin=0.25,
				xmax=0.7,
				ymin=0.0,
				ymax=100,
				ylabel=Variables \%,
				legend pos=north east,
				legend style={
					nodes={
						scale=0.55
					}
				},
				]
				\addplot+[
				color=primary,
				only marks,
				mark=*,
				mark options={scale=0.5,fill}
				]
				table[
				x=alpha,
				y=median,
				y error plus expr=\thisrow{q75} - \thisrow{median},
				y error minus expr=\thisrow{median} - \thisrow{q25},
				col sep=comma
				]
				{./data-partition-individual-experiment_full_n=6_findIndependentSubsets_vars_beta=0.0.csv};
				\addlegendentry{$\beta = 0.0$};
				\addplot+[
				color=tertiary,
				only marks,
				mark=triangle*,
				mark options={scale=0.9,fill}
				]
				table[
				x=alpha,
				y=median,
				y error plus expr=\thisrow{q75} - \thisrow{median},
				y error minus expr=\thisrow{median} - \thisrow{q25},
				col sep=comma
				]
				{./data-partition-individual-experiment_full_n=6_findIndependentSubsets_vars_beta=0.5.csv};
				\addlegendentry{$\beta = 0.5$};
				\addplot+[
				color=secondary,
				only marks,
				mark=square*,
				mark options={scale=0.5,fill}
				]
				table[
				x=alpha,
				y=median,
				y error plus expr=\thisrow{q75} - \thisrow{median},
				y error minus expr=\thisrow{median} - \thisrow{q25},
				col sep=comma
				]
				{./data-partition-individual-experiment_full_n=6_findIndependentSubsets_vars_beta=1.0.csv};
				\addlegendentry{$\beta = 1.0$};
				\addplot+[
				color=deep-red,
				only marks,
				mark=diamond*,
				mark options={scale=0.7,fill}
				]
				table[
				x=alpha,
				y=median,
				y error plus expr=\thisrow{q75} - \thisrow{median},
				y error minus expr=\thisrow{median} - \thisrow{q25},
				col sep=comma
				]
				{./data-partition-individual-experiment_full_n=6_findIndependentSubsets_vars_beta=0.01.csv};
				\addlegendentry{$\beta = 0.01$};
			\end{axis}
		\end{tikzpicture}%
	\end{minipage}%
	\begin{minipage}{0.33\linewidth}
		b)
		\hspace{-2ex}
		\vspace{-2ex}
		\linebreak	
		\begin{tikzpicture}\small
			\pgfplotsset{%
				width=1.0\linewidth,
				height=1.0\linewidth
			}
			\begin{axis}[
				xlabel=$\alpha$,
				xmin=0.25,
				xmax=0.7,
				ymin=0.0,
				ymax=100,
				ylabel=Variables \%,
				legend pos=north east,
				legend style={
					nodes={
						scale=0.55
					}
				},
				]
				\addplot+[
				color=primary,
				only marks,
				mark=*,
				mark options={scale=0.5,fill}
				]
				table[
				x=alpha,
				y=median,
				y error plus expr=\thisrow{q75} - \thisrow{median},
				y error minus expr=\thisrow{median} - \thisrow{q25},
				col sep=comma
				]
				{./data-partition-individual-experiment_full_n=6_edgeCut_vars_beta=0.0.csv};
				\addlegendentry{$\beta = 0.0$};
				\addplot+[
				color=tertiary,
				only marks,
				mark=triangle*,
				mark options={scale=0.9,fill}
				]
				table[
				x=alpha,
				y=median,
				y error plus expr=\thisrow{q75} - \thisrow{median},
				y error minus expr=\thisrow{median} - \thisrow{q25},
				col sep=comma
				]
				{./data-partition-individual-experiment_full_n=6_edgeCut_vars_beta=0.5.csv};
				\addlegendentry{$\beta = 0.5$};
				\addplot+[
				color=secondary,
				only marks,
				mark=square*,
				mark options={scale=0.5,fill}
				]
				table[
				x=alpha,
				y=median,
				y error plus expr=\thisrow{q75} - \thisrow{median},
				y error minus expr=\thisrow{median} - \thisrow{q25},
				col sep=comma
				]
				{./data-partition-individual-experiment_full_n=6_edgeCut_vars_beta=1.0.csv};
				\addlegendentry{$\beta = 1.0$};
				\addplot+[
				color=deep-red,
				only marks,
				mark=diamond*,
				mark options={scale=0.7,fill}
				]
				table[
				x=alpha,
				y=median,
				y error plus expr=\thisrow{q75} - \thisrow{median},
				y error minus expr=\thisrow{median} - \thisrow{q25},
				col sep=comma
				]
				{./data-partition-individual-experiment_full_n=6_edgeCut_vars_beta=0.01.csv};
				\addlegendentry{$\beta = 0.01$};
			\end{axis}
		\end{tikzpicture}%
	\end{minipage}
	\begin{minipage}{0.33\linewidth}
		c)
		\hspace{-2ex}
		\vspace{-2ex}
		\linebreak	
		\begin{tikzpicture}\small
			\pgfplotsset{%
				width=1.0\linewidth,
				height=1.0\linewidth
			}
			\begin{axis}[
				xlabel=$\alpha$,
				xmin=0.25,
				xmax=0.7,
				ymin=0.0,
				ymax=100,
				ylabel=Triples \%,
				legend pos=north east,
				legend style={
					nodes={
						scale=0.55
					}
				},
				]
				\addplot+[
				color=primary,
				only marks,
				mark=*,
				mark options={scale=0.5,fill}
				]
				table[
				x=alpha,
				y=median,
				y error plus expr=\thisrow{q75} - \thisrow{median},
				y error minus expr=\thisrow{median} - \thisrow{q25},
				col sep=comma
				]
				{./data-partition-individual-experiment_full_n=6_tripletCut_triples_beta=0.0.csv};
				\addlegendentry{$\beta = 0.0$};
				\addplot+[
				color=tertiary,
				only marks,
				mark=triangle*,
				mark options={scale=0.9,fill}
				]
				table[
				x=alpha,
				y=median,
				y error plus expr=\thisrow{q75} - \thisrow{median},
				y error minus expr=\thisrow{median} - \thisrow{q25},
				col sep=comma
				]
				{./data-partition-individual-experiment_full_n=6_tripletCut_triples_beta=0.5.csv};
				\addlegendentry{$\beta = 0.5$};
				\addplot+[
				color=secondary,
				only marks,
				mark=square*,
				mark options={scale=0.5,fill}
				]
				table[
				x=alpha,
				y=median,
				y error plus expr=\thisrow{q75} - \thisrow{median},
				y error minus expr=\thisrow{median} - \thisrow{q25},
				col sep=comma
				]
				{./data-partition-individual-experiment_full_n=6_tripletCut_triples_beta=1.0.csv};
				\addlegendentry{$\beta = 1.0$};
				\addplot+[
				color=deep-red,
				only marks,
				mark=diamond*,
				mark options={scale=0.7,fill}
				]
				table[
				x=alpha,
				y=median,
				y error plus expr=\thisrow{q75} - \thisrow{median},
				y error minus expr=\thisrow{median} - \thisrow{q25},
				col sep=comma
				]
				{./data-partition-individual-experiment_full_n=6_tripletCut_triples_beta=0.01.csv};
				\addlegendentry{$\beta = 0.01$};
			\end{axis}
		\end{tikzpicture}%
	\end{minipage}%
	\begin{minipage}{0.33\linewidth}
		d)
		\hspace{-2ex}
		\vspace{-2ex}
		\linebreak	
		\begin{tikzpicture}\small
			\pgfplotsset{%
				width=1.0\linewidth,
				height=1.0\linewidth
			}
			\begin{axis}[
				xlabel=$\alpha$,
				xmin=0.25,
				xmax=0.7,
				ymin=0.0,
				ymax=100,
				ylabel=Variables \%,
				legend pos=north east,
				legend style={
					nodes={
						scale=0.55
					}
				},
				]
				\addplot+[
				color=primary,
				only marks,
				mark=*,
				mark options={scale=0.5,fill}
				]
				table[
				x=alpha,
				y=median,
				y error plus expr=\thisrow{q75} - \thisrow{median},
				y error minus expr=\thisrow{median} - \thisrow{q25},
				col sep=comma
				]
				{./data-partition-individual-experiment_full_n=6_subsetJoin_vars_beta=0.0.csv};
				\addlegendentry{$\beta = 0.0$};
				\addplot+[
				color=tertiary,
				only marks,
				mark=triangle*,
				mark options={scale=0.9,fill}
				]
				table[
				x=alpha,
				y=median,
				y error plus expr=\thisrow{q75} - \thisrow{median},
				y error minus expr=\thisrow{median} - \thisrow{q25},
				col sep=comma
				]
				{./data-partition-individual-experiment_full_n=6_subsetJoin_vars_beta=0.5.csv};
				\addlegendentry{$\beta = 0.5$};
				\addplot+[
				color=secondary,
				only marks,
				mark=square*,
				mark options={scale=0.5,fill}
				]
				table[
				x=alpha,
				y=median,
				y error plus expr=\thisrow{q75} - \thisrow{median},
				y error minus expr=\thisrow{median} - \thisrow{q25},
				col sep=comma
				]
				{./data-partition-individual-experiment_full_n=6_subsetJoin_vars_beta=1.0.csv};
				\addlegendentry{$\beta = 1.0$};
				\addplot+[
				color=deep-red,
				only marks,
				mark=diamond*,
				mark options={scale=0.7,fill}
				]
				table[
				x=alpha,
				y=median,
				y error plus expr=\thisrow{q75} - \thisrow{median},
				y error minus expr=\thisrow{median} - \thisrow{q25},
				col sep=comma
				]
				{./data-partition-individual-experiment_full_n=6_subsetJoin_vars_beta=0.01.csv};
				\addlegendentry{$\beta = 0.01$};
			\end{axis}
		\end{tikzpicture}%
	\end{minipage}
	\caption{For the partition dataset with 48 elements, we report above the percentage of fixed pairs and triples with respect to the parameters $\alpha$ and $\beta$ when applying \Cref{lemma:persistency-subset-separation,proposition:edge-cut-persistency,proposition:edge-cut-persistency,lemma:persistency-triplet-cut}, (a)--(c), and \Cref{corollary:subset-join-all-pairs-at-once}, (d), separateley.}
	\label{fig:partition-individual-criteria}
\end{figure}
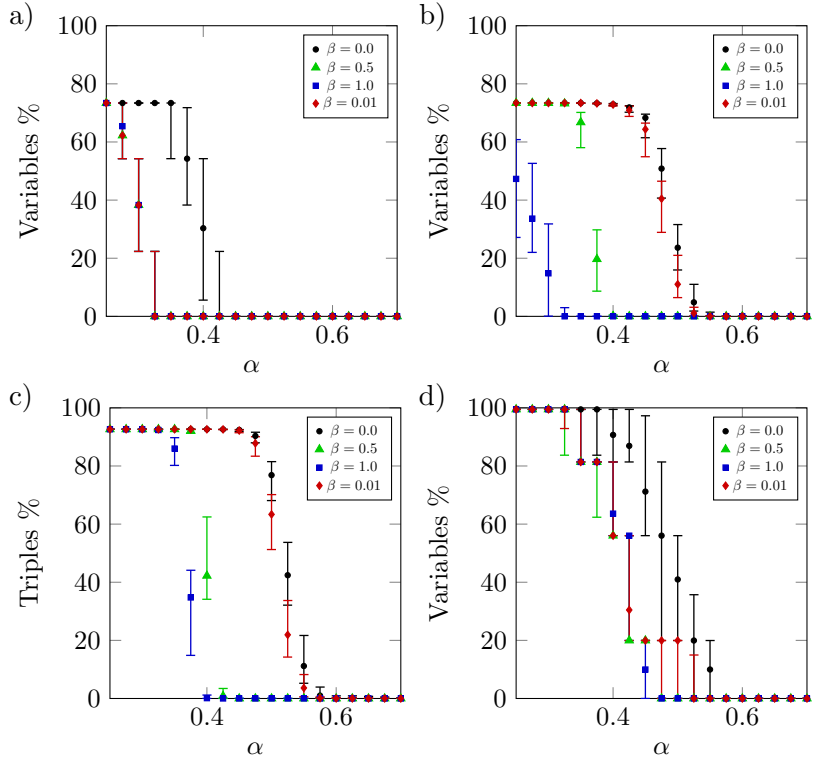%
We define the partition dataset with respect to a partition $\mathcal{R} = \left\{R_1, R_2, R_3, R_4\right\}$ of $|S| = 8n$ elements with $\vert R_1\vert = n$, $\vert R_2\vert = \vert R_3\vert = 2n$ and $\vert R_4 \vert = 3n$ elements, where $n\in \mathbb{N}$ is between 1 and 13. See also~\Cref{figure:experiments}a.
With respect to a design parameters $\alpha\in \left[0, 1\right]$, 
the costs of pairs and triplets are drawn from two Gaussian distributions with means $-1 + \alpha$ and $1-\alpha$, depending on whether their elements belong to the same set or distinct sets in the partition $\mathcal{R}$, and standard deviation $\sigma = \sigma_0 + \alpha(\sigma_1 - \sigma_0)$ with $\sigma_0 = 0.1$ and $\sigma_1 = 0.4$.
With respect to a design parameter $\beta\in \left[0, 1\right]$, the costs of pairs are multiplied by $1 - \beta$, and the costs of triples by $\beta$. 
The higher $\alpha$ is, the harder the problem becomes.
The higher $\beta$ is, the more important the costs of triples become.

The percentage of pairs and triples fixed by applying all conditions jointly, as described in \Cref{section:mixing-conditions}, is shown in \Cref{fig:partition-all-criteria-results}a. 
It can be seen from this figure that the percentage of fixed variables decreases with increasing $\alpha$. 
As $\alpha$ rises, the runtime increases but remains below one minute for all the instances; see \Cref{fig:partition-all-criteria-results}b. 
Varying $\beta$ does not affect the overall trend. 
However, the percentage of fixed variables decreases as soon as triple costs are introduced.
The percentage of pairs and triples fixed by applying \Cref{lemma:persistency-subset-separation,proposition:edge-cut-persistency,lemma:persistency-triplet-cut,corollary:subset-join-all-pairs-at-once} separately is shown in \Cref{fig:partition-individual-criteria}. 
The other partial optimality conditions do not fix any variables of these instances.
While all cut conditions settle the value of some variables, this is not the case for the join statements.
In fact, only one join condition provides partial optimality in this case: \Cref{corollary:subset-join-all-pairs-at-once}.
Interestingly, this is the one statement that fixes the most variables for almost all instances of this dataset.
For $\beta = 0.5$ and with respect the instance size, the runtime and percentage of variables fixed by applying all conditions jointly are shown in \Cref{fig:persistency-runtimes-and-variables-wrt-instance-size}a) and b). It can be seen that as the instance size increases, the number of fixed variables declines while the runtime increases. 
The runtime for $\alpha\in \{0.4, 0.5, 0.65\}$ roughly converges to $\mathcal{O}(n^{5.6})$.

\subsection{Geometric Dataset}
\label{sec:experiments-geometric}

\begin{figure}[t]
	\centering
	\begin{minipage}{0.33\linewidth}
		a)
		\hspace{-2ex}
		\vspace{-2ex}
		\linebreak
		\begin{tikzpicture}\small
			\pgfplotsset{%
				width=1.0\linewidth,
				height=1.0\linewidth
			}
			\begin{axis}[
				xlabel=$\sigma$,
				xmin=0,
				xmax=0.1,
				xtick distance=0.05,
				x tick label style={
					/pgf/number format/.cd,
					fixed
				},
				ymin=0.0,
				ymax=100,
				ylabel=Variables \%,
				legend pos=north east,
				legend style={
					nodes={
						scale=0.55
					}
				},
				]
				\addplot+[
				color=primary,
				only marks,
				mark=*,
				mark options={scale=0.5,fill}
				]
				table[
				x=sigma,
				y=median,
				y error plus expr=\thisrow{q75} - \thisrow{median},
				y error minus expr=\thisrow{median} - \thisrow{q25},
				col sep=comma
				]
				{./data-triangles-individual-experiment_findIndependentSubsets_vars.csv};
			\end{axis}
		\end{tikzpicture}
	\end{minipage}%
	\begin{minipage}{0.33\linewidth}
		b)
		\hspace{-2ex}
		\vspace{-2ex}
		\linebreak
		\begin{tikzpicture}\small
			\pgfplotsset{%
				width=1.0\linewidth,
				height=1.0\linewidth
			}
			\begin{axis}[
				xlabel=$\sigma$,
				xmin=0,
				xmax=0.1,
				xtick distance=0.05,
				x tick label style={
					/pgf/number format/.cd,
					fixed
				},
				ymin=0.0,
				ymax=100,
				ylabel=Variables \%,
				legend pos=north east,
				legend style={
					nodes={
						scale=0.55
					}
				},
				]
				\addplot+[
				color=primary,
				only marks,
				mark=*,
				mark options={scale=0.5,fill}
				]
				table[
				x=sigma,
				y=median,
				y error plus expr=\thisrow{q75} - \thisrow{median},
				y error minus expr=\thisrow{median} - \thisrow{q25},
				col sep=comma
				]
				{./data-triangles-individual-experiment_edgeCut_vars.csv};
			\end{axis}
		\end{tikzpicture}%
	\end{minipage}
	\begin{minipage}{0.33\linewidth}
		c)
		\hspace{-2ex}
		\vspace{-2ex}
		\linebreak
		\begin{tikzpicture}\small
			\pgfplotsset{%
				width=1.0\linewidth,
				height=1.0\linewidth
			}
			\begin{axis}[
				xlabel=$\sigma$,
				xmin=0,
				xmax=0.1,
				xtick distance=0.05,
				x tick label style={
					/pgf/number format/.cd,
					fixed
				},
				ymin=0.0,
				ymax=100,
				ylabel=Triples \%,
				legend pos=north east,
				legend style={
					nodes={
						scale=0.55
					}
				},
				]
				\addplot+[
				color=primary,
				only marks,
				mark=*,
				mark options={scale=0.5,fill}
				]
				table[
				x=sigma,
				y=median,
				y error plus expr=\thisrow{q75} - \thisrow{median},
				y error minus expr=\thisrow{median} - \thisrow{q25},
				col sep=comma
				]
				{./data-triangles-individual-experiment_tripletCut_triplets.csv};
			\end{axis}
		\end{tikzpicture}%
	\end{minipage}%
	\begin{minipage}{0.33\linewidth}
		d)
		\hspace{-2ex}
		\vspace{-2ex}
		\linebreak
		\begin{tikzpicture}\small
			\pgfplotsset{%
				width=1.0\linewidth,
				height=1.0\linewidth
			}
			\begin{axis}[
				xlabel=$\sigma$,
				xmin=0,
				xmax=0.1,
				xtick distance=0.05,
				x tick label style={
					/pgf/number format/.cd,
					fixed
				},
				ymin=0.0,
				ymax=100,
				ylabel=Variables \%,
				legend pos=north east,
				legend style={
					nodes={
						scale=0.55
					}
				},
				]
				\addplot+[
				color=primary,
				only marks,
				mark=*,
				mark options={scale=0.5,fill}
				]
				table[
				x=sigma,
				y=median,
				y error plus expr=\thisrow{q75} - \thisrow{median},
				y error minus expr=\thisrow{median} - \thisrow{q25},
				col sep=comma
				]
				{./data-triangles-individual-experiment_subsetJoin_vars.csv};
			\end{axis}
		\end{tikzpicture}%
	\end{minipage}
	\caption{For the geometric dataset with 45 points, we report above the percentage of fixed variables and triples with respect to the parameter $\sigma$ when employing \Cref{lemma:persistency-subset-separation,proposition:edge-cut-persistency,lemma:persistency-triplet-cut}, (a)--(c), and \Cref{corollary:subset-join-all-pairs-at-once}, (d), individually.}
	\label{fig:equilateral-individual-criteria}
\end{figure}
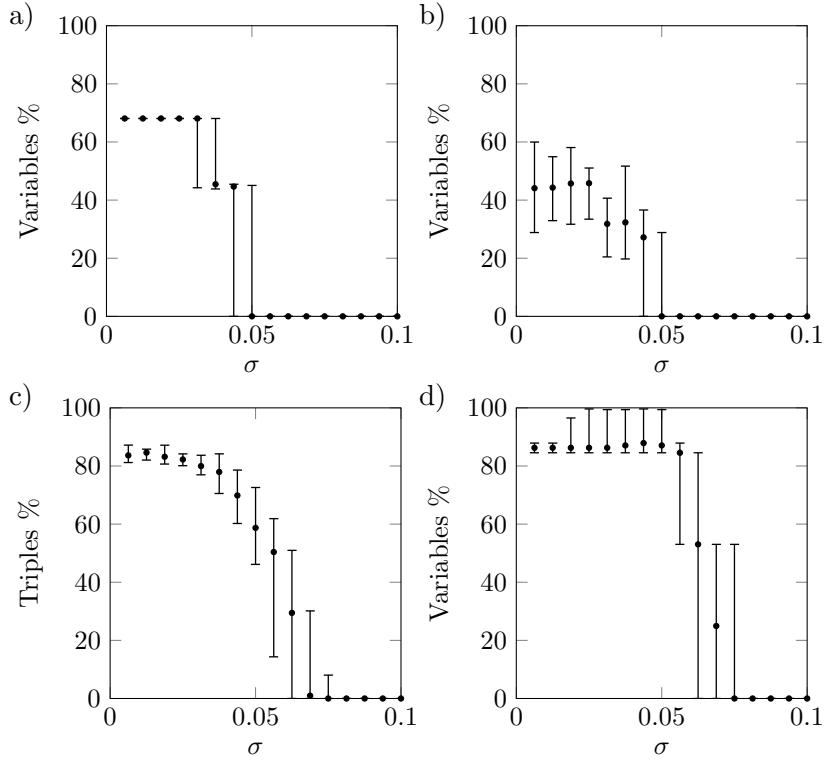

Next, we consider a dataset of instances that arise from the geometric problem of finding equilateral triangles in a noisy point cloud; see \Cref{figure:experiments}b.
For this, we fix three equilateral triangles in the plane. 
For each vertex $\vec{a}$ of a triangle, we draw a number of points from a Gaussian distribution with mean $\vec{a}$ and covariance matrix $\sigma^2\mathbbm{1}$. 
For any three points $\vec{a}_p, \vec{a}_q, \vec{a}_r$, 
let $\varphi_p$, $\varphi_q, \varphi_r$ be the interior angles of the triangle spanned by these points, and let $d^{max}_{pqr}$ and $d^{min}_{pqr}$ be the maximum and minimum length of edges in this triangle.
If the three points are mutually close, $d^{max}_{pqr} \leq 4\sigma$, we reward solutions in which these belong to the same set by letting $c_{pqr} = -1 + \frac{d^{max}_{pqr}}{4\sigma}$. 
If only two points are close, $d^{max}_{pqr} > 4\sigma$ and $d^{min}_{pqr} \leq 4\sigma$, we let $c_{pqr} = 0$. 
If the three points are mutually far apart, $d^{min}_{pqr} > 4\sigma$, we calculate the sum of the deviations of the inner angles from $\frac \pi 3, $ $\delta_{pqr} = \sum_i\vert \varphi_i - \frac{\pi}{3}\vert$. 
If this quantity is below $\frac \pi 6$, we let $c_{pqr} = -1 + \frac{6\delta_{pqr}}{\pi}$. Otherwise, $c_{pqr} = \frac{6}{7}\frac{\delta_{pqr} - \frac{\pi}{6}}{\pi}$. 

The percentage of pairs and triples fixed by applying all conditions jointly, as described in \Cref{section:mixing-conditions}, is reported in \Cref{fig:partition-all-criteria-results}c. 
Here, the hardness of the instances is embodied by $\sigma$. 
The number of points is 45.

As $\sigma$ increases, the percentage of fixed variables decreases. 
The runtime increases, as can be seen from \Cref{fig:partition-all-criteria-results}d, and stays below one minute for all these instances.
The percentage of pairs and triples fixed by applying \Cref{lemma:persistency-subset-separation,proposition:edge-cut-persistency,lemma:persistency-triplet-cut,corollary:subset-join-all-pairs-at-once} separately is shown in \Cref{fig:equilateral-individual-criteria}. 
Also here, all the cut conditions are effective whereas the only useful join condition is \Cref{corollary:subset-join-all-pairs-at-once}.
Moreover, \Cref{corollary:subset-join-all-pairs-at-once} is overall the most effective.
The runtime and percentage of variables fixed by applying all conditions jointly and with respect the instance size are shown in \Cref{fig:persistency-runtimes-and-variables-wrt-instance-size}c) and d). Similar to the partition dataset, we see that the number of fixed variables decreases as the instance size increases, while the runtime gets worse. 
The runtime for $\sigma\in \{0.06, 0.1\}$ roughly converges to $\mathcal{O}(n^{5.8})$. 

\section{Conclusion}
We establish partial optimality conditions for the cubic set partition problem, which can be seen as the special case of cubic correlation clustering for complete graphs.
In particular, we generalize all such conditions known for correlation clustering with a linear objective function to arbitrary cubic objective functions.
In addition, we establish new partial optimality conditions.
Furthermore, we define and implement exact algorithms and heuristics for testing all established conditions efficiently.
Lastly, we quantify the effect of these algorithms on two datasets.
Regarding these numerical experiments, we note that all cut conditions are effective on the tested datasets, whereas join conditions pose a bigger challenge. 
In fact, \Cref{proposition:subset-join-proposition}, in its simplified form of \Cref{corollary:subset-join-all-pairs-at-once}, is the only join property that is beneficial in our numerical experiments. 
Yet, in almost all cases, it is the one statement that fixes the most variables (\Cref{fig:partition-individual-criteria,fig:equilateral-individual-criteria}). 
We remark that \Cref{corollary:subset-join-all-pairs-at-once} is one of the newly proposed conditions. 
Perspectives for future work include the exploitation of sparsity of non-zero cost coefficients, as well as applications to subspace clustering and object recognition.

\section{Acknowledgement}
Bjoern Andres and David Stein acknowledge funding by the Federal Ministry of Education and Research of Germany, from grant 01LC2006A.

\bibliographystyle{plainnat}
\bibliography{manuscript}

\begin{thebibliography}{24}
\providecommand{\natexlab}[1]{#1}
\providecommand{\url}[1]{\texttt{#1}}
\expandafter\ifx\csname urlstyle\endcsname\relax
  \providecommand{\doi}[1]{doi: #1}\else
  \providecommand{\doi}{doi: \begingroup \urlstyle{rm}\Url}\fi

\bibitem[Adams et~al.(1998)Adams, Lassiter, and Sherali]{AdaLasShe98}
Warren~P. Adams, Julie~Bowers Lassiter, and Hanif~D. Sherali.
\newblock Persistency in 0-1 polynomial programming.
\newblock \emph{Mathematics of Operations Research}, 23\penalty0 (2):\penalty0
  359--389, 1998.
\newblock \doi{10.1287/moor.23.2.359}.

\bibitem[Agarwal et~al.(2005)Agarwal, Lim, Zelnik-Manor, Perona, Kriegman, and
  Belongie]{AgaLimZelPerKriBel05}
S.~Agarwal, J.~Lim, L.~Zelnik-Manor, P.~Perona, D.~Kriegman, and S.~Belongie.
\newblock Beyond pairwise clustering.
\newblock In \emph{CVPR}, 2005.
\newblock \doi{10.1109/CVPR.2005.89}.

\bibitem[Alush and Goldberger(2012)]{alush-2012}
Amir Alush and Jacob Goldberger.
\newblock Ensemble segmentation using efficient integer linear programming.
\newblock \emph{Transactions on Pattern Analysis and Machine Intelligence},
  34\penalty0 (10):\penalty0 1966--1977, 2012.
\newblock \doi{10.1109/TPAMI.2011.280}.

\bibitem[Billionnet and Sutter(1992)]{BilSut92}
Alain Billionnet and Alain Sutter.
\newblock Persistency in quadratic 0--1 optimization.
\newblock \emph{Mathematical Programming}, 54\penalty0 (1):\penalty0 115--119,
  1992.
\newblock \doi{10.1007/BF01586044}.

\bibitem[Boost(2022)]{boost}
Boost.
\newblock {Boost C++ Libraries}.
\newblock \url{https://www.boost.org}, 2022.

\bibitem[Boros et~al.(2008)Boros, Hammer, Sun, and Tavares]{boros-2008}
Endre Boros, Peter~L. Hammer, Richard Sun, and Gabriel Tavares.
\newblock A max-flow approach to improved lower bounds for quadratic
  unconstrained binary optimization (qubo).
\newblock \emph{Discrete Optimization}, 5\penalty0 (2):\penalty0 501--529,
  2008.
\newblock \doi{10.1016/j.disopt.2007.02.001}.

\bibitem[Goldberg and Tarjan(1988)]{goldberg-1988}
Andrew~V. Goldberg and Robert~E. Tarjan.
\newblock A new approach to the maximum-flow problem.
\newblock \emph{Journal of the ACM}, 35\penalty0 (4):\penalty0 921--940, 1988.
\newblock \doi{10.1145/48014.61051}.

\bibitem[Gr\"otschel and Wakabayashi(1989)]{goetschel-1989}
Martin Gr\"otschel and Y.~Wakabayashi.
\newblock A cutting plane algorithm for a clustering problem.
\newblock \emph{Mathematical Programming}, 45\penalty0 (1):\penalty0 59--96,
  1989.
\newblock \doi{10.1007/BF01589097}.

\bibitem[Hammer et~al.(1984)Hammer, Hansen, and Simeone]{HamHanSim84}
Peter~L. Hammer, Pierre Hansen, and Bruno Simeone.
\newblock Roof duality, complementation and persistency in quadratic 0--1
  optimization.
\newblock \emph{Mathematical Programming}, 28\penalty0 (2):\penalty0 121--155,
  1984.
\newblock \doi{10.1007/BF02612354}.

\bibitem[Kappes et~al.(2013)Kappes, Speth, Reinelt, and
  Schn{\"o}rr]{KapSpeReiSch13}
J{\"o}rg~Hendrik Kappes, Markus Speth, Gerhard Reinelt, and Christoph
  Schn{\"o}rr.
\newblock Towards efficient and exact map-inference for large scale discrete
  computer vision problems via combinatorial optimization.
\newblock In \emph{CVPR}, 2013.
\newblock \doi{10.1109/CVPR.2013.229}.

\bibitem[Kappes et~al.(2016)Kappes, Speth, Reinelt, and
  Schn{\"o}rr]{KapSpeReiSch16}
J{\"o}rg~Hendrik Kappes, Markus Speth, Gerhard Reinelt, and Christoph
  Schn{\"o}rr.
\newblock Higher-order segmentation via multicuts.
\newblock \emph{Computer Vision and Image Understanding}, 143:\penalty0
  104--119, 2016.
\newblock \doi{10.1016/j.cviu.2015.11.005}.

\bibitem[Kim et~al.(2014)Kim, Yoo, Nowozin, and Kohli]{kim-2014}
Sungwoong Kim, Chang~Dong Yoo, Sebastian Nowozin, and Pushmeet Kohli.
\newblock Image segmentation using higher-order correlation clustering.
\newblock \emph{Transactions on Pattern Analysis and Machine Intelligence},
  36\penalty0 (9):\penalty0 1761--1774, 2014.
\newblock \doi{10.1109/TPAMI.2014.2303095}.

\bibitem[Kohli et~al.(2008)Kohli, Shekhovtsov, Rother, Kolmogorov, and
  Torr]{KohEtAl08}
Pushmeet Kohli, Alexander Shekhovtsov, Carsten Rother, Vladimir Kolmogorov, and
  Philip Torr.
\newblock On partial optimality in multi-label mrfs.
\newblock In \emph{ICML}, 2008.
\newblock \doi{10.1145/1390156.1390217}.

\bibitem[Kolmogorov and Zabin(2004)]{kolmogorov-2004}
V.~Kolmogorov and R.~Zabin.
\newblock What energy functions can be minimized via graph cuts?
\newblock \emph{Transactions on Pattern Analysis and Machine Intelligence},
  26\penalty0 (2):\penalty0 147--159, 2004.
\newblock \doi{10.1109/TPAMI.2004.1262177}.

\bibitem[Lange et~al.(2018)Lange, Karrenbauer, and Andres]{Lange-2018}
Jan-Hendrik Lange, Andreas Karrenbauer, and Bjoern Andres.
\newblock Partial optimality and fast lower bounds for weighted correlation
  clustering.
\newblock In \emph{ICML}, 2018.
\newblock URL \url{https://proceedings.mlr.press/v80/lange18a.html}.

\bibitem[Lange et~al.(2019)Lange, Andres, and Swoboda]{Lange-2019}
Jan{-}Hendrik Lange, Bjoern Andres, and Paul Swoboda.
\newblock Combinatorial persistency criteria for multicut and max-cut.
\newblock In \emph{CVPR}, 2019.
\newblock \doi{10.1109/CVPR.2019.00625}.

\bibitem[Levinkov et~al.(2022)Levinkov, Kardoost, Andres, and
  Keuper]{LevKarAndKeu22}
Evgeny Levinkov, Amirhossein Kardoost, Bjoern Andres, and Margret Keuper.
\newblock Higher-order multicuts for geometric model fitting and motion
  segmentation.
\newblock \emph{Transactions on Pattern Analysis and Machine Intelligence},
  pages 1--1, 2022.
\newblock \doi{10.1109/TPAMI.2022.3148795}.

\bibitem[Ochs and Brox(2012)]{OchBro12}
Peter Ochs and Thomas Brox.
\newblock Higher order motion models and spectral clustering.
\newblock In \emph{CVPR}, 2012.
\newblock \doi{10.1109/CVPR.2012.6247728}.

\bibitem[Purkait et~al.(2017)Purkait, Chin, Sadri, and Suter]{PulChiSadSut17}
Pulak Purkait, Tat-Jun Chin, Alireza Sadri, and David Suter.
\newblock Clustering with hypergraphs: The case for large hyperedges.
\newblock \emph{Transactions on Pattern Analysis and Machine Intelligence},
  39\penalty0 (9):\penalty0 1697--1711, 2017.
\newblock \doi{10.1109/TPAMI.2016.2614980}.

\bibitem[Shekhovtsov(2013)]{shekhovtsov-2013}
Alexander Shekhovtsov.
\newblock \emph{Exact and Partial Energy Minimization in Computer Vision}.
\newblock PhD thesis, Center for Machine Perception, Czech Technical
  University, Prague, 2013.

\bibitem[Shekhovtsov(2014)]{shekhovtsov-2014}
Alexander Shekhovtsov.
\newblock Maximum persistency in energy minimization.
\newblock In \emph{CVPR}, 2014.
\newblock \doi{10.1109/CVPR.2014.152}.

\bibitem[Shekhovtsov et~al.(15)Shekhovtsov, Swoboda, and
  Savchynskyy]{shekhovtsov-2015}
Alexander Shekhovtsov, Paul Swoboda, and Bogdan Savchynskyy.
\newblock Maximum persistency via iterative relaxed inference with graphical
  models.
\newblock In \emph{CVPR}, 15.
\newblock \doi{10.1109/CVPR.2015.7298650}.

\bibitem[Stoer and Wagner(1997)]{stoer-1997}
Mechthild Stoer and Frank Wagner.
\newblock A simple min-cut algorithm.
\newblock \emph{Journal of the ACM}, 44\penalty0 (4):\penalty0 585--591, 1997.
\newblock \doi{10.1145/263867.263872}.

\bibitem[Veldt(2022)]{Vel22}
Nate Veldt.
\newblock Correlation clustering via strong triadic closure labeling: Fast
  approximation algorithms and practical lower bounds.
\newblock In \emph{ICML}, 2022.
\newblock URL \url{https://proceedings.mlr.press/v162/veldt22a.html}.

\end{thebibliography}

\appendix
\section{Appendices}
\ifthenelse{\boolean{proofs}}{}{\subsection{Proofs}}

\ifthenelse{\boolean{proofs}}{}{
\begin{delayedproof}{def: pb}
For each partition $\Pi$ of the set $S$ and every distinct $p, q \in S$, let $x_{pq} = 1$ if and only if $p$ and $q$ are in the same set of $\Pi$.
This establishes a one-to-one relation between the set $P_S$ of all partitions of $S$ and the set $X_S$ \citep{goetschel-1989}.
Under this bijection, the objective functions of \Cref{def: first def} and \Cref{def: pb}
are equivalent.
\end{delayedproof}
}

\ifthenelse{\boolean{proofs}}{}{
\begin{delayedproof}{lemma:persistency-predicate}
	Let $x^*$ be an optimal solution to $\min_{x\in X}\phi(x)$ such that $x^*\not\in Q$. Then $\sigma(x^*)$ is also an optimal solution to $\min_{x\in X}\phi(x)$ and $\sigma(x^*)\in Q$.
\end{delayedproof}
}

\ifthenelse{\boolean{proofs}}{}{
\begin{delayedproof}{lemma:persistency-subset-separation}
We define $\sigma \colon \cp_S \to \cp_S$ such that for all $x \in \cp_S$ we have that
\begin{equation}
\sigma(x) := \begin{cases}
x & \text{ if } x_{ij} = 0 \:\:	 \forall ij \in \delta(R) \\
\sigma_{\delta(R)}(x) & \textnormal{otherwise}
\end{cases}.
\end{equation}
For any $x\in \cp_S$, let $x' = \sigma(x)$. 
First, the map $\sigma$ is such that $x'_{ij} = 0$ 
for all $ij \in \delta(R)$. 
Second, for any $x\in \cp_S$ such that there exists $ij\in \delta(R)$ such that $x_{ij} = 1$, we have that
\begin{equation}
\phi_c(x') - \phi_c(x) = - \sum_{pqr \in T_{\delta(R)}} c_{pqr}x_{pq}x_{pr}x_{qr}- \sum_{pq\in \delta(R)} c_{pq}x_{pq} \leq - \sum_{pqr\in T_{\delta(R)}\cap T^-} c_{pqr} - \sum_{pq\in \delta(R)\cap P^-}c_{pq}
= \;0.
\end{equation}
The last equality is due to the fact that those sums vanish by assumptions \eqref{eq:edge-cut-condition-1} and \eqref{eq:edge-cut-condition-2}.
Applying Corollary~\ref{lemma:persistency-variable} concludes the proof.
\end{delayedproof}
}

\ifthenelse{\boolean{proofs}}{}{
\begin{delayedproof}{proposition:edge-cut-persistency}
Let 
$\sigma \colon \cp_S \to \cp_S$ be constructed as 
\begin{equation}
\sigma(x) := \begin{cases}
x & \textnormal{if $x_{ij} = 0$} \\
\sigma_{\delta(R)}(x) & \textnormal{otherwise}
\end{cases}.
\end{equation}
For any $x\in \cp_S$, let $x' = \sigma(x)$. 
First of all, the map $\sigma$ is such that $x'_{ij} = 0$ for all $x\in \cp_S$. 
Next, for any $x\in \cp_S$ such that $x_{ij} = 1$ we have that 
\begin{align}
\phi_c(x') - \phi_c(x) & = -c_{ij} -\sum_{pqr\in T_{\delta(R)}} c_{pqr}x_{pq} x_{pr}x_{qr} - \sum_{\substack{pq\in \delta(R) \\ pq \neq ij}}c_{pq}x_{pq} \\
& \leq -c_{ij} + \sum_{pqr\in T_{\delta(R)}}c_{pqr}^- + \sum_{\substack{pq \in \delta(R)\\ pq \neq ij}}c_{pq}^- = -c_{ij}^+ + \sum_{pqr\in T_{\delta(R)}}c_{pqr}^- + \sum_{pq \in \delta(R)}c_{pq}^- \leq \;0.
\end{align}
The last inequality follows from assumption \eqref{eq:assumption-edge-cut-inequality}.
We conclude the proof by applying Corollary~\ref{lemma:persistency-variable}. 
\end{delayedproof}
}

\ifthenelse{\boolean{proofs}}{}{
\begin{delayedproof}{lemma:persistency-triplet-cut}
We define $\sigma \colon \cp_S\to \cp_S$ as 
\begin{equation}
\sigma(x) := \begin{cases}
x & \textnormal{if $x_{ij}x_{ik}x_{jk} = 0$}\\
\sigma_{\delta(R)}(x) & \textnormal{otherwise}
\end{cases}.
\end{equation}
For any $x\in \cp_S$, we denote $\sigma(x)$ by $x'$. 
We observe that   
$x'_{ij}x'_{ik}x'_{jk} = 0$ for all $x\in \cp_S$. 
Second, for any $x\in \cp_S$ such that $x_{ij}x_{ik}x_{jk} = 1$ it holds that
\small{
\begin{align}
\phi_c(x') - \phi_c(x) & = - \sum_{pqr \in T_{\delta(R)}} c_{pqr}x_{pq}x_{pr}x_{qr} - \sum_{pq\in \delta(R)} c_{pq}x_{pq} \leq -c_{ijk} - c_{ij} - c_{ik} + \sum_{\substack{pqr\in T_{\delta(R)}\\ pqr \neq ijk}} c_{pqr}^- + \sum_{\substack{pq\in \delta(R)\\ pq\not\in \{ij, ik\}}} c_{pq}^-\\
& = -c_{ijk}^+ - c_{ij}^+ - c_{ik}^+ + \sum_{pqr\in T_{\delta(R)}} c_{pqr}^- + \sum_{pq\in \delta(R)} c_{pq}^- \leq \;0.
\end{align}
}
\normalsize
The last inequality holds because of 
assumption \eqref{eq:triplet-cut-condition}.
Applying Proposition~\ref{lemma:persistency-predicate} with $Q = \{x\in \cp_S \mid x_{ij}x_{ik}x_{jk} = 0\}$ concludes the proof.
\end{delayedproof}
}

\ifthenelse{\boolean{proofs}}{}{
\begin{delayedproof}{lemma:edge-join-persistency}
Let $\sigma \colon \cp_S \to \cp_S$ be such that for all $x \in \cp_S$ it holds 
that 
\begin{equation}
\sigma(x) := \begin{cases}
x & \textnormal{if $x_{ij} = 1$} \\
\left(\sigma_{ij}\circ \sigma_{\delta(R)}\right)(x) & \textnormal{otherwise}
\end{cases}.
\end{equation}
For any $x\in \cp_S$, let $x' = \sigma(x)$. 
Firstly, the map $\sigma$ is such that $x'_{ij} = 1$ for all $x\in \cp_S$. 
We now show that $\sigma$ is improving. 
In particular, let $x \in \cp_S$ such that $x_{ij} = 0$.  
We observe that $x'_{pq} = x_{pq}$ for all $pq \not\in \delta(R)$. 
Therefore,
\begin{align}
\phi_c(x') - \phi(x) & = \sum_{pqr\in T_{\delta(R)}}c_{pqr}\left(x'_{pq}x'_{pr}x'_{qr} - x_{pq}x_{pr}x_{qr}\right) + \sum_{pq\in \delta(R)}c_{pq}\left(x'_{pq} - x_{pq}\right) \\
& = \sum_{pqr\in T_{\delta(R)}\setminus T_{\{ij\}}}c_{pqr}\left(x'_{pq}x'_{pr}x'_{qr} - x_{pq}x_{pr}x_{qr}\right) +\sum_{pqr\in T_{\{ij\}}}c_{pqr}x'_{pq}x'_{pr}x'_{qr}  + c_{ij} \\
& + \sum_{\substack{pq\in \delta(R)\\ pq \neq ij}}c_{pq}\left(x'_{pq} - x_{pq}\right)\\
			& \leq \; \sum_{pqr\in T_{\delta(R)}\setminus T_{\{ij\}}}\vert c_{pqr}\vert + \sum_{pqr\in T_{\{ij\}}}c_{pqr}^+ + c_{ij} + \sum_{\substack{pq\in \delta(R)\\ pq \neq ij}}\vert c_{pq}\vert \\
			& = \sum_{pqr\in T_{\delta(R)}}\vert c_{pqr}\vert - \sum_{pqr\in T_{\{ij\}}}c_{pqr}^- - 2c_{ij}^-  + \sum_{pq\in \delta(R)}\vert c_{pq}\vert \leq  0 .
		\end{align}
We remark that the last inequality is due to assumption \eqref{eq:edge-join-inequality}.
\end{delayedproof}
}

\ifthenelse{\boolean{proofs}}{}{
\begin{delayedproof}{proposition:triplet-join}
We construct 
$\sigma \colon \cp_S \to \cp_S$ as follows 
\begin{equation}
\sigma(x) := \begin{cases}
x & \textnormal{if $x_{ij}x_{ik}x_{jk} = 1$} \\
\left(\sigma_{ijk}\circ \sigma_{\delta(R)}\right)(x) & \textnormal{otherwise}
\end{cases}.
\end{equation}
For any $x \in \cp_S$, let $x' = \sigma(x)$. 
First, the map $\sigma$ is such that $x'_{ij}x'_{ik}x'_{jk} = 1$ for all $x\in \cp_S$. 
Note that, for any $x\in \cp_S$ such that $x_{ij}x_{ik}x_{jk} = 0$, we have that 
$x'_{pq} \geq x_{pq}$ for all $pq \not\in \delta(R)$. 
It follows that 
\begin{align}
\phi_c(x') - \phi(x) & = \sum_{\substack{pqr\in T_{\delta(R)} \\ pqr \neq ijk}}c_{pqr}(x'_{pq}x'_{pr}x'_{qr} - x_{pq}x_{pr}x_{qr}) + c_{ijk} + \sum_{pqr\not\in T_{\delta(R)}}c_{pqr}(x'_{pq}x'_{pr}x'_{qr} - x_{pq}x_{pr}x_{qr}) \\					
& + \sum_{pq\in \{ij, ik, jk\}}c_{pq}(1-x_{pq}) + \sum_{\substack{pq\in \delta(R)\\ pq \not\in \{ij, ik\}}}c_{pq}(x'_{pq} - x_{pq}) + \sum_{pq\not\in \delta(R)\cup \{jk\}}c_{pq}(x'_{pq} - x_{pq})\\	
& \leq \; c_{ijk} +  \max_{\substack{x\in \cp{ijk} \\ x_{ij}x_{ik}x_{jk} = 0}}\sum_{pq\in \tbinom{ijk}{2}}c_{pq} (1 - x_{pq}) + \sum_{\substack{pqr\in T_{\delta(R)} \\ pqr \neq ijk}}\vert c_{pqr}\vert  + \sum_{\substack{pqr\in T^+ \\ pqr \notin T_{\delta(R)}}} c_{pqr} \\
& + \sum_{\substack{pq\in \delta(R) \\ pq \notin \{ij, ik\}}} \vert c_{pq} \vert + \sum_{\substack{pq \in P^+ \\ pq \notin \left(\delta(R)\cup \{jk\}\right)}}c_{pq} \\	
& = \; c_{ijk} +  \max_{\substack{x\in \cp_{ijk} \\ x_{ij}x_{ik}x_{jk} = 0}}\sum_{pq\in \tbinom{ijk}{2}}c_{pq} (1 - x_{pq}) + \sum_{\substack{pqr\in T^+ \cup T_{\delta(R)} \\ pqr \neq ijk}}\vert c_{pqr}\vert + \sum_{\substack{pq\in P^+ \cup \delta(R)\\ pq \not\in \{ij, ik, jk\}}}\vert c_{pq} \vert \\
& = \; -2c_{ijk}^- -  2c_{ij}^- - 2c_{ik}^- - c_{jk}^- - \min_{\substack{x\in \cp_{ijk} \\ x_{ij}x_{ik}x_{jk} = 0}}\sum_{pq\in \tbinom{ijk}{2}}c_{pq} x_{pq}\\
& + \sum_{\substack{pqr\in T_{\delta(R)}}}c_{pqr}^- +\sum_{pqr\in \tbinom S3}c_{pqr}^+ + \sum_{pq\in \tbinom S2}c_{pq}^+ + \sum_{pq\in \delta(R)}c_{pq}^- \leq  0.
\end{align}
Assumption \eqref{eq:assumption-triplet-join} provides the last inequality.
We arrive at the thesis by applying Proposition~\ref{lemma:persistency-predicate} with $Q = \{x\in \cp_S \mid x_{ij}x_{ik}x_{jk} = 1\}$. 
\end{delayedproof}
}

\ifthenelse{\boolean{proofs}}{}{
\begin{delayedproof}{lemma:persistency-triangle-edge-join}
Let $\sigma \colon \cp_S \to \cp_S$ be defined as 
\begin{equation}
\sigma(x) := \begin{cases}
x & \textnormal{if $x_{ik} = 1$}\\
(\sigma_{ik} \circ \sigma_{\delta(R)})(x) & \textnormal{if $x_{ik} = x_{ij} = 0 ,x_{jk} = 1$} \\
(\sigma_{ik}\circ \sigma_{\delta(R')})(x) & \textnormal{if $x_{ik} = x_{jk} = 0 , x_{ij} = 1$} \\
(\sigma_{ijk}\circ \sigma_{\delta(ijk)})(x) & \textnormal{if $x_{ik} = x_{ij} = x_{jk} = 0$}
\end{cases}.
\end{equation}
We use the notation $x' = \sigma(x)$ for all $x \in \cp$.
We note that the map $\sigma$ is such that $x'_{ik} = 1$ for all $x\in \cp_S$. 
Second, for any $x\in \cp_S$ such that $x_{ik} = x_{ij} = 0$ and $x_{jk} = 1$, the map $\sigma_{ik}\circ \sigma_{\delta(R)}$ is such that $x'_{ij}=x'_{ik} = x'_{jk} = 1$ and $x'_{pq} = x_{pq}$ for any $pq \not\in \delta(R)$. 
It holds that
\begin{align}
\phi_c(x') - \phi_c(x) & = c_{ijk} + c_{ij} + c_{ik} + \sum_{\substack{pqr \in T_{\{ij, ik\}}\\ pqr\neq ijk}} c_{pqr}x'_{pq}x'_{pr}x'_{qr} \\
& + \sum_{\substack{pqr\in T_{\delta(R)} \\ pqr \notin T_{\{ij, ik\}}}} c_{pqr}\left(x'_{pq}x'_{pr}x'_{qr} - x_{pq}x_{pr}x_{qr}\right) + \sum_{\substack{pq\in \delta(R) \\ pq \not\in \{ij, ik\}}} c_{pq}(x'_{pq} - x_{pq}) \\
& \leq c_{ijk} + c_{ij} + c_{ik} +\sum_{\substack{pqr\in T_{\{ij, ik\}}\cap T^+ \\ pqr\neq ijk}}c_{pqr} + \sum_{\substack{pqr\in T_{\delta(R)} \\ pqr \notin T_{\{ij, ik\}}}}\vert c_{pqr}\vert + \sum_{\substack{pq\in \delta(R) \\ pq \not\in \{ij, ik\}}}  \vert c_{pq}\vert \\
& = -c_{ijk}^- - 2c_{ij}^- - 2c_{ik}^- +\sum_{\substack{pqr\in T_{\{ij, ik\}}}}c_{pqr}^+ + \sum_{pqr\in T_{\delta(R)}}\vert c_{pqr}\vert - \sum_{pqr\in T_{\{ij, ik\}}}\vert c_{pqr}\vert+ \sum_{pq\in \delta(R)}  \vert c_{pq}\vert \\
& = -c_{ijk}^- - 2c_{ij}^- - 2c_{ik}^- -\sum_{\substack{pqr\in T_{\{ij, ik\}}}}c_{pqr}^- + \sum_{pqr\in T_{\delta(R)}}\vert c_{pqr}\vert + \sum_{pq\in \delta(R)}  \vert c_{pq}\vert \leq 0.
\end{align}	
The last inequality follows from 
assumption \eqref{eq:triangle-edge-join-1}.
Third, for any $x\in \cp_S$ such that $x_{ik} = x_{jk} = 0$ and $x_{ij} = 1$, the map $\sigma_{ik}\circ \sigma_{\delta(R')}$ is improving by analogous arguments and assumption \eqref{eq:triangle-edge-join-2}.
Finally, for any $x\in \cp_S$ such that $x_{ik} = x_{jk} = x_{ij} = 0$, the map $\sigma_{ijk}\circ \sigma_{\delta(ijk)}$ is such that 
\begin{equation}
(\sigma_{ijk}\circ \sigma_{\delta(ijk)})_{pq} = \begin{cases}
0 & \textnormal{if $pq \in \delta(ijk)$} \\
1 & \textnormal{if $pq \in \{ij, ik, jk\}$} \\
x_{pq} & \textnormal{otherwise}
\end{cases}.
\end{equation}
Therefore, 
\begin{align}
\phi_c(x') - \phi_c(x) & = c_{ijk} + c_{ij} + c_{ik} + c_{jk} - \sum_{pqr\in T_{\delta(ijk)}\setminus T_{\{ij, ik, jk\}}} c_{pqr}x_{pq}x_{pr}x_{qr} - \sum_{pq\in \delta(ijk)} c_{pq} x_{pq} \\
& \leq c_{ijk} + c_{ij} + c_{ik} + c_{jk}  +\sum_{\substack{pqr\in T_{\delta(ijk)} \cap T^- \\ pqr\not\in T_{\{ij, ik, jk\}}}} \vert c_{pqr} \vert + \sum_{pq\in \delta(ijk)\cap P^-} \vert c_{pq}\vert \leq 0.
\end{align}
The last inequality is true thanks to to assumption \eqref{eq:triangle-edge-join-3}.
Applying Corollary~\ref{lemma:persistency-variable} concludes the proof.
\end{delayedproof}
}

\ifthenelse{\boolean{proofs}}{}{
\begin{delayedproof}{lemma:general-subgraph-edge-join}
First, we prove a lemma that establishes a relation which will be needed at the end of this proof. 
\begin{lemma}
	\label{lemma:subgraph-helping-lemma-1}
	Let $S \neq \emptyset$ and $c \in \mathbb{R}^{\mathcal{I}_S}$. 
	We define $c'\in \mathbb{R}^{\mathcal{I}_S}$ as in \eqref{eq: defcprime1}, \eqref{eq: defcprime2}, \eqref{eq: defcprime3} for $S_H = S$.
	Then for any partition $\mathcal{R}$ of $S$ it holds that 
	{
		\small
		\begin{align}
			\phi_{c'}(x^\mathcal{R}) & = \;\frac 12 \sum_{pqr\in \tbinom S3} c_{pqr}\prod_{uv\in \tbinom{pqr}{2}}(1-x^\mathcal{R}_{uv}) + \sum_{pqr\in \tbinom S3}c_{pqr}\sum_{uv \in \tbinom{pqr}{2}}x^\mathcal{U}_{uv}\prod_{\substack{u'v'\in \tbinom{pqr}{2} \\ u'v' \neq uv}}(1- x^\mathcal{R}_{u'v'}) \\
			& + \sum_{pq\in \tbinom S2}c_{pq}(1-x^\mathcal{R}_{pq})	\label{eq:partition-subgraph-helper-1} \\ 
			& = \;\frac 12 \sum_{RR'R''\in \tbinom{\mathcal{R}}{3}}\sum_{pqr\in T_{RR'R''}}c_{pqr} +\sum_{RR'\in \tbinom{\mathcal{R}}{2}}\sum_{pq\in \delta(R, R')}c_{pq} + \sum_{RR'\in \tbinom{\mathcal{R}}{2}}\left(\sum_{pqr\in T_{RRR'}}c_{pqr} +  \sum_{pqr\in T_{RR'R'}}c_{pqr}\right) ,
			\label{eq:partition-subgraph-helper-2}
		\end{align}
	}%
	where $x^\mathcal{R}$ denotes the feasible vector corresponding to the partition $\mathcal{R}$ of $S$.
\end{lemma}

\begin{cpfc}[\Cref{lemma:subgraph-helping-lemma-1}]
We use the fact that for any partition $\mathcal{R}$ of $S$ and any $pqr\in \tbinom S3$ we have that 
$x_{pq}^\mathcal{R}x_{qr}^\mathcal{R}=x_{pq}^\mathcal{R}x_{pr}^\mathcal{R} = x_{pr}^\mathcal{R}x_{qr}^\mathcal{R} = x_{pq}^\mathcal{R}x_{pr}^\mathcal{R}x_{qr}^\mathcal{R}$. 
Expanding the inner products and the inner sums leads to
\begin{align}
		\prod_{uv\in \tbinom{pqr}{3}}\left(1-x^\mathcal{R}_{uv}\right) & = 1 - x_{pq}^\mathcal{R} - x_{pr}^\mathcal{R}-x_{qr}^\mathcal{R} + 2x_{pq}^\mathcal{R}x_{pr}^\mathcal{R}x_{qr}^\mathcal{R}, \\
		\sum_{uv\in \tbinom{pqr}{2}}x_{uv}^\mathcal{R}\prod_{\substack{u'v'\in \tbinom{pqr}{2} \\ u'v' \notin \{uv\}}}(1-x_{u'v'}^\mathcal{R}) & = x^\mathcal{R}_{pq} + x^\mathcal{R}_{pr} + x^\mathcal{R}_{qr} -3 x^\mathcal{R}_{pq}x^\mathcal{R}_{qr}x^\mathcal{R}_{pr}.
\end{align}
By plugging in and collecting terms we conclude the proof for equality \eqref{eq:partition-subgraph-helper-1}. 
Equality~\eqref{eq:partition-subgraph-helper-2} follows instead from the following observations: 
\begin{align}
\prod_{ab\in \tbinom{pqr}{3}}\left(1-x^\mathcal{R}_{ab}\right) = 1 & \Leftrightarrow \; \exists \;RR'R''\in \tbinom{\mathcal{R}}{3}\colon pqr\in T_{RR'R''}, \\
\sum_{ab\in \tbinom{pqr}{2}}x_{ab}^\mathcal{R}\prod_{a'b'\in \tbinom{pqr}{2}\setminus \{ab\}}(1-x_{a'b'}^\mathcal{R}) = 1 & \Leftrightarrow \; \exists \;RR'\in \tbinom{\mathcal{R}}{2}\colon \left(pqr\in T_{RRR'} \lor pqr\in T_{RR'R'}\right),\\
1 - x^\mathcal{R}_{pq} = 1 & \Leftrightarrow \exists \; RR'\in \tbinom{\mathcal{R}}{2}\colon pq\in \delta(R, R') .
\end{align}
This concludes the proof.
\end{cpfc}

\noindent
We define $\sigma \colon \cp_S \to \cp_S$ as 
\begin{equation}
\sigma(x) := \begin{cases}
x & \textnormal{if $x_{ij} = 1$}\\
(\sigma_{S_H} \circ \sigma_{\delta(S_H)})(x) & \textnormal{otherwise}
\end{cases}.
\end{equation}
Let $x' = \sigma(x)$ for any $x \in \cp_S$. 
It is easy to see that $x'_{ij} = 1$ for all $x\in \cp_S$. 
Similarly to before, 
we show that $\sigma$ is an improving map. 
For any $x \in \cp_S$ such that $x_{ij} = 1$ we have that $\phi_c(x') - \phi_c(x) = 0$ by definition of $x'$. 
Now, we consider $x\in \cp_S$ such that $x_{ij} = 0$. 
Let $P_H = \tbinom{S_H}{2}$ and $T_H = \tbinom{S_H}{3}$. 
We denote the restriction of $x$ containing only components corresponding to elements in 
$P_H$ by $x\vert_{P_H}$. 
Let $\mathcal{R}$ be the partition of $S$ such that $x = x^\mathcal{R}$, and let $\mathcal{R}_H$ be the induced partition of $S_H$ such that $x\vert_{P_H} = x^{\mathcal{R}_H}$. 
Since $x_{ij} = 0$, there exist $R_1, R_2\in \mathcal{R}_H$ such that $i \in R_1$, $j \in R_2$. 
We have that
\begin{equation}
x'_{pq} = \begin{cases}
1 & \textnormal{if $pq \in P_H$}\\
0 & \textnormal{if $pq \in \delta(S_H)$} \\
x_{pq} & \textnormal{otherwise}
\end{cases}.
\end{equation}
Therefore, it follows that
\small{
\begin{equation}\label{eq:subgraph-plugin-map-1} 
\phi_c(x') - \phi_c(x) = \sum_{pq\in P_H} c_{pq}(1- x_{pq}) + \sum_{pqr\in T_H}c_{pqr}(1-x_{pq}x_{pr}x_{qr}) - \sum_{pq\in \delta(S_H)}c_{pq}x_{pq} - \sum_{pqr\in T_{\delta(S_H)}}c_{pqr}x_{pq}x_{pr}x_{qr} .
\end{equation}
}
\normalsize
In order to find an upper bound for the sums over $P_H$ and $T_H$, we show that there exists a subset $R \subset S_H$ with $i \in R$ and $j \in S_H \setminus R$ such that 
\begin{equation} \label{eq:subgraph-simplification-1}
\sum_{pq\in P_H}c_{pq}\left(1-x_{pq}\right) + \sum_{pqr\in T_H}c_{pqr}(1-x_{pq}x_{pr}x_{qr}) \leq \sum_{pq\in \delta(R, S_H \setminus R)}c_{pq} + \sum_{pqr\in T_{\delta(R, V_H \setminus R)}\cap T_H} c_{pqr}.
\end{equation}

\noindent	
For the sake of contradiction, we assume that there is no such $R \subset S_H$.
For any $\mathcal{R}'\subset \mathcal{R}_H$ let $R_{\mathcal{R}'} = \bigcup_{P'\in \mathcal{R}'}P'$. 
Furthermore we define $t \colon\tbinom{\mathcal{R}_H}{2}\cup \tbinom{\mathcal{R}_H}{3} \to \mathbb{R}$ and $p \colon \tbinom{\mathcal{R}_H}{2} \to \mathbb{R}$ as 
\begin{align}
t_{RR'R''} &= \sum_{pqr\in T_{RR'R''}}c_{pqr}, \qquad\quad\;\;\, \forall RR'R''\in \tbinom{\mathcal{R}_H}{3}\\ 
t_{RR'} &= \sum_{pqr\in T_{RRR'}\cup T_{RR'R'}}c_{pqr}, \quad \forall RR'\in \tbinom{\mathcal{R}_H}{2} 
\\
p_{RR'} &= \sum_{pq\in \delta(R, R')}c_{pq}, \qquad\qquad\quad\, \forall RR'\in \tbinom{\mathcal{R}_H}{2}.
\end{align}
Therefore, let $\mathcal{R}'\subset \mathcal{R}_H$ with $R_1\in \mathcal{R}'$ and $R_2\not\in \mathcal{R}'$. 
We observe that this implies that $i \in R_{\mathcal{R}'}$ and $j \notin R_{\mathcal{R}'}$, since $\mathcal{R}_H$ is a partition of $H$. 
It holds that 
\begin{equation}\label{eq:subgraph-simplification-2}
\sum_{pq \in P_H}c_{pq}\left(1-x_{pq}\right) + \sum_{pqr\in T_H}c_{pqr}(1-x_{pq}x_{pr}x_{qr}) > \sum_{pq \in \delta(R_{\mathcal{R}'}, S_H\setminus R_{\mathcal{R}'})}c_{pq} + \sum_{pqr\in T_{\delta(R_{\mathcal{R}'}, S_H\setminus R_{\mathcal{R}'})}\cap T_H} c_{pqr}.
\end{equation}
We evaluate the terms in \eqref{eq:subgraph-simplification-2} one-by-one, and express them as sums over elements in $\mathcal{R}'$ and $\mathcal{R}_H \setminus \mathcal{R}'$.
First, we observe that for any $pq \in P_H$ we have that $x_{pq} = 0$ if and only if there exist $RR'\in \tbinom{\mathcal{R}_H}{2}$ such that $pq\in \delta(R, R')$.
Therefore,
\begin{equation}
\sum_{pq \in P_H}c_{pq}\left(1-x_{pq}\right) = \sum_{RR'\in \tbinom{\mathcal{R}_H}{2}}p_{RR'},
\end{equation}
whereas
\small{
\begin{equation}
\sum_{pq \in \delta(R_{\mathcal{R}'}, S_H \setminus R_{\mathcal{R}'})}c_{pq} = \sum_{R \in \mathcal{R}'} \sum_{R'\in \mathcal{R}_H \setminus \mathcal{R}'} p_{RR'}.
\end{equation}
}
\normalsize
For the first sum we use the decomposition 
\begin{equation} \label{eq: decomposition1}
\tbinom{\mathcal{R}_H}{2} = \tbinom{\mathcal{R}'}{2}\cup \left\{RR' \mid R \in \mathcal{R}'\land R'\in \mathcal{R}_H\setminus \mathcal{R}'\right\} \cup \tbinom{\mathcal{R}_H\setminus \mathcal{R}'}{2}, 
\end{equation}
where the subsets are mutually disjoint.
Consequently, it holds that
\begin{equation}\label{eq:subgraph-edge-difference}
\sum_{pq\in P_H}c_{pq}\left(1-x_{pq}\right) - \sum_{pq\in \delta(W_{\mathcal{R}'}, V_H\setminus W_{\mathcal{R}'})}c_{pq} = \sum_{RR'\in \tbinom{\mathcal{R}'}{2}}p_{RR'} + \sum_{RR'\in \tbinom{\mathcal{R}_H\setminus \mathcal{R}'}{2}}p_{RR'}.
\end{equation}
Second, for any $pqr\in T_H$ we have that $x_{pq}x_{pr}x_{qr} = 0$ if and only if there exist $RR'\in \tbinom{\mathcal{R}_H}{2}$ such that $pqr\in T_{RRR'}\cup T_{RR'R'}$ or there exist $RR'R''\in \tbinom{\mathcal{R}_H}{3}$ such that $pqr\in T_{RR'R''}$.
Therefore,
\begin{equation}
\sum_{pqr\in T_H}c_{pqr} \left(1 - x_{pq}x_{pr}x_{qr}\right) = \sum_{RR'R''\in \tbinom{\mathcal{R}_H}{3}}t_{RR'R''} + \sum_{RR'\in \tbinom{\mathcal{R}_H}{2}}t_{RR'} ,
\end{equation}
whereas
\begin{align}
\sum_{pqr\in T_{\delta(R_{\mathcal{R}'}, S_H \setminus R_{\mathcal{R}'})}\cap T_H} c_{pqr} & = \sum_{RR'\in \tbinom{\mathcal{R}'}{2}}\sum_{R''\in \mathcal{R}_H \setminus \mathcal{R}'}t_{RR'R''} + \sum_{R\in \mathcal{R}'}\sum_{R'R''\in \tbinom{\mathcal{R}_H \setminus \mathcal{R}'}{2}}t_{RR'R''} \\
& + \sum_{R\in \mathcal{R}'}\sum_{R'\in \mathcal{R}_H \setminus \mathcal{R}'}t_{RR'}.
\end{align}
For the first sum we use the decomposition
\small{
\begin{equation} \label{eq: decomposition2}
\tbinom{\mathcal{R}_H}{3} =\tbinom{\mathcal{R}'}{3}\cup \left\{RR'R''\mid RR'\in \tbinom{\mathcal{R'}}{2}\land R''\in \mathcal{R}_H \setminus \mathcal{R}'\right\} 
\cup \left\{RR'R''\mid R\in \mathcal{R}'\land R'R''\in \tbinom{\mathcal{R}_H\setminus \mathcal{R}'}{2}\right\} \cup \tbinom{\mathcal{R}_H \setminus \mathcal{R}'}{3},
\end{equation}
}
\normalsize
where again the subsets are mutually disjoint.	
By \eqref{eq: decomposition1} and \eqref{eq: decomposition2}
, it follows that
\begin{align}
\sum_{pqr\in T_H}c_{pqr} \left(1 - x_{pq}x_{pr}x_{qr}\right) - \sum_{pqr\in T_{\delta(R_{\mathcal{R}'}, S_H \setminus R_{\mathcal{R}'})}\cap T_H}c_{pqr} & = \sum_{RR'R''\in \tbinom{\mathcal{R'}}{3}}t_{RR'R''} + \sum_{RR'R''\in \tbinom{\mathcal{R}_H\setminus \mathcal{R'}}{3}}t_{RR'R''}
\\
&  + \sum_{RR'\in \tbinom{\mathcal{R'}}{2}}t_{RR'} + \sum_{RR'\in \tbinom{\mathcal{R}_H\setminus \mathcal{R'}}{2}}t_{RR'}. \label{eq:subgraph-triplet-difference}
\end{align}

\noindent
Combining \eqref{eq:subgraph-simplification-2}, \eqref{eq:subgraph-edge-difference} and \eqref{eq:subgraph-triplet-difference} 
yields
\begin{align}
0 < & \sum_{pq\in P_H} c_{pq}(1-x_{pq}) + \sum_{pqr\in T_H} c_{pqr}(1-x_{pq}x_{pr}x_{qr}) - \sum_{pq\in \delta(R_{\mathcal{R}'}, S_H \setminus R_{\mathcal{R}'})}c_{pq} -\sum_{pqr\in T_{\delta(R_{\mathcal{R}'}, S_H \setminus R_{\mathcal{R}'})}\cap T_H}c_{pqr}\\
& = \sum_{RR'\in \tbinom{\mathcal{R'}}{2}}p_{RR'} + \sum_{RR'\in \tbinom{\mathcal{R}_H\setminus \mathcal{R'}}{2}}p_{RR'} + \sum_{RR'R''\in \tbinom{\mathcal{R'}}{3}}t_{RR'R''} + \sum_{RR'R''\in \tbinom{\mathcal{R}_H\setminus \mathcal{R'}}{3}}t_{RR'R''}\\
& + \sum_{RR'\in \tbinom{\mathcal{R'}}{2}}t_{RR'} + \sum_{RR'\in \tbinom{\mathcal{R}_H\setminus \mathcal{R'}}{2}}t_{RR'} =: S_{\mathcal{R}'}.
\end{align}
Let $k = \vert \mathcal{R}_H\vert$, and $S_{\mathcal{R}'}$ the right-hand side of the last inequality. 
Recall that $R_1, R_2 \in \mathcal{R}_H$, $R_1 \in \mathcal R'$, and $R_2 \notin \mathcal R'$. 
As $S_{\mathcal{R}'} > 0$, it follows that at least one of the sums in its definition must not be vacuous. 
Moreover, since its sums are indexed by pairs or triplets of subsets all belonging either to $\mathcal R'$ or to $\mathcal R_H \setminus \mathcal R'$, we observe that there must exist at least another subset
of elements in $\mathcal R_H$ different from $R_1$ and $R_2$. 
Hence, $k \geq 3$. 
We calculate 
\begin{equation}
\mathcal S = \sum_{\substack{\mathcal{R'}\subseteq \mathcal{R}_H\colon \\R_1\in \mathcal{R}', R_2\not\in \mathcal{R}'}}S_{\mathcal{R}'}.
\end{equation}
We need this in order to contradict $\max_{x\in \cp_{S_H}}\phi_{c'}(x) = 0$. 
For any $RR'\in \tbinom{\mathcal{R}_H}{2}\setminus \{R_1R_2\}$, there are exactly $2^{k-3}$ 
subsets $\mathcal{R}'\subseteq \mathcal{R}_H$ such that $p_{RR'}$ or $t_{RR'}$ occurs in $S_{\mathcal{R}'}$ and $R_1\in \mathcal{R}', R_2 \not\in \mathcal{R}'$. 
There is no $\mathcal{R}'\subseteq \mathcal{R}_H$ such that $p_{R_1R_2}$ or $t_{R_1R_2}$ occurs in $S_{\mathcal{R}'}$ with $R_1\in \mathcal{R}', R_2\not\in \mathcal{R}'$. 
For any $RR'R''\in \tbinom{\mathcal{R}_H}{3}\setminus\left\{R_1R_2R\mid R\in \mathcal{R}_H\setminus \{R_1, R_2\}\right\}$, there are exactly $\lfloor 2^{k-4} \rfloor$ 
subsets $\mathcal{R}'\subseteq \mathcal{R}_H$ such that $t_{RR'R''}$ occurs in $S_{\mathcal{R}'}$ and $R_1\in \mathcal{R}', R_2\not\in \mathcal{R}'$. 
There is no $\mathcal{R}'\subseteq \mathcal{R}_H$ such that $t_{R_1R_2R}$ occurs in $S_{\mathcal{R}'}$ for any $R\in \mathcal{R}_H\setminus \{R_1, R_2\}$ for which $R_1\in \mathcal{R}', R_2\not\in \mathcal{R}'$.
Therefore, 
\begin{align}
0 < \mathcal S & = 2^{k-3}\sum_{RR'\in \tbinom{\mathcal{R}_H}{2}}p_{RR'} - 2^{k-3}p_{R_1R_2} + \lfloor 2^{k-4} \rfloor \sum_{RR'R''\in \tbinom{\mathcal{R}_H}{3}}t_{RR'R''} - \lfloor 2^{k-4} \rfloor \sum_{\substack{R\in \mathcal{R}_H \\ R \not\in \{R_1, R_2\}}}t_{R_1R_2R}\\
& + 2^{k-3}\sum_{RR'\in \tbinom{\mathcal{R}_H}{2}}t_{RR'} - 2^{k-3}t_{R_1R_2} \\
& = 2^{k-3}\sum_{RR'\in \tbinom{\mathcal{R''}}{2}}p_{RR'}+ \lfloor 2^{k-4} \rfloor \sum_{RR'R''\in \tbinom{\mathcal{R''}}{3}}t_{RR'R''} + 2^{k-3}\sum_{RR'\in \tbinom{\mathcal{R''}}{2}}t_{RR'} = 2^{k-3}\phi_{c'}(x^{\mathcal{R''}}),
\end{align}
where $\mathcal{R}'' = \left(\mathcal{R}_H\setminus \{R_1, R_2\}\right)\cup \{R_1 \cup  R_2\}$ is the partition obtained by merging $R_1$ and $R_2$. 
The last equality follows from Lemma~\ref{lemma:subgraph-helping-lemma-1}. 
That contradicts $\max_{x\in \cp_{S_H}}\phi_{c'}(x) = 0$. 
Therefore, this implies that there exists a subset $R \subset S_H$ with $i \in R$ and $j \in S_H \setminus R$ such that inequality \eqref{eq:subgraph-simplification-1} is fulfilled.
	
Let $R \subset S_H$ be a subset such that \eqref{eq:subgraph-simplification-1} holds. Therefore, we have that
\begin{align}
\phi_c(x') - \phi_c(x) & \overset{\eqref{eq:subgraph-plugin-map-1}}{=} \sum_{pq\in P_H} c_{pq}(1- x_{pq}) + \sum_{pqr\in T_H}c_{pqr}(1-x_{pq}x_{pr}x_{qr}) - \sum_{pq\in \delta(S_H)}c_{pq}x_{pq} \\
&  - \sum_{pqr\in T_{\delta(S_H)}}c_{pqr}x_{pq}x_{pr}x_{qr} \\
& \overset{\eqref{eq:subgraph-simplification-1}}{\leq} \sum_{pq\in \delta(R, S_H\setminus R)}c_{pq} + \sum_{pqr\in T_{\delta(R, S_H\setminus R)}\cap T_H} c_{pqr} - \sum_{pq\in \delta(S_H)\cap P^-}c_{pq} \\
& - \sum_{pqr\in T_{\delta(S_H)}\cap T^-}c_{pqr} \overset{\eqref{eq:assumption-subgraph-criterion-uv-cuts}}{\leq} \;0.
\end{align}
Consequently, the map $p$ is improving. 
By applying Corollary~\ref{lemma:persistency-variable} we conclude the proof.
\end{delayedproof}
}

\ifthenelse{\boolean{proofs}}{}{
\begin{delayedproof}{proposition:subset-join-proposition}
We define $\sigma \colon \cp_S \to \cp_S$ such that 
\begin{equation}
\sigma(x) := \begin{cases}
x & \textnormal{if $x_{ij} = 1$, $\forall ij \in \tbinom R2$}\\
\left(\sigma_R \circ \sigma_{\delta(R)}\right)(x) & \textnormal{otherwise}
\end{cases}.
\end{equation}
Let $x' = \sigma(x)$, for every $x \in \cp_S$. 
First, 
it holds that $x'_{ij} = 1$, for every $ij \in \tbinom R2$. 
Second, we show that $\sigma$ is an improving map. 
Let $x \in \cp_S$ such that $x_{ij} = 1$, for all $ij \in \tbinom R2$.
In this case, we have that $\phi_c(x') = \phi_c(x)$ by definition of $x'$. 
Now, let us consider the complementary case, i.e. let $x \in \cp_S$ such that there exists $ij \in \tbinom R2$ for which $x_{ij} = 0$. 
Then,
\begin{equation}
x'_{pq} = \begin{cases}
1 & \textnormal{if $pq \in \tbinom R2$}\\
0 & \textnormal{if $pq \in \delta(R)$}\\
x_{pq} & \textnormal{otherwise}
\end{cases}.
\end{equation} 
Therefore, it follows that
\begin{align}
\phi_c(x') - \phi_c(x) & = \sum_{pq\in \tbinom R2} c_{pq}(1-x_{pq}) - \sum_{pq\in \delta(R)} c_{pq}x_{pq} + \sum_{pqr \in \tbinom R3} c_{pqr}(1-x_{pq}x_{pr}x_{qr})  \\
& - \sum_{pqr\in T_{\delta(R)}} c_{pqr}x_{pq}x_{pr}x_{qr}   \\
& \leq \max_{\substack{x\in \cp_S \\ x_{ij} = 0}} \Bigl\{ \sum_{pqr\in \tbinom R2} c_{pqr}(1-x_{pq}x_{pr}x_{qr}) + \sum_{pq\in \tbinom R2} c_{pq}(1-x_{pq}) \Bigr\} \\
& - \min_{\substack{x\in \cp_S \\ x_{ij} = 0}}\Bigl\{ \sum_{pqr\in T_{\delta(R)}} c_{pqr}x_{pq}x_{pr}x_{qr} 
 + 
\sum_{pq\in \delta(R)}c_{pq}x_{pq} \Bigr\} \overset{\eqref{eq:subset-join-inequality}}{\leq} \; 0.
\end{align}
This concludes the proof.
\end{delayedproof}
}

\ifthenelse{\boolean{proofs}}{}{
\begin{delayedproof}{prop: algo works}
We start by observing that \Cref{alg:region-growing-separation} always terminates.
If it returns a nontrivial partition $\mathcal{R}$, then $\mathcal{R}$ contains a subset $R$ that satisfies \eqref{eq:edge-cut-condition-1}--\eqref{eq:edge-cut-condition-2} by construction.
Therefore, let us assume that the output of \Cref{alg:region-growing-separation} is the trivial partition $\mathcal{R} = \{ S \}$.
If indeed there exists no nontrivial subset of $S$ for which \eqref{eq:edge-cut-condition-1}--\eqref{eq:edge-cut-condition-2} hold, then \Cref{alg:region-growing-separation} is returning the correct output.
Next, we consider the case in which there exists a nontrivial subset of $S$ that satisfies the assumptions of \Cref{lemma:persistency-subset-separation}, but \Cref{alg:region-growing-separation} still returns the trivial partition.
We prove that this cannot happen.
Let $R \subseteq S$ be a nontrivial subset of $S$ for which \eqref{eq:edge-cut-condition-1}--\eqref{eq:edge-cut-condition-2} are satisfied.
Note that such a set must exist by the assumptions of this case.
Moreover we have that both $R$ and $S \setminus R$ 
are not empty.
Two cases can arise at this point: \Cref{alg:region-growing-separation} starts either from an element of $R$ or from an item of $S \setminus R$.
Let \Cref{alg:region-growing-separation} start sampling from $R$.
The fact that $\mathcal{R} = \{ S \}$ implies that $\exists pq \in \delta(R)$ such that $c_{pq} < 0$ or $\exists pqr \in T_{\delta(R)}$ such that $c_{pqr} < 0$ by definition of \Cref{alg:region-growing-separation}.
However, this is in contradiction with the assumption that $R$ satisfies \eqref{eq:edge-cut-condition-1}--\eqref{eq:edge-cut-condition-2}.
Since the second scenario is symmetrical, we again reach a contradiction 
by applying an analogous reasoning.
Therefore, we have showed that if there exists a nontrivial subset of $S$ that fulfills \eqref{eq:edge-cut-condition-1}--\eqref{eq:edge-cut-condition-2}, then \Cref{alg:region-growing-separation} finds such a subset.
\end{delayedproof}
}

\ifthenelse{\boolean{proofs}}{}{
\begin{delayedproof}{eq:cubic-st-cut-reduction-min-st-cut}
Let $R \subseteq S$. 
Observe that
\begin{align} \label{eq:identity-sum-all-TdeltaR-triplets}
\sum_{pqr\in T_{\delta(R)}}c_{pqr} & = \sum_{pq\in \tbinom R2}\sum_{r\in S \setminus R}c_{pqr} + \sum_{pq\in \tbinom{S\setminus R}{2}}\sum_{r\in R} c_{pqr} \\
& = \frac 12 \sum_{p\in R}\sum_{q\in R \setminus \{p\}}\sum_{r\in S\setminus R}c_{pqr} + \frac 12 \sum_{p\in S \setminus R}\sum_{q\in S \setminus \left(R \cup \{p\}\right)}\sum_{r\in R}c_{pqr} \\
& = \frac 12 \sum_{p\in R}\sum_{q\in S \setminus R}\left(\sum_{r\in R \setminus \{p\}}c_{pqr} + \sum_{r\in S \setminus \left(R \cup \{q\}\right)}c_{pqr}\right) = \frac 12 \sum_{pq\in \delta(R)}\sum_{r\in S \setminus \{p, q\}}c_{pqr}.
\end{align}%
\end{delayedproof}
}

\ifthenelse{\boolean{proofs}}{}{
\begin{delayedproof}{lemma:qpbo-translation}
Let $R \subseteq S$ such that $i\in R$ and $j\not \in R$, $\forall j \in S_0$. 
We define $y\in \{0, 1\}^S$ such that $y = \mathbbm{1}_R$. 
Then, we have that $y_i = 1$ and $y_j = 0$, $\forall j \in S_0$. 
Moreover, it follows that
\begin{align}
\sum_{pq\in \delta(R)} c_{pq} & = \sum_{pq\in \tbinom S2}c_{pq}\left(y_p (1-y_q) + y_q (1-y_p)\right) = \sum_{pq\in \tbinom S2}c_{pq}\left(y_p + y_q - 2y_p y_q\right) \\
& = -2\sum_{pq\in \tbinom S2}c_{pq}y_p y_q + \sum_{\substack{p,q \in S \\ p \neq q}} c_{pq}y_p  \\ 
& = -2 \sum_{pq\in \tbinom{S'}{2}}c_{pq}y_p y_q - 2\sum_{p\in S'}c_{pi} y_p + \sum_{p\in S'}\sum_{q\in S \setminus \{p\}}c_{pq}y_p + \sum_{q\in S \setminus \{i\}}c_{qi}\\
& = -2\sum_{pq\in \tbinom{S'}{2}}c_{pq}y_p y_q + \sum_{p\in S'}\left(-2c_{pi} + \sum_{q\in S \setminus \{p\}}c_{pq}\right)y_p + \sum_{q\in S \setminus \{i\}}c_{qi} \\
& = \sum_{pq\in \tbinom{S'}{2}}c'_{pq}y_p y_q + \sum_{p\in S'} c'_p y_p + c'_\emptyset.
\end{align}
This concludes the proof. 
\end{delayedproof}
}

\ifthenelse{\boolean{proofs}}{}{
\begin{delayedproof}{lemma:contraction-cost-adjustments}
Let $x\in \cp_S \vert_{x_{ij} = 1}$. 
We show that $\phi_c(x) = \phi_{c'}(\varphi_{ij}(x))$. 
Let $x' = \varphi_{ij}(x)$. 
We use the fact that $x_{pi} = x_{pj}$, $\forall p\in S \setminus \{i, j\}$, and $x_{ij} = 1$. 
It follows that
\small{
\begin{align}
\phi_{c'}(x') & = \sum_{pqr\in \tbinom{S'}{3}}c'_{pqr}x'_{pq}x'_{pr}x'_{qr} + \sum_{pq\in \tbinom S2} c'_{pq}x'_{pq} + c'_{\emptyset}\\
& = \sum_{pq\in \tbinom{S \setminus\{i, j\}}{2}}c'_{pqi}x'_{pi}x'_{qi}x'_{pq} + \sum_{pqr\in \tbinom{S \setminus\{i, j\}}{3}}c'_{pqr}x'_{pq}x'_{pr}x'_{qr} + \sum_{p\in S \setminus\{i, j\}}c'_{pi}x'_{pi}+ \sum_{pq\in \tbinom{V\setminus\{i, j\}}{2}}c'_{pq}x'_{pq} + c'_{\emptyset}\\
& = \sum_{pq\in \tbinom{S \setminus\{i, j\}}{2}}(c_{pqi} + c_{pqj})x_{pi}x_{qi}x_{pq} + \sum_{pqr\in \tbinom{S \setminus\{i, j\}}{3}}c_{pqr}x_{pq}x_{pr}x_{qr} + \sum_{pq\in \tbinom{S \setminus\{i, j\}}{2}}c_{pq}x_{pq} \\
&+ \sum_{p\in S \setminus\{i, j\}}(c_{pi} + c_{pj} + c_{pij})x_{pi}+ c'_{\emptyset}\\
&= \sum_{pq\in \tbinom{S \setminus\{i, j\}}{2}}c_{pqi}x_{pi}x_{qi}x_{pq} + \sum_{pq\in \tbinom{S \setminus\{i, j\}}{2}}c_{pqj}x_{pj}x_{qj}x_{pq} + \sum_{pqr\in \tbinom{S \setminus\{i, j\}}{3}}c_{pqr}x_{pq}x_{pr}x_{qr} \\
& + \sum_{p\in V}c_{pij}x_{ij}x_{pi}x_{pj} 
+ \sum_{pq\in \tbinom{S \setminus\{i, j\}}{2}}c_{pq}x_{pq} + \sum_{p\in V\setminus\{i, j\}}c_{pi}x_{pi} + \sum_{p\in S \setminus\{i, j\}}c_{pj}x_{pj} + c_{ij}x_{ij} + c_{\emptyset}  \\
&= \sum_{pqr\in \tbinom S3}c_{pqr}x_{pq}x_{pr}x_{qr} + \sum_{pq\in S}c_{pq}x_{pq} + c_{\emptyset} = \phi_c(x).
\end{align}
}
\normalsize
Therefore we have that
\begin{equation}
\min_{x\in \cp_S \vert_{x_{ij} = 1}} \phi_c(x) = \min_{x\in \cp_S \vert_{x_{ij} = 1}} \phi_{c'}(\varphi_{ij}(x)) = \min_{x\in \cp_{S'}}\phi_{c'}(x).
\end{equation}
This concludes the proof.
\end{delayedproof}
}

\subsection{Reduction of QPBO to Min-$st$-Cut}
\label{appendix:qpbo-to-min-st-cut-propositions}

\begin{lemma}
	\label{lemma:qubo-to-max-flow-analogue}
	Let $S \neq \emptyset$ and $c\in \mathbb{R}^{S \cup \tbinom S2}$. 
	We define $c'\in \mathbb{R}^{S \cup \tbinom S2}$ as
		$c'_{p} = c_{p} + \frac 12 \sum_{q\in S \setminus \{p\}} c_{pq}$, for every $p \in S$, $c'_{pq} = -\frac 12 c_{pq}$, for every $pq \in \tbinom S2$.
	Then, for any $y\in \{0, 1\}^V$ we have that
		\begin{align}
			\sum_{pq\in \tbinom S2}c_{pq} y_py_q + \sum_{p\in S} c_p y_p =  \sum_{p\in S} \sum_{q\in S\setminus \{p\}} c'_{pq}y_p (1-x_q) + \sum_{\substack{p\in S \\ c'_p > 0}} c'_p y_p -\sum_{\substack{p\in S \\ c'_p< 0}} c'_p (1-y_p) + \sum_{\substack{p\in S \\ c'_p < 0}}c'_p.
		\end{align}
\end{lemma}

\begin{proof}
Let $y\in \{0, 1\}^S$. 
We have that
\begin{align}
\sum_{pq\in \tbinom S2}c_{pq}y_p y_q + \sum_{p\in S}c_py_p & = \frac{1}{2}\sum_{p\in S}\sum_{q\in S \setminus \{p\}}c_{pq}y_p y_q + \sum_{p\in V}c_py_p \\
& = -\frac{1}{2}\sum_{p\in S} \sum_{q\in S \setminus \{p\}} c_{pq}y_p (1-y_q) + \frac{1}{2}\sum_{p\in S} \sum_{q\in S \setminus \{p\}} c_{pq}y_p  + \sum_{p\in S} c_{p}y_p \\
& = -\frac{1}{2}\sum_{p\in S} \sum_{q\in S \setminus \{p\}} c_{pq}y_p (1-y_q) + \sum_{p\in S} \left(c_p + \frac{1}{2}\sum_{q\in S \setminus \{p\}} c_{pq}\right)y_p  \\
& = \sum_{p\in S} \sum_{q\in S \setminus \{p\}} c'_{pq}y_p (1-y_q) + \sum_{p\in V} c'_p y_p = \sum_{p\in S} \sum_{q\in S \setminus \{p\}} c'_{pq}y_p (1-y_q) \\ 
& + \sum_{\substack{p\in V\\ c'_p > 0}} c'_p y_p -\sum_{\substack{p\in S \\ c'_p< 0}} c'_p (1-y_p) + \sum_{\substack{p\in S \\ c'_p < 0}}c'_p.
\end{align}
We therefore reach the thesis. 
\end{proof}

We reduce this problem to solving an instance of min-$st$-cut.
If $c_{pq} \leq 0$, $pq\in \tbinom S2$, and therefore $c'_{pq} \geq 0$, $\forall pq\in \tbinom S2$, in Lemma~\ref{lemma:qubo-to-max-flow-analogue} the resulting instance can be solved efficiently.

\begin{proposition}\label{prop:last-prop}
	Let $S \neq \emptyset$ and $c\in \mathbb{R}^{S \cup \tbinom S2}$. 
	We define $\phi_c \colon \{0, 1\}^S \to \mathbb{R}$ such that for all $y\in \{0, 1\}^S$ it holds that
		\begin{align}
			\phi_c(y) = &\sum_{p\in S} \sum_{q\in S \setminus \{p\}} c_{pq}y_p (1-y_q) 
			+ \sum_{\substack{p\in S\\ c_p > 0}} c_p y_p - \sum_{\substack{p\in S \\ c_p< 0}} c_p(1-y_p).
		\end{align}
	Furthermore, we define $S' = S \cup \{s, t\}$, $P'\subset S'\times S'$ such that 
	\begin{align}
		(s, p)\in P' \Leftrightarrow c_p < 0&,\quad \forall p \in S, \\
		(p, t)\in P' \Leftrightarrow c_p > 0&,\quad \forall p \in S, \\
		(p, q)\in P' \land (q, p)\in P'& , \quad \forall pq\in \tbinom S2,
	\end{align}
	and $c' \in \mathbb{R}^{P'}$ such that 
	\begin{align}
		c'_{(s, p)} &= -c_p, \;\,\forall (s, p) \in P', \\
		c'_{(p, t)} &= c_p, \quad\,\, \forall (p, t)\in P',\\
		c'_{(p, q)} = c'_{(q, p)}&= c_{pq}, \quad \forall pq\in \tbinom S2 .
	\end{align}
	Moreover, we define the function $\varphi_{c'}: \{0, 1\}^{S'} \to \mathbb{R}$ such that for all $y\in \{0, 1\}^{S'}$ it holds that 
	\begin{equation}
		\varphi_{c'}(y) = \sum_{(p,q)\in P'}c'_{(p, q)}y_p(1- y_q).
	\end{equation}
	Then we have that
	\begin{equation}
		\min_{x\in \{0, 1\}^S}\phi_c(x) = \min_{\substack{y\in \{0, 1\}^{S'} \\ y_s = 1 \\ y_t = 0}} \varphi_{c'}(y).
	\end{equation}
\end{proposition}

\begin{proof}
First, the map $\chi\colon \{0, 1\}^S \to \{y \in \{0, 1\}^{S'}\mid y_s = 1 \,\land \, y_t = 0\}$ such that $\chi(y)_s = 1$, $\chi(y)_t = 0$ and $\chi(y)_p = y_p$ for any $p\in S$ is bijective. 
Second, for any $y\in \{0, 1\}^S$ it holds that
\small{
\begin{align}
\varphi_{c'}(\chi(y)) & = \sum_{(p, q)\in P'}c'_{pq}\chi(y)_p (1- \chi(y)_q) = \sum_{p\in S}\sum_{q\in S \setminus \{p\}}c'_{(p, q)}y_p(1- y_q) + \sum_{\substack{p\in S \\ c_p > 0}}c'_{(p, t)}y_p + \sum_{\substack{p\in S \\ c_p < 0}}c'_{(s, p)}(1- y_p) \\
& = \sum_{p\in S}\sum_{q\in V\setminus \{p\}}c_{pq}y_p(1- y_q) +\sum_{\substack{p\in S \\ c_p > 0}}c_py_p - \sum_{\substack{p\in S \\ c_p < 0}}c_p(1- y_p) = \phi_c(y)
\end{align}
}
\normalsize
This concludes the proof.
\end{proof}

\end{document}